\documentclass[%
 reprint,
 superscriptaddress,
 amsmath,amssymb,
 aps,
 floatfix,
]{revtex4-2}

\usepackage{dcolumn}
\usepackage{bm}
\usepackage{hyperref}
\usepackage[dvipsnames, table]{xcolor}
\usepackage{mathtools}
\usepackage[caption=false]{subfig}
\usepackage{amsmath,amsthm,amssymb}
\usepackage{graphicx}
\usepackage{array}
\usepackage{braket}
\usepackage{bbm}
\usepackage{tcolorbox}
\usepackage{tikz}

\newtheorem{proposition}{Proposition}[section]
\theoremstyle{definition}

\usetikzlibrary{fit}
\tikzset{%
  highlight/.style={rectangle,rounded corners,fill=blue!30,draw,
    fill opacity=0.30,thin,inner sep=-2.0pt}
}

\tikzset{%
  highlightt/.style={rectangle,rounded corners,fill=blue!30,draw,
    fill opacity=0.30,thin,inner sep=2.0pt}
}

\newcommand{\tikzmark}[2]{\tikz[overlay,remember picture,
  baseline=(#1.base)] \node (#1) {#2};}
\newcommand{\Highlight}[2][submatrix]{%
    \tikz[overlay,remember picture]{
    \node[highlight,fit=(left#2.north west) (right#2.south east)] (#1) {};}
}

\newcommand{\Highlightt}[2][submatrix]{%
    \tikz[overlay,remember picture]{
    \node[highlightt,fit=(left#2.north west) (right#2.south east)] (#1) {};}
}

\newcommand{\gl}{\bm{\gamma}}

\newcommand{\nl}{\bm{n}}
\newcommand{\zl}{\bm{z}}
\newcommand{\ml}{\bm{m}}

\DeclareMathOperator{\tr}{tr}
\DeclareMathOperator{\lhaf}{lhaf}

\DeclareMathOperator{\filldiag}{filldiag}
\DeclareMathOperator{\diag}{diag}

\begin{document}
    
    \title{Simulating lossy and partially distinguishable
quantum optical circuits: theory, algorithms and applications to experiment validation and state preparation}
    \author{Jacob F.F. Bulmer}
    \thanks{These authors contributed equally to this work.}
	\affiliation{Quantum Engineering Technology Labs, H. H. Wills Physics Laboratory, Department of Electrical and Electronic Engineering, University of Bristol, Bristol, BS8 1FD, UK}
 
    \author{Javier Mart\'inez-Cifuentes}
    \thanks{These authors contributed equally to this work.}
	\affiliation{D\'epartement de g\'enie physique, \'Ecole Polytechnique de Montr\'eal, Montr\'eal, QC, H3T 1J4, Canada}

    \author{Bryn A. Bell}
	\affiliation{Department of Physics, Imperial College London, Prince Consort Rd, London SW7 2AZ, UK}
 
    \author{Nicol\'as Quesada}
        \email{nicolas.quesada@polymtl.ca}

	\affiliation{D\'epartement de g\'enie physique, \'Ecole Polytechnique de Montr\'eal, Montr\'eal, QC, H3T 1J4, Canada}

    \begin{abstract}
        To understand quantum optics experiments, we must perform calculations that consider the principal sources of noise, such as losses, spectral impurity and partial distinguishability. In both discrete and continuous variable systems, these can be modeled as mixed Gaussian states over multiple modes. The modes are not all resolved by photon-number measurements and so require calculations on \textit{coarse-grained} photon-number distributions. Existing methods can lead to a combinatorial explosion in the time complexity, making this task unfeasible for even moderate sized experiments of interest. In this work, we prove that the computation of this type of distribution can be done in exponential time, providing a combinatorial speedup. We develop numerical techniques that allow us to determine coarse-grained photon number distributions of Gaussian states, as well as density matrix elements of heralded non-Gaussian states prepared in the presence of spectral impurity and partial distinguishability. These results offer significant speed-up and accuracy improvements to validation tests of both Fock and Gaussian boson samplers that rely on binned probability distributions. Moreover, they pave the way to a more efficient simulation of realistic photonic circuits, unlocking the ability to perform exact calculations at scales which were previously out of reach. In addition to this, our results, including loop Hafnian master theorems, may be of interest to the fields of combinatorics and graph theory.
	\end{abstract}
    
    \date{\today}
    
    \maketitle
    
    \section{\label{sec:introduction}Introduction}
    
        \begin{figure*}[t!]
            \centering
            \subfloat[]{\label{fig:coarse_grained_probs}\includegraphics[width=0.49\textwidth]{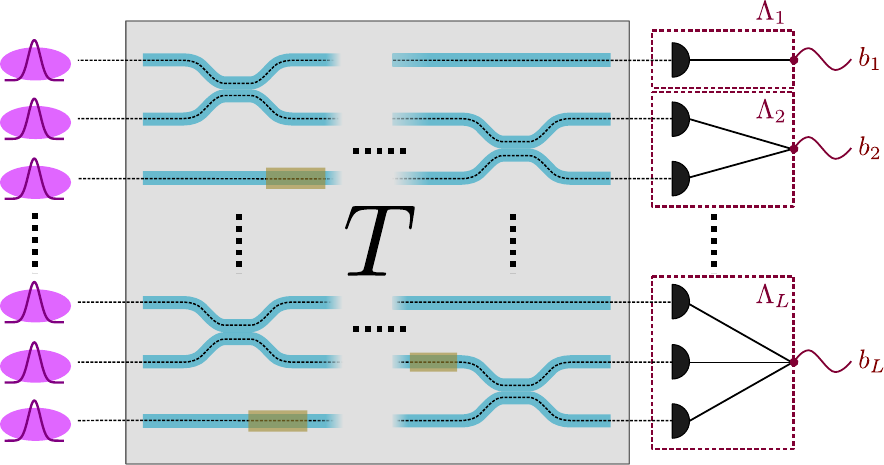}}\hfill
            \subfloat[]{\label{fig:internal_modes_heralding}\includegraphics[width=0.49\textwidth]{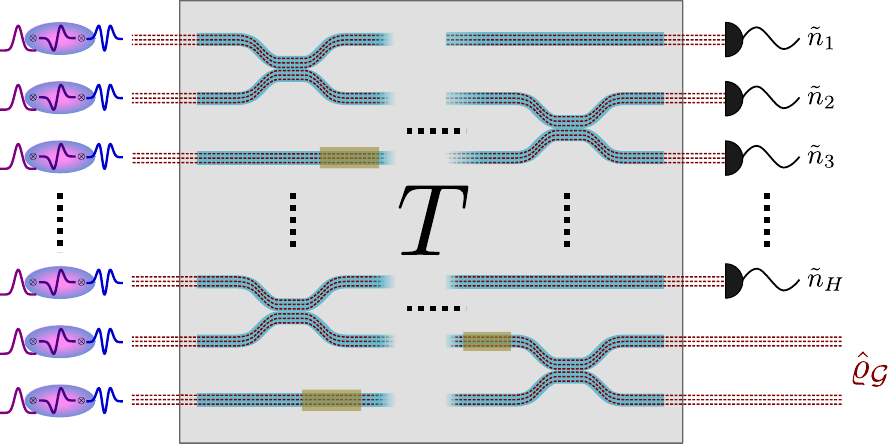}}\\
            \subfloat[]{\label{fig:scattershot_heralding}\includegraphics[width=0.8\textwidth]{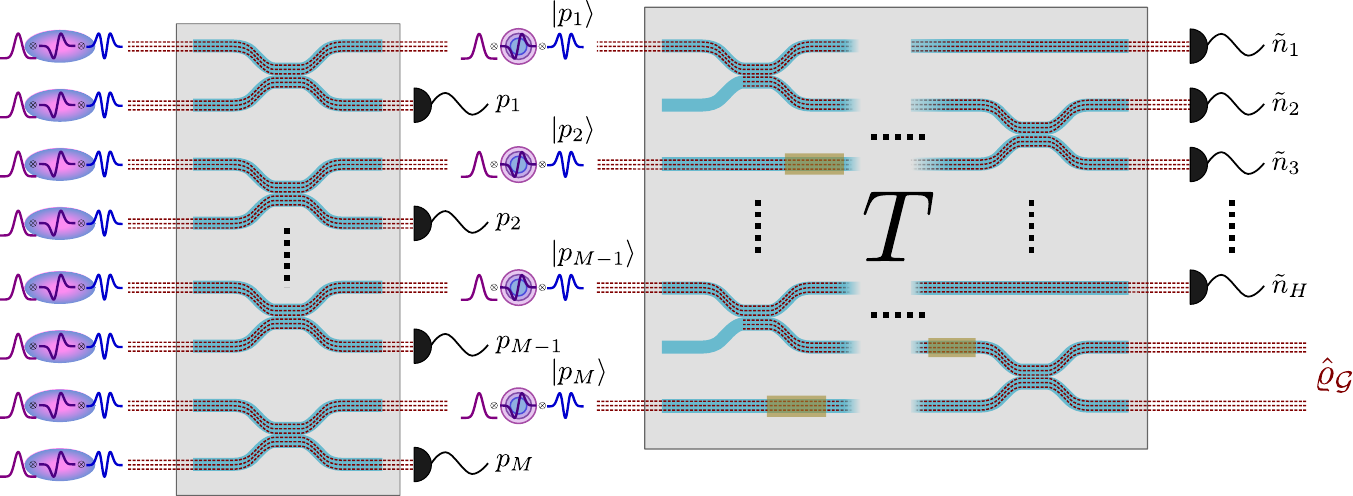}}\hfill
            \caption{(a) Illustration of the problem of computing general coarse-grained probabilities in the context of Gaussian Boson Sampling (GBS). A set of input Gaussian states (purple ellipses) are sent into a linear interferometer mathematically described by a transfer matrix $\bm{T}$. The output light is measured using photon-number-resolving detectors (black semicircles). The detectors are grouped according to the sets $\{\Lambda_1,\Lambda_2,\dots,\Lambda_L\}$, leading to a coarse-grained detection pattern $(b_1,b_2,\dots,b_L)$. We are interested in computing the probability of this photon detection event. (b) Illustration of the problem of state generation (or heralding) in the presence of internal modes. A set of spectrally mixed Gaussian states (multicolor ellipses) are sent into a linear interferometer represented by a transfer matrix $\bm{T}$. The output light of the interferometer is \textit{partially measured} in the first $H$ output ports of the interferometer using photon-number-resolving detectors. Upon the measurement of the heralding detection pattern $\tilde{\bm{n}}=(\tilde{n}_1,\dots,\tilde{n}_H)$, the state $\hat{\varrho}_{\mathcal{G}}$ is generated in the remaining output ports of the interferometer. We are interested in studying how the presence of internal modes (in this case, also spectral modes represented by dashed red lines) affects the computation of the heralded state $\hat{\varrho}_{\mathcal{G}}$. By measuring all the output ports of the interferometer, one can determine photon-number probabilities in the presence of internal modes. (c) Illustration of the problem of state generation (or heralding) in the presence of internal modes, for input photon-number (Fock) states (multicolor circles). The input Fock states are generated via the partial measurement of a set of two-mode squeezed states (which are obtained by interfering two single-mode squeezed states into a balanced beam-splitter). This approach requires scaling the probability of successfully generating $\hat{\varrho}_{\mathcal{G}}$ by $\Pr(\bm{p})=\Pr(p_1)\dots\Pr(p_M)$, i.e., the probability of heralding the desired input multimode Fock state.}
            \label{fig:pnr_statistics_blocked}
        \end{figure*}
        
        Gaussian states are of fundamental and practical interest in both quantum information and in quantum optics~\cite{serafini2017quantum,ferraro2005gaussian,weedbrook2012gaussian,asavanant2022optical}. These states have the property of having quasiprobability distributions and characteristic functions that are Gaussian functions~\cite{barnett2002methods,gerry2005introductory}. This property implies that they can be succinctly described in terms of a covariance matrix and a vector of means. Despite their simple description, problems related to their photon-number statistics possess a rich mathematical structure related to perfect matchings and Hamiltonian walks of undirected complex-weighted graphs~\cite{dodonov1994multidimensional,krenn2017quantum,gu2019quantum,quesada2019franck,bulmer2022boundary,cardin2022photon}. On the experimental side, measuring a subset of the modes of an entangled Gaussian state can be used to herald the generation of interesting and resourceful non-Gaussian states~\cite{roeland2021mode,thomas2021general,quesada2015very,tiedau2019scalability,sonoyama2024generation,konno2024logical,endo2024optically,asavanant2024multipartite,kawasaki2024high,wiesner2024influence,ra2020non}. Moreover, when all the modes are measured, a sampling problem that is computationally intractable by classical computers is obtained, the so-called Gaussian Boson Sampling (GBS) problem~\cite{kruse2019detailed,hamilton2017gaussian,grier2022complexity,deshpande2022quantum,go2024exploring}. 
	
    	In the analysis of heralding protocols and Gaussian Boson Samplers, the leading cause of imperfections is optical losses. Conveniently, while loss is a decoherence process that typically turns pure states into mixed ones, it has the property of being a Gaussian channel, thus it will map a pure Gaussian state into a (typically) mixed one. This implies that this source of decoherence can be easily incorporated into simulation frameworks. 
        
        The next leading source of imperfections is related to the spectrum of light sources. On one hand, parametric sources of squeezed light, the fundamental ingredient to build Gaussian states~\cite{braunstein2005squeezing}, often produce light over multiple spectral modes that photon-number-resolving detectors cannot tell apart~\cite{thekkadath2024gain, houde2023waveguided}. This implies that, upon heralding, there is classical ignorance about which photon from which spectral mode caused a given detector to count a certain number of photons~\cite{de2014coincidence,villalobos2023quantum}. On the other hand, on-demand single photon sources, such as quantum dots, often produce spectrally impure single photons as their emission process is affected by the solid-state environment in which they are created~\cite{nawrath2019coherence,nawrath2023bright,yu2023telecom,phillips2024purcell,dusanowski2022optical}. To understand the effects of spectral distinguishability, one can start to enumerate all the possible events that could have led to a particular detection outcome. This approach, while viable for small number of detected photons in a small number of modes, quickly becomes unfeasible for even moderate system sizes~\cite{tiedau2019scalability,quesada2015very}.
        
    	The problem of calculating coarse-grained probabilities also appears naturally when validating both GBS and Fock Boson Sampling~\cite{aaronson2011computational} (which can be interpreted as a subproblem of GBS when combined with post-selection over some subset of heralding modes~\cite{lund2014boson}). In this context, the fine-grained probability space is so large that it is impossible in any reasonable time to collect sufficiently many samples to obtain a frequentist approximation of detection probabilities. However, one may define coarse-grained detection events, obtaining much better approximations for the probabilities of binned measurement outcomes~\cite{seron2022efficient,bressanini2023gaussian}. A good example of these type of coarse-grained distributions is the total photon-number distribution (the probability of detecting a given number of photons across all available detectors), which has been previously used to partially validate the correct functioning of GBS machines~\cite{madsen2022quantum}. It has been shown that this distribution can also be used to obtain a partial characterization of the covariance matrix of a multimode Gaussian state~\cite{avagyan2023multi,avagyan2024multi}. To avoid dealing with the combinatorial explosion of time complexity related to the computation of coarse-grained probabilities, one may perform numerical estimations using Monte Carlo phase space methods. These methods can be efficient for small numbers of coarse-grained bins, but a careful statistical analysis of the obtained results is required~\cite{dellios2022simulating,dellios2022validation,drummond2022simulating,martinez2023classical,dellios2024validation}.
    	Finally, quantities analogous to coarse-grained distributions can be used to evaluate cross-entropy measures between photon-number distributions of different Gaussian states, providing an alternative for the validation of GBS experiments~\cite{martinez2024linear}.
        
    	In this work, we show that one can very generally  avoid a combinatorial explosion in the time complexity of calculating photon-number coarse-grained probabilities (as well as photon-number moments and cumulants) of Gaussian states entangled on passive optical elements. We prove that the computation time of these binned probabilities scales exponentially, not combinatorially, with the number of bins in which photons are detected.
        This reduction of complexity is also manifest in the computation of heralded, generally non-Gaussian states obtained from the partial photon-number measurements of Gaussian states with internal modes. We develop new numerical techniques that allow us to simulate the realistic preparation of these non-Gaussian states, extending the sources of noise that could previously be taken into account~\cite{quesada2019simulating}. Our results can also be applied to systems in which spectrally mixed Fock states are entangled using passive optical linear networks.

        This paper is structured as follows. In Sec.~\ref{sec:gaussian_states} we give a quick review of Gaussian states and introduce the mathematical expressions, in terms of the loop Hafnian, that describe their photon-number statistics. In Sec.~\ref{sec:master_theorem} we prove a \textit{loop Hafnian master theorem}, which will be of central importance to obtain the results of later sections. Sec.~\ref{sec:total_ph_prob} is devoted to the computation of the total photon-number distribution of Gaussian states, while in Sec.~\ref{sec:trace_formulas} we give a proof of the \textit{trace formula} for the loop Hafnian. This proof will allow us to introduce the notion of a \textit{finite-difference sieve}, which is the second most important ingredient needed to obtain the results of the following sections. In Sec.~\ref{sec:blocked_hafnian} we introduced the notion of \textit{blocked loop Hafnian}, and combine the loop Hafnian master theorem and the finite-difference sieve to obtain a formula for its computation. This formula will be used to obtain the most general expression for coarse-grained probability distributions of Gaussian states, as well as their corresponding coarse-grained photon-number moments and cumulants. Sec.~\ref{sec:off_diagonal} is devoted to applying the blocked loop Hafnian to compute density matrix elements of generally non-Gaussian states heralded from the partial measurement of a Gaussian state. To illustrate the various results of the manuscript, we present some numerical experiments in Sec.~\ref{sec:numerical}. Finally, we conclude in Sec.~\ref{sec:conclusion}.

    \section{\label{sec:gaussian_states}Review of Gaussian states}
    
        The quantum state of a system of $M$ bosonic modes is described by a density matrix $\hat{\varrho}$. Throughout this manuscript we will think of this system as a photonic system. However, all of our discussion will also be valid for any other type of bosonic system. Each mode has annihilation and creation operators $\hat{a}_i$ and $\hat{a}_i^\dagger$ that satisfy the usual bosonic canonical commutation relations
        \begin{align}
            [\hat{a}_i,\hat{a}_j] = [\hat{a}_i^\dagger, \hat{a}_j^\dagger] = 0 \text{ and } [\hat{a}_i,\hat{a}_j^\dagger] = \delta_{i,j},
            \label{eq:canonical_comm_rels}
        \end{align}
        where the commutator of two operators is defined as $[\hat{A},\hat{B}]:=\hat{A} \hat{B} - \hat{B} \hat{A}$.
        By defining a vector of annihilation and creation operators, $\bm{\hat{\zeta}} = (\hat{a}_1,\ldots,\hat{a}_M,\hat{a}_1^\dagger,\ldots,\hat{a}_M^\dagger)^{\mathrm{T}}$, these commutation relations are more compactly expressed as
        \begin{align}
            [\hat{\zeta}_i,\hat{\zeta}_j^\dagger] = \mathcal{Z}_{i,j}, \text{ with } \bm{\mathcal{Z}}=\begin{pmatrix} \mathbb{I} & 0 \\0 & -\mathbb{I} \end{pmatrix}.
            \label{eq:canonical_comm_rels_compact}
        \end{align}
        Here, $\mathbb{I}$ stands for the $M\times M$ identity matrix.
        
        Continuous variable states can also be described in terms of their $s$-ordered characteristic function
        \begin{align}
            \chi_s(\bm{w}) = \exp\left(\tfrac{s}{4}||\bm{w}||^2 \right) \mathrm{Tr}\left[ \exp\left(-\bm{w}^\dagger \bm{\mathcal{Z}} \bm{\hat{\zeta}} \right) \hat{\varrho} \right],
            \label{eq:characteristic_function_def}
        \end{align}
        where $s \in \mathbb{R}$ and $\bm{w} = (\alpha_1,\ldots,\alpha_M,\alpha_1^*,\ldots,\alpha_M^*)^{\mathrm{T}}$. In this relation, $||\bm{w}||$ is the Euclidean norm of vector $\bm{w}$. Gaussian states are those for which
        \begin{align}
            \chi_s(\bm{w}) = \exp\left(-\tfrac{1}{2}  \bm{w}^\dagger\bm{\mathcal{Z}} \bm{\Sigma}^{(s)}  \bm{\mathcal{Z}} \bm{w} + \bm{\bar{z}}^\dagger \bm{\mathcal{Z}} \bm{w} \right),
            \label{eq:gaussian_charachteristic_function}
        \end{align}
        where the vector of means, $\bm{\bar{z}}$, and covariance matrix, $\bm{\Sigma}^{(s)}$, have entries
        \begin{align}
            \bar{z}_i &= \text{Tr}(\hat{\zeta}_i\,\hat{\varrho})
            \label{eq:vector_of_means}
            \\
            \Sigma^{(s)}_{i,j} &= \tfrac12 \left[ \text{Tr}\left( \left\{ \hat{\zeta}_i \hat{\zeta}_j^\dagger + \hat{\zeta}_j^\dagger \hat{\zeta}_i \right\}\hat{\varrho}\right) -s \delta_{i,j}\right] - \bar{z}_i \bar{z}_j^*.
            \label{eq:covariance_matrix}
        \end{align}
        
        Yet another equivalent description of a Gaussian state can be obtained in terms of its adjacency matrix, $\bm{A}$, and its loop vector, $\bm{\gamma}$, defined as
        \begin{align}
            \bm{A} &= \bm{X}\left(\mathbb{I} - \bm{\Sigma}^{-1}\right),
            \label{eq:defA}
            \\
            \bm{\gamma} &=  \bm{X} \bm{\Sigma} ^{-1} \bm{\bar{z}},
            \label{eq:defgamma}
        \end{align}
        with
        \begin{align}
            \bm{X} =
            \begin{pmatrix}
                0 & \mathbb{I}  \\
                \mathbb{I} & 0
            \end{pmatrix},
            \text{ and }
            \bm{\Sigma} := \bm{\Sigma}^{(s=-1)}.
            \label{eq:x_matrix_husimi_cov}
        \end{align}
        
        To represent a valid Gaussian state, the covariance matrix must have the form
        \begin{align}
            \bm{\Sigma}^{(s)} = \frac{1-s}{2}\mathbb{I} + \begin{pmatrix} 
                \bm{\mathcal{N}}^T & \bm{\mathcal{M}} \\ 
                \bm{\mathcal{M}}^* & \bm{\mathcal{N}} 
            \end{pmatrix}
        \end{align}
        with $\bm{\mathcal{M}} = \bm{\mathcal{M}}^T$ and $\bm{\mathcal{N}}^\dagger = \bm{\mathcal{N}}$ and, moreover, needs to satisfy the uncertainty relation $\bm{\Sigma}^{(s)} + \tfrac12 \bm{\mathcal{Z}} + \tfrac{s}{2} \mathbb{I} \succcurlyeq 0 $, which in turn gives a positive semidefinite constraint on the inverse of $\bm{A}$.

        Suppose now that we are interested in computing the diagonal matrix elements of the state $\hat{\varrho}$ in the Fock basis. Each element of this basis can be written as $|\bm{n}\rangle=|n_1,\ldots,n_M\rangle$, with $n_k\in \{0,1,2,\ldots\}\;\forall k$. We can think of the string of non-negative integers $\bm{n}=(n_1,\ldots, n_M)$ as representing the outcome of a photon-number measurement (with $n_k$ representing the number of detected photons at mode $k$). The matrix element $\langle\bm{n}|\hat{\varrho}|\bm{n}\rangle$ can be interpreted as the probability of obtaining the outcome $\bm{n}$ given the state $\hat{\varrho}$. For a Gaussian state, this probability can be parametrized in terms of the matrix $\bm{A}$ and vector $\gl$ defined in Eqs.~\eqref{eq:defA} and~\eqref{eq:defgamma} as~\cite{kruse2019detailed, quesada2019simulating}
        \begin{align}
        	\langle\bm{n}|\hat{\varrho}|\bm{n}\rangle&=\Pr(\nl| \bm{A}, \gl)\nonumber\\
            &=\Pr(\bm{0}|\bm{A}, \gl) \frac{\lhaf\left[\filldiag\left({\bm{A}}_{\nl\oplus\nl}, {\gl}_{\nl \oplus \nl}\right)\right]}{\prod_{i=1}^M n_i!},
            \label{eq:gbs_probability_distribution}
        \end{align}
        where $\Pr(\bm{0}|\bm{A}, \gl)$ is the vacuum probability, computed as
        \begin{align}
        	\Pr(\bm{0} | \bm{A}, \gl) = \exp\bigl( -\tfrac12 \gl^\mathrm{T}&\left[\mathbb{I} - \bm{X} \bm{A}\right]^{-1} \bm{X} \gl \bigr)\nonumber\\ 
            &\times \sqrt{\det\left(\mathbb{I} - \bm{X}\bm{A} \right)}.
            \label{eq:vacuum_probability}
        \end{align}

        In Eq.~\eqref{eq:gbs_probability_distribution}, $\filldiag(\bm{A}, \gl)$ is an operation where we overwrite the diagonal elements of $\bm{A}$ with the elements of the vector $\gl$. We form the matrix ${\bm{A}}_{\nl\oplus\nl}$ by taking $n_k$ copies of the $k$th and $(k+M)$th rows and columns of $\bm{A}$. Similarly, ${\gl}_{\nl \oplus \nl}$  is constructed by taking $n_k$ copies of the $k$th and $(k+M)$th elements of $\gl$. Finally, for a symmetric matrix $\bm{B}$ of size $2m\times 2m$, the \textit{loop Hafnian}, $\lhaf(\bm{B})$ is defined as~\cite{bjorklund2019faster}
        \begin{align}
            \lhaf(\bm{B})=\sum_{Z\in \mathrm{SPM}(2m)}\prod_{(i,j)\in Z}B_{i,j},
            \label{eq:gen_loop_haf_def}
        \end{align}
        where $\text{SPM}(2m)$ is the set of single-pair matchings on a graph of order $2m$ with loops~\cite{bjorklund2019faster}.

    \section{\label{sec:master_theorem}A loop Hafnian master theorem}
        
        In this section we recall some useful properties of the loop Hafnian and use them to obtain a loop Hafnian master theorem that will be the basis for the rest of the results derived in the manuscript.

        For an arbitrary $2M\times 2M$ symmetric matrix $\bm{A}$, and a $2M$-length vector $\gl$, let us write $\lhaf(\bm{A},\bm{\gamma}):=\lhaf(\filldiag(\bm{A}, \gl))$. From Eq.~\eqref{eq:gen_loop_haf_def}, we have
        \begin{align}
            \lhaf(\bm{A},\bm{\gamma}) &= \sum_{Z \in \text{SPM}(2M)}\prod_{\substack{(i,j) \in Z \\ i \neq j}}A_{i,j} \left(\prod_{(i,i) \in Z} \gamma_i\right).
            \label{eq:lhaf_def}
        \end{align}
        
        The function $\lhaf(\bm{A}, \gl)$ has the following scaling property. If we transform the elements of $\bm{A}$ as $A_{i,j} \to c_i c_j A_{i,j}$ and the elements of $\gl$ as $\gamma_i \to c_i \gamma_i$, we find~\cite{bjorklund2019faster}
        \begin{align}
        	\lhaf&(\diag(\bm{c})\bm{A}\diag(\bm{c}), \diag(\bm{c})\gl)\nonumber\\
            &= \sum_{Z \in \text{SPM}(2M)} \prod_{\substack{(i,j) \in Z \\ i \neq j}} c_i c_j A_{i,j} \left(\prod_{(i,i) \in Z} c_i \gamma_i\right) \nonumber\\
        	& = \left(\prod_{i=1}^{2M} c_i\right) \sum_{Z \in \text{SPM}(2M)} \prod_{\substack{(i,j) \in Z \\ i \neq j}}A_{i,j} \left(\prod_{(i,i) \in Z} \gamma_i\right) \nonumber\\
        	& =  \left(\prod_{i=1}^{2M} c_i\right) \lhaf(\bm{A}, \gl),
            \label{eq:lhaf_scaling}
        \end{align}
        where $\bm{c}=(c_1,\ldots,c_{2M})$, and $\diag(\bm{c})$ is a diagonal matrix whose entries are the components of $\bm{c}$.
        
        The loop Hafnian is also invariant under permutations of rows and columns, so for some permutation matrix $\bm{P}$, we have that
        \begin{equation}
        	\lhaf(\bm{P} \bm{A} \bm{P}^{\mathrm{T}}, \bm{P} \gl) = \lhaf(\bm{A},\gl),
            \label{eq:lhafP}
        \end{equation}
        which can be seen from the observation that the index which labels each row and column is arbitrary, and so we are free to relabel them.
        
        Now, suppose that $\bm{A}$ and $\gl$ correspond to a valid Gaussian state. Noticing that the distribution defined in Eq.~\eqref{eq:gbs_probability_distribution} must be normalized,
        \begin{align}\label{eq:prob1}
        	1 = \sum_{\nl} \Pr(\nl | \bm{A}, \gl),
        \end{align}
        we obtain the relation
        \begin{align}
        	\sum_{\nl} \frac{\lhaf \left({\bm{A}}_{\nl\oplus\nl}, {\gl}_{\nl \oplus \nl} \right)}{\prod_{i=1}^M n_i!} = \frac{\exp\left( \tfrac12 \gl^\mathrm{T} \left[\mathbb{I} - \bm{X} \bm{A} \right]^{-1} \bm{X} \gl \right)}{\sqrt{\det\left(\mathbb{I} - \bm{X}\bm{A} \right)}}
        	\label{eq:sum_lhafs} .
        \end{align}
        
        Define now the matrices $\bm{D}(\zl) = \diag(\zl) \oplus \diag(\zl)$ and $\bm{B}(\zl) = \sqrt{\bm{D}(\zl)}$, where $\zl=(z_1,\ldots,z_M)$ is an $M$-length vector, and substitute $\bm{A} \to \bm{B}(\zl) \bm{A} \bm{B}(\zl)$ and $\gl \to \bm{B}(\zl) \gl$ into Eq.~\eqref{eq:sum_lhafs}. Then, we have
        \begin{align}
        	&\sum_{\nl} \frac{\lhaf \left[\bm{B}(\zl)_{\nl\oplus\nl} {\bm{A}}_{\nl\oplus\nl} \bm{B}(\zl)_{\nl\oplus\nl}, \bm{B}(\zl)_{\nl\oplus\nl}{\gl}_{\nl \oplus \nl} \right]}{\prod_{i=1}^M n_i!} \nonumber\\
            &= \frac{\exp\left(\tfrac12 (\bm{B}(\bm{z}) \gl)^\mathrm{T} \left[\mathbb{I} - \bm{X} \bm{B}(\zl) \bm{A} \bm{B}(\zl) \right]^{-1} \bm{B}(\bm{z}) \bm{X} \gl \right)}{\sqrt{\det\left[\mathbb{I} - \bm{X}\bm{B}(\zl) \bm{A} \bm{B}(\zl)\right]}}\nonumber\\
            & = q(\bm{A}, \gl, \zl).
            \label{eq:lhaf_master1}
        \end{align}
        
        We will do some rearranging on the expression defining $q(\bm{A}, \gl, \zl)$.
        Firstly, we notice that we can use Sylvester's determinant theorem~\cite{sylvester1883xxxix,pozrikidis2014introduction}, which states that $\det(\mathbb{I} + \bm{A}\bm{B}) = \det(\mathbb{I} + \bm{B}\bm{A})$, to write
        \begin{align}
        	\det\left[\mathbb{I} - \bm{X} \bm{B}(\zl) \bm{A} \bm{B}(\zl) \right] &= \det\left[\mathbb{I} -  \bm{B}(\zl) \bm{X} \bm{B}(\zl) \bm{A}  \right] \nonumber\\
            &= \det\left[\mathbb{I} -  \bm{D}(\zl) \bm{X} \bm{A}  \right],
            \label{eq:sylvester_reogarnize}
        \end{align}
        where we have used the fact that $\bm{B}(\zl)$ and $\bm{D}(\zl)$ commute with $\bm{X}$, and $\bm{D}(\zl)=\bm{B}(\zl)\bm{B}(\zl)$.
        We will also rearrange the numerator to be in terms of $\bm{D}(\zl)$ rather than $\bm{B}(\zl)$.
        We start by expanding the exponent using a binomial expansion, $(1 - x)^{-1} = \sum_{j=0}^\infty x^j$, and noticing that
        \begin{align}
        	\left(\bm{X} \bm{B} \bm{A} \bm{B} \right)^{j} & = (\bm{X} \bm{B} \bm{A} \bm{B} )(\bm{X} \bm{B} \bm{A} \bm{B}) \ldots \nonumber\\
        	& = (\bm{X} \bm{B}^{-1} \bm{D} \bm{A}) (\bm{X} \bm{D} \bm{A}) \ldots (\bm{X} \bm{D} \bm{A}) \bm{B} \nonumber\\
        	& = \bm{B}^{-1} \left( \bm{D} \bm{X} \bm{A} \right)^{j} \bm{B}.
            \label{eq:exponent_reorganization}
        \end{align}
        Then, we can collapse the binomial expansion back down and write
        \begin{align}
        	&\frac12 \left[\bm{B}(\zl) \gl\right]^\mathrm{T} \left[\mathbb{I} - \bm{X} \bm{B}(\zl) \bm{A} \bm{B}(\zl) \right]^{-1} \bm{B}(\zl) \bm{X} \gl \nonumber\\
            & = \frac12 \left[\bm{B}(\zl) \gl\right]^\mathrm{T} \bm{B}^{-1}(\zl) \left[\mathbb{I} - \bm{D}(\zl) \bm{X} \bm{A} \right]^{-1} \bm{B}(\zl) \bm{B}(\zl) \bm{X} \gl \nonumber\\
        	& = \frac12 \gl^\mathrm{T} \left[ \mathbb{I} - \bm{D}(\zl) \bm{X} \bm{A} \right]^{-1} \bm{D}(\zl) \bm{X} \gl.
            \label{eq:final_exponent_D}
        \end{align}
        Combining Eqs.~\eqref{eq:sylvester_reogarnize} and~\eqref{eq:final_exponent_D}, we have
        \begin{equation}
        	q(\bm{A}, \gl, \zl) = \frac{\exp\left( \frac12 \gl^\mathrm{T} \left[ \mathbb{I} - \bm{D}(\zl) \bm{X} \bm{A} \right]^{-1} \bm{D}(\zl) \bm{X} \gl \right)}{\sqrt{\det\left[\mathbb{I} -  \bm{D}(\zl) \bm{X} \bm{A}  \right]}}. 
        \label{eq:q_d}
        \end{equation}
        
        With this form of $q(\bm{A}, \gl, \zl)$, and applying the loop Hafnian scaling relation of Eq.~\eqref{eq:lhaf_scaling} to Eq.~\eqref{eq:lhaf_master1}, we obtain
        \begin{align}
        	\sum_{\nl} \lhaf \left({\bm{A}}_{\nl\oplus\nl}, {\gl}_{\nl \oplus \nl} \right)\left(\prod_{i=1}^M \frac{z_i^{n_i}}{n_i!} \right) = q(\bm{A}, \gl, \zl).
        	\label{eq:lhaf_master}
        \end{align}
        
        We have thus derived a loop Hafnian master theorem, which reduces to the Hafnian master theorem of ref.~\cite{kocharovsky2022hafnian} when $\gl = \bm{0}$.

    \section{\label{sec:total_ph_prob}Total photon number probabilities}

        In this section we start our program of evaluating coarse-grained probabilities by evaluating the total photon number distribution. For pure Gaussian states this distribution can be easily calculated by noting that any $M$-mode Gaussian pure state can be obtained by mixing $M$ single-mode pure Gaussian states on a lossless interferometer. Since a lossless interferometer commutes with the total number of particles operator $\sum_{i=1}^M \hat{a}_i^\dagger \hat{a}_i$, the total number of particles distribution remains the same before and after its action. Thus, the total photon number distribution at the output is just the convolution of the single-mode probabilities at the input.

        For a general Gaussian state, suppose that we detect a total of $N$ photons in a subset of modes $Y\subseteq [M]$, while we detect vacuum in the remaining modes, which we group together in the subset $Z\subseteq [M]$. Here, $[M]=\{1,\ldots, M\}$. Note that $Y\cup Z=[M]$ and $Y\cap Z=\emptyset$. The probability associated to this detection event can be computed as
        \begin{align}
        	\Pr\left(N_Y = N, N_Z = 0| \bm{A}, \gl \right) & = \sum_{\nl \in \mathcal{K}_{Y,Z}^{N,0}} \Pr(\nl | \bm{A}, \gl),
            \label{eq:total_photon_number_probability}
        \end{align}
        where $N_Y=\sum_{k\in Y}n_k$, $N_Z$ is similarly defined, and
        \begin{equation}
            \mathcal{K}_{Y,Z}^{N,N'}=\left\{\bm{n}\,\bigg|\,\sum_{k\in Y}n_k=N\text{ and }\sum_{k\in Z}n_k=N'\right\}.
            \label{eq:two_division_subset}
        \end{equation}
        In terms of the loop Hafnian, this probability reads
        \begin{align}
        	\Pr(N_Y = N, N_Z = &0| \bm{A}, \gl ) = \Pr\left(\bm{0}| \bm{A}, \gl \right)\nonumber\\
            &\times\sum_{\nl \in \mathcal{K}_{Y,Z}^{N,0}}\frac{\lhaf \left({\bm{A}}_{\nl\oplus\nl}, {\gl}_{\nl \oplus \nl} \right)}{\prod_{i=1}^M n_i!}.
            \label{eq:total_photon_number_probability_lhaf}
        \end{align}

        We can now use the result in Eq.~\eqref{eq:lhaf_master} to express the probability $\Pr\left(N_Y = N, N_Z = 0| \bm{A}, \gl \right)$ in terms
        of function $q(\bm{A}, \gl, \vec{z})$. Indeed, by defining a vector $\zl'$ such that $z_k'=\eta$ for $k\in Y$, and $z_k'=0$ for $k\in Z$, we can see that 
        \begin{align}
            \!\!\!q(\bm{A}, \gl, \zl')&=\sum_{\nl} \lhaf \left({\bm{A}}_{\nl\oplus\nl}, {\gl}_{\nl \oplus \nl} \right)\left(\prod_{i=1}^M \frac{(z_i')^{n_i}}{n_i!} \right)\nonumber\\
            &=\sum_{N_Y=0}^\infty\sum_{\nl \in \mathcal{K}_{Y,Z}^{N_Y,0}}\!\!\!\frac{\lhaf \left({\bm{A}}_{\nl\oplus\nl}, {\gl}_{\nl \oplus \nl} \right)}{\prod_{i\in Y} n_i!}\,\eta^{N_Y},
            \label{eq:reduction_to_subset}
        \end{align}
        where we took into account that the product $\prod_{i=1}^M(z_i')^{n_i}$ vanishes unless $n_k=0$ for all $k\in Z$, which, in turn, is the only way to obtain $N_Z=0$. We may readily see that 
        \begin{align}
            \frac{1}{N!}\left.\frac{\partial^N q(\bm{A}, \gl, \zl')}{\partial \eta^N}\right|_{\eta=0}=\sum_{\nl \in \mathcal{K}_{Y,Z}^{N,0}}\!\!\!\frac{\lhaf \left({\bm{A}}_{\nl\oplus\nl}, {\gl}_{\nl \oplus \nl} \right)}{\prod_{i\in Y} n_i!},
            \label{eq:q_function_derivative}
        \end{align}
        which implies 
        \begin{align}
        	\Pr\left(N_Y = N, N_Z = 0| \bm{A}, \gl \right) &=\Pr\left(\bm{0}| \bm{A}, \gl \right)\nonumber\\
            &\times \frac{1}{N!}\left.\frac{\partial^N q(\bm{A}, \gl, \zl')}{\partial \eta^N}\right|_{\eta=0}.
            \label{eq:total_photon_number_probability_q_func}
        \end{align}

        We can further notice from Eq.~\eqref{eq:q_d} that
        \begin{equation}
	        q(\bm{A}, \gl, \zl') = q(\bm{A}_Y, \gl_Y,\zl = (\eta,\dots,\eta))
            \label{eq:q_func_reduction}
        \end{equation}
        where $\bm{A}_Y$ and $\gl_Y$ are found by keeping only the $k$th and $(k+M)$th rows and columns, or elements of $\bm{A}$ and $\gl$ for $k \in Y$. In this expression, $\zl = (\eta,\dots,\eta)$ has length $|Y|$ (i.e., the number of elements in $Y$).

        For future purposes, it is convenient to introduce the function $f_N(\bm{A},\gl)$, defined as
        \begin{equation}
        	f_N(\bm{A}, \gl) = \frac{1}{N!} \left. \frac{\partial^N}{\partial \eta^N} q\left(\bm{A}, \gl, \zl=(\eta, \eta, \dots \eta) \right) \right|_{\eta=0}
            \label{eq:f_function},
        \end{equation}
        where it is understood that $\zl$ has half the length of $\gl$. This definition allows us to recast Eq.~\eqref{eq:total_photon_number_probability_q_func} into the form
        \begin{equation}
            \Pr\left(N_Y=N,N_Z=0|\bm{A},\gl\right)=\Pr(\bm{0}|\bm{A},\gl)f_N(\bm{A}_Y,\gl_Y).
            \label{eq:total_photon_number_probability_f_func}
        \end{equation}

        If we consider $q(\bm{A}, \gl, \zl=(\eta, \eta, \dots, \eta))$ as a polynomial in $\eta$, we can see that $f_N(\bm{A}, \gl)$ corresponds to the $N$th polynomial coefficient. Therefore, to compute $f_N(A, \gl)$, we need simply express $q(\bm{A}, \gl, \zl=(\eta, \eta, \dots, \eta))$ a series expansion in $\eta$. 
        
        In order to do this, let us recall the following useful series expansions
        \begin{align}
            &\exp(\bm{A})=\sum_{j=0}^\infty \frac{\bm{A}^j}{j!},\;\;(\mathbb{I}-\bm{A})^{-1}=\sum_{j=0}^\infty\bm{A}^k,\\ 
            &\log\left[\det\left(\mathbb{I}-\bm{A}\right)\right]=-\sum_{j=1}^\infty\frac{\tr(\bm{A}^j)}{j}.
            \label{eq:useful_series_expansions}
        \end{align}
        Then, we can write the numerator and denominator of $q(\bm{A}, \gl, \zl=(\eta, \eta, \dots, \eta))$ as
        \begin{align}
            &\exp\left(\frac12 \gl^\mathrm{T} [\mathbb{I} - \eta \bm{X}\bm{A}]^{-1} \eta \bm{X} \gl \right) \nonumber\\
            &\qquad = \sum_{j=0}^\infty \frac{1}{j!} \left(\frac12 \gl^\mathrm{T} [\mathbb{I} - \eta \bm{X}\bm{A}]^{-1} \eta \bm{X} \gl  \right)^j \nonumber\\
	        &\qquad = \sum_{j=0}^\infty \frac{1}{j!} \left( \sum_{k=0}^\infty \frac{\gl^\mathrm{T} [\eta \bm{X}\bm{A}]^{k} \eta \bm{X} \gl }{2}  \right)^j  \nonumber\\
	        &\qquad = \sum_{j=0}^\infty \frac{1}{j!} \left( \sum_{k=1}^\infty \frac{\gl^\mathrm{T} [\bm{X}\bm{A}]^{k-1}  \bm{X} \gl }{2} \eta^k \right)^j,
            \label{eq:q_numerator_expansion}
        \end{align}
        and
        \begin{align}
            \frac{1}{\sqrt{\det(\mathbb{I} - \eta \bm{X} \bm{A})}} &= \exp\left(-\tfrac12 \log\left[\det\left(\mathbb{I} - \eta \bm{X} \bm{A}\right)\right]\right) \nonumber\\
	        &= \sum_{j=0}^\infty \frac{1}{j!} \left(\sum_{k=1}^\infty \frac{\tr\left([\bm{X} \bm{A}]^k\right)}{2k} \eta^k \right)^j.
            \label{eq:q_denominator_expansion}
        \end{align}
        Bringing these expansions together, we obtain the desired expression for $q(\bm{A}, \gl, \zl=(\eta, \eta, \dots, \eta))$.

        Focusing only on terms on terms up to order $N$ in $\eta$, we may truncate series~\eqref{eq:q_numerator_expansion} and~\eqref{eq:q_denominator_expansion} at $N$, and define a finite-sum version of $q(\bm{A}, \gl, \zl=(\eta, \eta, \dots, \eta))$ as
        \begin{align}
            &q_N(\bm{A}, \gl, \eta) \nonumber\\
            &=\sum_{j=0}^N \frac{1}{j!}\left(\sum_{k=1}^N \left(\frac{\tr\left([\bm{X} \bm{A}]^k\right)}{2k} + \frac{\gl^\mathrm{T} [\bm{X}\bm{A}]^{k-1}  \bm{X} \gl }{2} \right)\eta^k \right)^j.
            \label{eq:truncated_series_q}
        \end{align}
        Since
        \begin{align}
            &\left. \frac{\partial^N}{\partial \eta^N} q\left(\bm{A}, \gl, \zl=(\eta, \eta, \dots \eta) \right) \right|_{\eta=0}\!\!\!=\left. \frac{\partial^Nq_N\left(\bm{A}, \gl, \eta \right)}{\partial \eta^N}  \right|_{\eta=0},
            \label{eq:truncated_complete_derivative}
        \end{align}
        we can see that $f_N(\bm{A},\bm{\gamma})$ is not affected by the truncation. So, we can recast Eq.~\eqref{eq:total_photon_number_probability_f_func}, as
        \begin{align}
            &\Pr\left(N_Y=N,N_Z=0|\bm{A},\gl\right)=\Pr(\bm{0}|\bm{A},\gl)f_N(\bm{A}_Y,\gl_Y)\nonumber\\
            &\qquad\qquad\qquad\qquad=\frac{1}{N!}\left. \frac{\partial^Nq_N\left(\bm{A}_Y, \gl_Y, \eta \right)}{\partial \eta^N}  \right|_{\eta=0}
            \label{eq:total_photon_number_prob}
        \end{align}

        Note that in the special case where $Y=[M]$, where $\bm{A}_{[M]}=\bm{A}$ and $\gl_{[M]}=\gl$, we obtain an expression for the \textit{total photon number} distribution corresponding to the Gaussian state described by $\bm{A}$ and $\gl$. This expression agrees with the results from Ref.~\cite{avagyan2023multi}, where it is shown that the total photon number distribution depends only the spectrum of the covariance matrix and absolute displacements in each eigenspace. Indeed, note that the expression above only depends on the trace of powers of $\bm{X}\bm{A} = \mathbb{I} - \bm{\Sigma}^{-1}$ (equivalently, they only depend on the spectrum of the hermitian covariance matrix $\bm{\Sigma}$) and the products $\gl^\mathrm{T} [\bm{X}\bm{A}]^{k-1}  \bm{X} \gl = \bar{\zl}^\mathrm{T} (\bm{\Sigma}^{\mathrm{T}})^{-1} \bm{X} (\mathbb{I} - \bm{\Sigma}^{-1})^{k-1} (\bm{\Sigma})^{-1} \bar{\zl} $.

        In Appendix~\ref{app:algorithm} we describe an algorithm, first introduced in Ref.~\cite{bjorklund2019faster} and whose implementation is available in~\cite{lhaf2024walrus}, for computing $f_N(\bm{A}, \gl)$ (with $\bm{A}$ a $2M\times 2M$ matrix and $\gl$ a $2M$ length vector) in $\mathcal{O}(NM^3+N^2\log N)$ time.

    \section{\label{sec:trace_formulas}Trace formula for the loop Hafnian}

        An interesting application of the loop Hafnian master theorem is the derivation of the so called \textit{trace formula} for the loop Hafnian, which has been previously reported in Refs.~\cite{bjorklund2019faster, cygan2015faster, bulmer2022boundary, bulmer2022threshold}. In this section we will work out the details of this derivation. To do this, we will first introduce the notion of a \textit{finite-difference sieve}, closely following the definition given in Ref.~\cite{bax1998finite}. Then, we will describe how to nicely combine it with the master theorem in order to obtain the trace formula of the loop Hafnian. 
        
            Let $P(x_1, x_2, \dots, x_n)$ be a polynomial of degree at most $n$, which we write as $P()$ for short. This polynomial will have the general form
            \begin{equation}
		        P(x_1, x_2, \dots, x_n) = \sum_{\substack{\{p_k\}\text{ s.t. }\\ \sum_k p_k \leq n}} c_{p_1, p_2, \dots, p_n} x_1^{p_1} x_2^{p_2} \cdots x_n^{p_n}.
                \label{eq:poly_def}
	        \end{equation}
            and let us define the operator $D_j$
        	\begin{align}
        		D_{x_j} P() = \frac{P(u_j) - P(v_j)}{u_j - v_j} \ \text{for}\ u_j \neq v_j,
                \label{eq:finite_diff_operator}
        	\end{align}
            where we use $P(u_j)$ as a shorthand for $P(x_1, x_2, \cdots, x_j=u_j, \cdots, x_n)$, and similarly for $P(v_j)$.

            We can readily check that $D_{x_j} x_j^0 = 0$ and $D_{x_j} x_j^1 = 1$. Moreover, the different $D_{x_j}$ commute with each other. These properties allow us to compute the multilinear term of $P()$ as (see Ref.~\cite{bax1998finite} for a simple proof of this statement)
            \begin{equation}
		          D_{x_1} D_{x_2} \cdots D_{x_n} P() = c_{1,1,\dots,1}.
                \label{eq:multilinear_coeff}
	        \end{equation}
            The calculation of this coefficient involves $2^n$ evaluations of the polynomial $P()$. 

            To determine the coefficients corresponding to non-multilinear terms in $P()$, we can use the $k$th finite-difference operator
            \begin{align}
    			D_{x_j}^{(k)} P() &= \frac{1}{(u_j - v_j)^k} \nonumber\\&\times \sum_{m=0}^k \binom{k}{m} (-1)^{k-m} P(v_j + m(u_j - v_j)).
                \label{eq:non_multilinear_finite_diff}
		      \end{align}
            For any $k_1 + k_2 + \cdots + k_n = n$, the coefficient $c_{k_1,k_2,\dots,k_n}$ can be computed using the relation~\cite{bax1998finite}
            \begin{equation}
    			D_{x_1}^{(k_1)} D_{x_2}^{(k_2)} \cdots D_{x_n}^{(k_n)} P() = \left(\prod_{i=1}^n k_i!  \right)c_{k_1, k_2, \dots k_n}.
                \label{eq:non_multilinear_coeff}
		      \end{equation}
            The calculation of this coefficient involves $\prod_{j=1}^n (k_j+1)$ evaluations of the polynomial $P()$.

            The definition of the finite-difference operators gives us a lot of freedom in choosing the values of $u_j$ and $v_j$, provided only that they are different. An example of the utility of this freedom comes by letting $P()$ be the generating function of the permanent. If we choose $u_j=1$ and $v_j=0$ for all $j$, the application of the finite-difference operators recovers Ryser's formula for the permanent~\cite{ryser1963combinatorial} and its generalizations for repeated rows. Using $u_j=1$ and $v_j=-1$, we can instead recover Glynn's formula for the permanent~\cite{glynn2010permanent} and its repeated row generalizations. A repeated row generalization of Glynn's  formula by Shchesnovich~\cite{shchesnovich2020classical} uses something along the lines of $u_j=1$, $v_j=\exp[2 i \pi / (n_j+1)]$, where $n_j$ is the number of repetitions of a given row.

            Consider now Eq.~\eqref{eq:lhaf_master} with the replacement $\zl\to\eta\zl$:
            \begin{align}
            	&q(\bm{A}, \gl, \eta\zl)\nonumber\\
                &=\sum_{\nl} \lhaf \left({\bm{A}}_{\nl\oplus\nl}, {\gl}_{\nl \oplus \nl} \right)\left(\prod_{i=1}^M \frac{(\eta z_i)^{n_i}}{n_i!} \right)\nonumber\\
                &=\sum_{N'=0}^\infty\eta^{N'}\!\sum_{\bm{n}\in \mathcal{K}_{N'}}\lhaf \left({\bm{A}}_{\nl\oplus\nl}, {\gl}_{\nl \oplus \nl} \right)\left(\prod_{i=1}^M \frac{z_i^{n_i}}{n_i!} \right),
            	\label{eq:lhaf_finite_diff_1}
            \end{align}
            where $N'=n_1+\cdots+n_M$ and $\mathcal{K}_{N'}\equiv \mathcal{K}_{[M], \emptyset}^{N', 0}$, following definition~\eqref{eq:two_division_subset}, is the set of all detection patterns $\nl$ with a total of $N'$ photons. 
            Differentiating $N$ times with respect to $\eta$ and evaluating at $\eta =0$, we obtain the relation
            \begin{align}
            	f_N(\bm{A}, \gl, \zl)&:=\frac{1}{N!}\left.\frac{\partial^Nq(\bm{A},\gl,\eta\zl)}{\partial \eta^N}\right|_{\eta=0}\nonumber\\
                &=\sum_{\bm{n}\in \mathcal{K}_{N}}\lhaf \left({\bm{A}}_{\nl\oplus\nl}, {\gl}_{\nl \oplus \nl} \right)\left(\prod_{i=1}^M \frac{z_i^{n_i}}{n_i!} \right),
            	\label{eq:lhaf_finite_diff_2}
            \end{align}
            which can be viewed as a loop Hafnian master theorem related to patterns having a fixed number of detected photons.

            As can be seen, $f_N(\bm{A},\gl, \zl)$ is a polynomial of degree $N$ in the variables $z_1,\dots,z_M$. Thus, we can find the coefficients of this polynomial by applying the finite-difference operators defined in Eq.~\eqref{eq:non_multilinear_finite_diff}. Since these coefficients correspond to loop Hafnians, this procedure gives us a new technique for their computation:  
            \begin{equation}
        		\lhaf\left( {\bm{A}}_{\nl \oplus \nl}, {\gl}_{\nl \oplus \nl} \right) = \left( \prod_{i=1}^M D_{z_i}^{(n_i)} \right) f_N(\bm{A}, \gl, \zl).
                \label{eq:finite_diff_lhaf}
	        \end{equation}
            If we let $u_j = n_j$ and $v_j = -n_j$ for all $j$, we recover the finite-difference sieve method used in Ref.~\cite{bulmer2022boundary}. It is also worth mentioning that we can apply the finite-difference sieve directly to Eq.~\eqref{eq:lhaf_master}, in which case we can obtain Kan's algorithm for computing the loop Hafnian~\cite{kan2008moments}.

            Note that the computation of $f_N(\bm{A},\gl,\zl)$ can be done using the same algorithm that is used to compute $f_N(\bm{A},\gl)$, we need only modify Eq.~\eqref{eq:truncated_series_q} according to the transformations
            \begin{align}
                \frac{\tr\left(\left[\bm{X}\bm{A}\right]^k\right)}{2k}&\to\frac{\tr\left(\left[\bm{D}(\zl)\bm{X}\bm{A}\right]^k\right)}{2k},\\
                \frac{\gl^{\mathrm{T}}\left[\bm{X}\bm{A}\right]^{k-1}\bm{X}\gl}{2}&\to\frac{\gl^{\mathrm{T}}\left[\bm{D}(\zl)\bm{X}\bm{A}\right]^{k-1}\bm{D}(\zl) \bm{X}\gl}{2}.
            \end{align}
            
    \section{\label{sec:blocked_hafnian}The blocked loop Hafnian}

        Let us now turn our attention to three problems related to the computation of the photon-number statistics of Gaussian states, two of which can be interpreted as slight generalizations of the problem of determining photon-number probabilities.

        First, we seek to complete our program of evaluating general coarse-grained photon number distributions of $M$-mode photonic systems. So far we have considered only the case where we detect a total of $N$ photons in a subset of modes, while in the remaining modes we detect vacuum. Here, we will consider general groupings of the photon number detectors, and an arbitrary number of total detected photons in each group (see Fig.~\ref{fig:coarse_grained_probs} for an illustration of this problem in the context of GBS).

        Second, suppose that each of the original $M$ modes of our photonic system can support a number, say $K$, of \textit{internal modes}. For example, we can think of our original $M$ modes as the input-output \textit{ports} of an interferometer, while the internal modes can be related to the polarization, frequency, or any other type of degree of freedom of the input-output light (see Fig.~\ref{fig:internal_modes_heralding}). From now on we will refer to the $M$ original degrees of freedom as \textit{external modes}. Furthermore, suppose that we have access to photon-number-resolving detectors that can only distinguish between the $M$ external modes of our system. We are interested in studying how the presence of these internal modes affects the corresponding photon number statistics. The study of this type of systems is of significant physical relevance, as it allows us to model the effect of spectral or temporal imperfections, as well as of other sources of photon distinguishability, in the preparation of both Gaussian and non-Gaussian states.

        Finally, we are interested in computing photon-number moments and cumulants of Gaussian states in the same scenarios we illustrated above. That is, we are interested in computing photon-number moments and cumulants for an arbitrary grouping of the photon-number detectors, and for situations in which our system has a certain number of internal modes that our detectors cannot readily distinguish.

        We will show that all of these problems can be studied using the loop Hafnian master theorem of Sec.~\ref{sec:master_theorem}, and the finite-difference sieve introduced in Sec.~\ref{sec:trace_formulas}. Specifically, we will prove that they can be related to computing a generalized version of the loop Hafnian, which we call the \textit{blocked loop Hafnian}.

        In what follows we will work out the details of the definition of the blocked loop Hafnian, and we will find a finite-difference formula for its computation. Then, we will describe its relation to the generalizations of the problem of computing photon-number probability distributions, and that of computing photon-number moments and cumulants. 
        
        \subsection{\label{sec:def_Lamb_lhaf}Definition of the blocked loop Hafnian}

            Consider a Gaussian state represented by a $2M\times2M$ adjacency matrix $\bm{A}$ and a $2M$-length loop vector $\gl$. Let us arrange the $M$ modes of the system in $L$ groups that do not necessarily have the same size (see Fig.\ref{fig:coarse_grained_probs}). Note that $1\leq L\leq M$. Mathematically, this grouping of modes is represented by a \textit{set partition} of $[M]=\{1,\dots, M\}$ with $L$ blocks, i.e., a set $\Lambda=\{\Lambda_1,\dots,\Lambda_L\}$ where $\Lambda_j\subseteq [M]$ (a set of subsets of $[M]$), $\Lambda_j\neq \emptyset\,\forall j$ (that are non-empty), $\Lambda_j\cap\Lambda_k=\emptyset$ for $j\neq k$ (mutually disjoint), and $\bigcup_{j=1}^L\Lambda_j=[M]$ (whose union is $[M]$). Associated to this set partition, we define a \textit{coarse-grained photon detection pattern} $\bm{b}=(b_1,\dots,b_L)$, where $b_j=\sum_{i\in\Lambda_j}n_i$ and, as before, $n_i$ is the number of detected photons at mode $i$.

            The set of all detection patterns $\bm{n}=(n_1,\dots,n_M)$ that are compatible with $\bm{b}$ can be defined in a similar fashion to Eq.~\eqref{eq:two_division_subset}:
            \begin{equation}
                \mathcal{K}_{\Lambda}^{\bm{b}}=\left\{\bm{n}\,\bigg|\,\sum_{i\in \Lambda_j}n_i=b_j\text{ for all }\Lambda_j\in \Lambda\right\}.
                \label{eq:coarse_grained_subset}
            \end{equation}
            We will now isolate all the possible $\bm{n}\in \mathcal{K}_{\Lambda}^{\bm{b}}$ from the loop Hafnian master theorem.

            In Eq.~\eqref{eq:lhaf_master}, for each value of $j\in \{1,\dots, L\}$, let $z_i = w_j$ for all $i\in\Lambda_j$. For example, if $M=4$ and $\Lambda = \{\Lambda_1=\{1\},\Lambda_2=\{3\},\Lambda_3=\{2,4\}\}$, we will have that 
            $z_1= w_1$, $z_3=w_2$ and $z_2=z_4=w_3$, which is equivalent to writing $\bm{z}=(w_1,w_3,w_2,w_3)$. From this, we can see that
            \begin{align}
                \prod_{i=1}^M\frac{z_{i}^{n_{i}}}{n_{i}!}&=\left(\prod_{j=1}^L\prod_{i\in\Lambda_j}\frac{1}{n_{i}!}\right)\prod_{j=1}^Lw_j^{\sum_{k\in\Lambda_j}n_{k}}\nonumber\\
                &=\left(\prod_{j=1}^L\prod_{i\in\Lambda_j}\frac{1}{n_{i}!}\right)\prod_{j=1}^Lw_j^{b_j},
                \label{eq:change_in_product}
            \end{align}
            which allows us to write the master theorem as
            \begin{align}
                &q(\bm{A},\gl,\bm{w})\nonumber\\
                &=\sum_{\bm{n}} \lhaf \left(\bm{A}_{\bm{n}\oplus\bm{n}}, \gl_{\bm{n} \oplus \bm{n}} \right)\left(\prod_{j=1}^L\prod_{i\in\Lambda_j}\frac{1}{n_{i}!}\right)\prod_{j=1}^Lw_j^{b_j}\nonumber\\
                &=\sum_{\bm{b}}\prod_{j=1}^Lw_j^{b_j}\!\!\left[\sum_{\bm{n}\in \mathcal{K}_{\Lambda}^{\bm{b}}}\!\!\lhaf \left(\bm{A}_{\bm{n}\oplus\bm{n}}, \gl_{\bm{n} \oplus \bm{n}} \right)\!\!\left(\prod_{j=1}^L\prod_{i\in\Lambda_j}\!\!\frac{1}{n_{i}!}\right)\right],
                \label{eq:change_in_master_theorem}
            \end{align}
            where we noted that $\sum_{\bm{n}}=\sum_{\bm{b}}\sum_{\bm{n}\in \mathcal{K}_{\Lambda}^{\bm{b}}}$, and we made the replacement $\zl\to\bm{w}$, with $\bm{w}$ the vector whose components are the different $w_j$ repeated according to the elements in the different $\Lambda_j$ (just as in the example above). 

            Eq.~\eqref{eq:change_in_master_theorem} closely resembles an expression for a master theorem involving only the photon detection patterns $\bm{b}$. Indeed, we can think of this expression as the master theorem for the term inside the square brackets (up to a normalization factor). Following this idea, let us define the \textit{$\Lambda$-blocked loop Hafnian}, $\lhaf_{\Lambda}(\bm{A},\gl,\bm{b})$, according to the relation
            \begin{align}
                \lhaf_{\Lambda}&\left(\bm{A},\gl,\bm{b}\right)\left(\prod_{j=1}^L\frac{1}{b_j!}\right) \nonumber\\
                &= \sum_{\bm{n}\in \mathcal{K}_{\Lambda}^{\bm{b}}}\!\!\lhaf \left(\bm{A}_{\bm{n}\oplus\bm{n}}, \gl_{\bm{n} \oplus \bm{n}} \right)\left(\prod_{j=1}^L\prod_{i\in \Lambda_j}\frac{1}{n_{i}!}\right).
                \label{eq:Lamb_lhaf_def}
            \end{align}
            Then, we can write the \textit{$\Lambda$-blocked loop Hafnian master theorem} as
            \begin{align}
            	\sum_{\bm{b}}\lhaf_{\Lambda}\left(\bm{A},\gl,\bm{b}\right)\left(\prod_{j=1}^L\frac{w_j^{b_j}}{b_j!}\right) = q(\bm{A}, \gl, \bm{w}).
            	\label{eq:Lamb_lhaf_master}
            \end{align}

            At first glance, the definition of $\lhaf_{\Lambda}(\bm{A},\gl,\bm{b})$ might seem a bit underwhelming, since it generally involves the computation of several loop Hafnians with high time complexity. Indeed, the number of elements in the set $\mathcal{K}_{\Lambda}^{\bm{b}}$, for a fixed $\bm{b}$, is $|\mathcal{K}_{\Lambda}^{\bm{b}}|=\prod_{j=1}^L\binom{|\Lambda_j|+b_j-1}{b_j}$, with $|\Lambda_j|$ the number of elements or \textit{length} of the block $\Lambda_j$, and the time complexity of computing a single loop Hafnian of an $2N'\times 2N'$ matrix, with $N'=n_1+\cdots+n_M$, is $\mathcal{O}(N'^32^{N'})$~\cite{bjorklund2019faster}. However, we can use the master theorem to obtain an equation that offers a considerable speedup in the computation of the $\Lambda$-blocked loop Hafnian.

            In Eq.~\eqref{eq:Lamb_lhaf_master}, let $w_j\to\eta w_j$ for all $j\in \{1,\dots, L\}$. Then, we can write
            \begin{align}
                q&(\bm{A},\gl,\eta\bm{w})\nonumber\\
                &=\sum_{\bm{b}}\lhaf_{\Lambda}\left(\bm{A},\gl,\bm{b}\right)\left(\prod_{j=1}^L\frac{(\eta w_j)^{b_j}}{b_j!}\right)\nonumber\\
                &=\sum_{N'=0}^\infty\eta^{N'}\!\!\!\!\sum_{b_1+\cdots+ b_L=N'}\!\!\!\lhaf_{\Lambda}\left(\bm{A},\gl,\bm{b}\right)\left(\prod_{j=1}^L\frac{w_j^{b_j}}{b_j!}\right),
                \label{eq:Lamb_lhaf_reduction}
            \end{align}
            where, again, $N'=\sum_{j=1}^Lb_j=\sum_{i=1}^Mn_i$. Differentiating $N$ times with respect to $\eta$ and evaluating at $\eta=0$, we obtain
            \begin{align}
            	f_N(\bm{A}, \gl, \bm{w})&=\frac{1}{N!}\left.\frac{\partial^Nq(\bm{A},\gl,\eta\bm{w})}{\partial \eta^N}\right|_{\eta=0}\nonumber\\
                &\mkern-80mu=\sum_{b_1+\cdots+ b_L=N'}\lhaf_{\Lambda}\left(\bm{A},\gl,\bm{b}\right)\left(\prod_{j=1}^L\frac{w_j^{b_j}}{b_j!}\right).
            	\label{eq:Lamb_lhaf_reduction_2}
            \end{align}
            Note that $f_N(\bm{A}, \gl, \bm{w})$ is a polynomial of degree $N$ in the variables $\{w_j\}$, so we can repeatedly apply finite-difference operators to compute the coefficients of the polynomial, i.e., to compute $\lhaf_{\Lambda}\left(\bm{A},\gl,\bm{b}\right)$. The final expression reads
            \begin{equation}
        		\lhaf_{\Lambda}\left(\bm{A},\gl,\bm{b}\right) = \left( \prod_{j=1}^L D_{w_j}^{(b_j)} \right) f_N(\bm{A}, \gl, \bm{w}).
                \label{eq:finite_diff_Lamb_lhaf}
	        \end{equation}
    
            The computation of Eq.~\eqref{eq:finite_diff_Lamb_lhaf} requires the evaluation of $f_N(\bm{A}, \gl, \bm{w})$ a total of $\prod_{j=1}^L(b_j+1)$ times, which is generally much less than the $\prod_{j=1}^L\binom{|\Lambda_j|+b_j-1}{b_j}$ operations required by Eq.~\eqref{eq:Lamb_lhaf_def}. Moreover, evaluating $f_N(\bm{A}, \gl, \bm{w})$ can be done in $\mathcal{O}(NM^3+N^2\log N)$ time, much faster than the computation of a single loop Hafnian. This shows the advantage of using Eq.~\eqref{eq:finite_diff_Lamb_lhaf} for determining $\lhaf_{\Lambda}\left(\bm{A},\gl,\bm{b}\right)$ over naively computing it using the definition of Eq.~\eqref{eq:Lamb_lhaf_def}. 

            To finish this section, let us note that, combining Eqs.~\eqref{eq:gbs_probability_distribution} and~\eqref{eq:Lamb_lhaf_def}, we can readily compute the probability of the detection event represented by $\bm{b}$ (and $\Lambda$).
            \begin{align}
                \Pr&(\bm{b}|\bm{A},\gl)=\sum_{\nl\in \mathcal{K}_{\Lambda}^{\bm{b}}}\Pr(\nl|\bm{A},\gl)\nonumber\\
                &=\Pr\left(\bm{0}| \bm{A}, \gl \right)\lhaf_{\Lambda}\left(\bm{A},\gl,\bm{b}\right)\left(\prod_{j=1}^L\frac{1}{b_j!}\right)\nonumber\\
                &=\frac{\Pr\left(\bm{0}| \bm{A}, \gl \right)}{\prod_{j=1}^Lb_j!}\left( \prod_{j=1}^L D_{w_j}^{(b_j)} \right) f_N(\bm{A}, \gl, \bm{w}).
                \label{eq:coarse_grained_prob}
            \end{align}
            This relation generalizes Eq.~\eqref{eq:total_photon_number_probability_f_func} and is the expression for the coarse-grained photon number distribution we were looking for. 

        \subsection{\label{sec:photon_dist_internal}Photon number distribution in the presence of internal modes}

            We are now in the position to see that the computation of photon number distributions for Gaussian states with internal modes corresponds to the calculation of a special type of blocked loop Hafnian.
            
            A Gaussian state with $M$ external modes and $K$ internal modes per external mode can be completely parametrized by a $2MK\times 2MK$ adjacency matrix $\bar{\bm{A}}$ and a $2MK$-length loop vector $\bar{\gl}$, where we use the overline to indicate that we are considering a Gaussian state with multiple internal modes. Let $s_{k,l}$, with $k\in\{1,\dots,M\}$ and $l\in\{1,\dots, K\}$, denote the number of photons at the external mode $k$ and the internal mode $l$. We can group the different $\{s_{k,l}\}$ in a string of non-negative integers $\bm{s}$, of length $MK$, defined as
            \begin{align}
                \bm{s}=(s_{1,1},\dots,s_{1,K})\oplus\cdots\oplus(s_{M,1},\dots,s_{M,K}).
                \label{eq:internal_mode_detection_pattern}
            \end{align}
    
            If we have access to detectors that can only distinguish between the $M$ external modes of the system, we need to define a coarse-grained photon number detection pattern $\nl=(n_1,\dots,n_M)$, whose components are computed as $n_k=\sum_{l=1}^Ks_{k,l}$. To this measurement outcome we can associate the $M$-block set partition (of $[MK]=\{1,\dots,MK\}$) $\Lambda(K) = \{\Lambda_1,\dots, \Lambda_M\}$, where each block is defined as $\Lambda_l = \{(l-1)K+1,\dots,lK\}$. Then, we can see that
            \begin{align}
                \Pr&(\bm{n}|\bar{\bm{A}},\bar{\gl})=\sum_{\bm{s}\in \mathcal{K}_{\Lambda(K)}^{\bm{n}}}\Pr(\bm{s}|\bar{\bm{A}},\bar{\gl})\nonumber\\
                &=\Pr\left(\bm{0}| \bar{\bm{A}}, \bar{\gl} \right)\lhaf_{\Lambda(K)}\left(\bar{\bm{A}},\bar{\gl},\bm{n}\right)\left(\prod_{k=1}^M\frac{1}{n_k!}\right)\nonumber\\
                &=\frac{\Pr\left(\bm{0}| \bar{\bm{A}}, \bar{\gl} \right)}{\prod_{k=1}^Mn_k!}\left(\prod_{k=1}^M D_{w_k}^{(n_k)} \right) f_N(\bar{\bm{A}}, \bar{\gl}, \bar{\bm{w}}),
                \label{eq:prob_internal_modes}
            \end{align}
            where $\bar{\bm{w}}=(w_1,\dots,w_1)\oplus\cdots\oplus(w_M,\dots,w_M)$ with each $w_l$ appearing a total of $K$ times.   

            It is worth mentioning once more that Eq.~\eqref{eq:prob_internal_modes} offers a significant advantage at the computation of photon number distributions in the presence of internal modes. Indeed, in most of the models considering the presence of these modes (see, for instance, Ref.~\cite{shi2022effect}), the exact calculation of probabilities requires a sum over $|\mathcal{K}_{\Lambda(K)}^{\bm{n}}|=\prod_{k=1}^M\binom{K+n_k-1}{n_k}$ terms that, in principle, correspond to loop Hafnians of large $\bar{\bm{A}}$ and $\bar{\gl}$. If we think of the external modes as representing the input-output ports of an interferometer, and of the internal modes as representing other degrees of freedom of the light entering the interferometer, we can see that photons that are in different internal modes will not interfere between them. This allows us to decompose $\bar{\bm{A}}$ and $\bar{\gl}$ as a direct sum of $K$ smaller $2M\times 2M$ blocks and $2M$-length integer strings, i.e., $\bar{\bm{A}}=\bigoplus_{l=1}^K\bm{A}_l$,$\bar{\gl}=\bigoplus_{l=1}^K\gl_l$, which turns the computation of a very expensive loop Hafnian into the computation of loop Hafnians of several smaller matrices (on account of the property $\lhaf(\bm{A}\oplus\bm{B})=\lhaf(\bm{A})\lhaf(\bm{B})$). Furthermore, if we consider that the photons in some of the internal modes are fully distinguishable, calculating the corresponding loop Hafnians reduces to the computation of quantities with polynomial time complexity~\cite{shi2022effect}. However, even after all these considerations, we cannot alleviate the fact that we are summing over a combinatorial number of elements. Eq.~\eqref{eq:prob_internal_modes}, on the other hand, requires only an exponential number of evaluations, namely $\prod_{k=1}^M(n_k+1)$, of a quantity that has time complexity $\mathcal{O}(NK^3M^3+N^2\log N)$. This represents a notable reduction of the computational cost of calculating these photon number probabilities.
        \subsection{Coarse-grained photon-number moments and cumulants}
           Let us now study the computation of coarse-grained photon-number moments and cumulants of Gaussian states. We start by recalling that the photon-number moment generating function of a Gaussian state can be written as
           \begin{align}
                \label{eq:M}
            	\mathcal{M}(\bm{t}) &= \text{Tr}\left(\hat{\varrho} \exp\left[\sum_{j=1}^M t_j \hat{a}^\dagger_j \hat{a}_j\right]\right)\nonumber\\
            	&=\sum_{\bm{n}} e^{\bm{n} \cdot \bm{t}} \Pr(\bm{n}|\bm{A},\bm{\gamma})\\
            	&=\Pr(\bm{0}|\bm{A},\bm{\gamma}) \sum_{\bm{n}} e^{\bm{n} \cdot \bm{t}}\frac{\lhaf \left({\bm{A}}_{\nl\oplus\nl}, {\gl}_{\nl \oplus \nl} \right)}{\prod_{i=1}^M n_i!}
            	%
            	%
            	%
           \end{align}
       We will now use the loop-Hafnian master theorem to simplify this expression. To this end, we can introduce the diagonal matrix $\bm{E}(\bm{t}) = [\oplus_{i=1}^M e^{t_i/2}] \oplus [\oplus_{i=1}^M e^{t_i/2}]$, and using the scaling property Eq.~\eqref{eq:lhaf_scaling} together with the fact that $e^{\bm{n}\cdot \bm{t}} = \prod_{i=1}^M e^{t_i n_i}$ and the results leading to Eq.~\eqref{eq:q_d}, we write
       \begin{align}
       	\mathcal{M}(\bm{t}) &= \Pr(\bm{0}|\bm{A},\bm{\gamma}) \nonumber \\
       	&\times \frac{\exp\left( \frac12 \gl^\mathrm{T} \left[ \mathbb{I} - \bm{E}(\bm{t})^2 \bm{X} \bm{A} \right]^{-1} \bm{E}(\bm{t})^2 \bm{X}\gl \right)}{\sqrt{\det\left[\mathbb{I} -  \bm{E}(\bm{t})^2 \bm{X} \bm{A}  \right]}}
       \end{align}
       Using the expression for the vacuum probability and defining $\bm{G}(\bm{t})=\bm{E}(\bm{t})^2 - \mathbb{I}=\left[ \oplus_{i=1}^M (e^{t_i}-1) \right] \oplus \left[ \oplus_{i=1}^M (e^{t_i}-1) \right]$ we finally obtain, after some lengthy but straightforward algebra, a simple expression for the moment generating function
      \begin{align}
      		\mathcal{M}(\bm{t}) &=\frac{\exp\left( \frac{1}{2} \bar{\bm{z}}^\dagger  [\mathbb{I} -  \bm{G}(\bm{t})\bm{\Sigma}^{(s=1)}]^{-1} \bm{G}(\bm{t}) \bar{\bm{z}}\right)}{\sqrt{\det\left( \mathbb{I}- \bm{G}(\bm{t}) \bm{\Sigma}^{(s=1)}  \right)}},
      \end{align}
           where, we remind the reader, the vector of first moments, $\bar{\bm{z}}$, and the covariance matrix, $\bm{\Sigma}^{(s)}$, are defined according to Eqs.~\eqref{eq:vector_of_means} and~\eqref{eq:covariance_matrix}. 
           In the original Ref.~\cite{cardin2022photon}, there is a typo in Eq.~(30) where the matrix $\bm{G}$ incorrectly multiplies the transposed vector of means $\bar{\bm{z}}^\dagger$ instead of the vector of means form the $\bar{\bm{z}}$. Finally, note the formal similarity between in Eq.~\eqref{eq:M} and Eq.~\eqref{eq:q_d}.
           
        The \emph{cumulant} generating function is simply $\mathcal{K}(\bm{t}) = \log \mathcal{M}(\bm{t})$. This quantity can be used to extract the cumulants of the photon-number distribution, which encode the genuine correlations present in the probability distribution~\cite{ursell_1927,fisher1932derivation}.
       Moments and cumulants of the photon-number distribution can be obtained by taking derivatives of their respective generating function,
       \begin{align}
       	\braket{n_1^{p_1} \ldots n_M^{p_M}} &= \left. \frac{\partial^{p_1}}{\partial t_1^{p_1} } \ldots \frac{\partial^{p_M}}{\partial t_M^{p_M} } \mathcal{M}(\bm{t}) \right|_{\bm{t}=\bm{0}},\\
       	\braket{\braket{n_1^{p_1} \ldots n_M^{p_M}}} &= \left. \frac{\partial^{p_1}}{\partial t_1^{p_1} } \ldots \frac{\partial^{p_M}}{\partial t_M^{p_M} } \mathcal{K}(\bm{t}) \right|_{\bm{t}=\bm{0}}.
       \end{align}
   For the case where the $p_i$ are bitstrings, and thus the photon-number random variables appear only once inside the brackets or double brackets, we can approximate $e^t-1\approx t$ (since quadratic and higher-order terms are only relevant if any of the photon numbers is raised to a power higher than 1) and replace $\bm{G}(\bm{t}) \to \bm{D}(\bm{t})$, thus showing that the calculation of moments reduces to that of loop-Hafnians. Similarly, for cumulants the relevant graph-theoretic quantity is called the loop-Montrealer~\cite{cardin2022photon}, for which one can apply the same finite-difference sieve described before, but using only the $j=1$ term of the expansion in Eq.~\eqref{eq:truncated_series_q}, namely
   \begin{align}
   	 &\tilde{q}_N(\bm{\Sigma}^{(s=1)}, \bar{\bm{z}}^*, \eta) \nonumber\\
   	&=\sum_{k=1}^N \left(\frac{\tr\left([\bm{\Sigma}^{(s=1)}]^k\right)}{2k} + \frac{\bar{\bm{z}}^\mathrm{\dagger} [\bm{\Sigma}^{(s=1)}]^{k-1}  \bar{\bm{z}} }{2} \right)\eta^k .
   \end{align}
   
   Perhaps more interestingly, one can define coarse grained moments and cumulants by introducing a set partition $\Lambda=\{\Lambda_1,\dots,\Lambda_L\}$ of $\{1,\dots,M\}$, and writing 
   \begin{align}\label{eq:CGmom}
   	\braket{n_{\Lambda_1}\ldots n_{\Lambda_L}} &= \left\langle\left(\sum_{i_1 \in \Lambda_1} n_{i_1} \right)\ldots \left( \sum_{i_L \in \Lambda_L} n_{i_L}\right)\right\rangle,\\
   	   	\braket{\braket{n_{\Lambda_1}\ldots n_{\Lambda_L}}} &= \left\langle \left\langle\left(\sum_{i_1 \in \Lambda_1} n_{i_1} \right)\ldots \left( \sum_{i_L \in \Lambda_L} n_{i_L}\right)\right\rangle \right\rangle\label{eq:CGmom2}.
   \end{align}
   These coarse-grained quantities can be obtained by applying a finite-difference sieve using the same variable $t_{\Lambda_j}$ for all the indices in a given block $\Lambda_j$, which greatly simplifies the calculation of these quantities relative to the combinatorial explosion entailed in expanding the products in the right-hand side of the equations above. 

   For the case of Gaussian states with a given number of internal modes per external mode, we can find analogous relations to Eqs.~\eqref{eq:CGmom} and~\eqref{eq:CGmom2} following a similar procedure as the one we used to obtain the photon-number distribution in Eq.~\eqref{eq:prob_internal_modes}.  
   
   
    \section{\label{sec:off_diagonal}Off-diagonal density matrix elements}

        We will now extend our study on the loop and blocked loop Hafnians to the computation of the off-diagonal density matrix elements of the, generally non-Gaussian, states that are obtained when one measures a subset of modes of a multimode Gaussian state~\cite{su2019conversion, sabapathy2019production, quesada2019simulating}. First, we will suppose that the generated state is conditioned to a ``fine-grained'' measurement outcome. Then, we will consider the case where several of the measured modes are grouped together.

        \subsection{\label{sec:off_diag_no_internal}Ungrouped measurement case}
        
            For a Gaussian state, $\hat{\varrho}$, its off-diagonal density matrix elements can be shown~\cite{quesada2019simulating} to be given by
            \begin{align}
        		\braket{\ml | \hat{\varrho} | \nl} &= \frac{\Pr(\bm{0} | \bm{A}, \gl)}{\prod_{k=1}^M \sqrt{n_k! m_k!}} \lhaf(\bm{A}_{\nl \oplus \ml}, \gl_{\nl \oplus \ml}),
        		\label{eq:lhaf_offd}
        	\end{align}
            where $\nl = (n_1,\dots, n_M)$, $\ml = (m_1,\dots, m_M)$ are strings of non-negative integers. 
            
            This expression is similar to Eq.~\eqref{eq:gbs_probability_distribution}, however we cannot readily apply the methods we developed in Sec.~\ref{sec:trace_formulas} to compute the corresponding loop Hafnian. This is because terms of the form $\lhaf(\bm{A}_{\nl \oplus \ml}, \gl_{\nl \oplus \ml})$ with $\ml\neq\nl$ do not appear in the loop Hafnian master theorem, Eq.~\eqref{eq:lhaf_master}. Nevertheless, we can use the invariance of the loop Hafnian under permutations in order to avoid this limitation. 
            
            Our strategy relies on the following proposition. Let $T=\sum_{k=1}^M(n_k + m_k)$. There exist a permutation matrix $\bm{P}$ of size $T\times T$, the same size as $\bm{A}_{\nl\oplus\ml}$; a string of non-negative integers $\bm{t}$, whose length we will specify later; and a symmetric matrix $\bm{A}'$, whose size we will also specify later, such that
            \begin{align}
                &\bm{P}\bm{A}_{\nl\oplus\ml}\bm{P}^{\mathrm{T}}=\bm{A}'_{\bm{t}\oplus\bm{t}},\,\text{if }T\text{ is even}\label{eq:permutation_trick_even}\\
                &\left[\bm{P}\bm{A}_{\nl\oplus\ml}\bm{P}^{\mathrm{T}}\right]\oplus(1)=\bm{A}'_{\bm{t}\oplus\bm{t}}\,\text{if }T\text{ is odd}.
                \label{eq:permutation_trick_odd}
            \end{align}
            
            In order to motivate these expressions, let us consider a simple example. Suppose that $M=2$, so $\bm{A}$ is a $4\times4$ matrix:
            \begin{equation}
                \bm{A} =
                \left(\! 
                \begin{array}{cc|cc}
                    A_{11}&A_{12}&A_{13}&A_{14}\\
                    A_{12}&A_{22}&A_{23}&A_{24}\\
                    \hline
                    A_{13}&A_{23}&A_{33}&A_{34}\\
                    A_{14}&A_{24}&A_{34}&A_{44}\\
                \end{array} 
                \!\right).
                \label{eq:example_matrix_1}
            \end{equation}
            Let $\nl=(0,2)$ and $\ml=(1,1)$. This implies that $\bm{A}_{\nl\oplus\ml}$ also has size $4\times4$, and reads
            \begin{equation}
                \bm{A}_{\nl\oplus\ml} =
                \left(\! 
                \begin{array}{cc|cc}
                    A_{22}&A_{22}&A_{23}&A_{24}\\
                    A_{22}&A_{22}&A_{23}&A_{24}\\
                    \hline
                    A_{23}&A_{23}&A_{33}&A_{34}\\
                    A_{24}&A_{24}&A_{34}&A_{44}\\
                \end{array} 
                \!\right).
                \label{eq:example_matrix_2}
            \end{equation}
            
            Consider now the following permutations of $\bm{A}_{\nl\oplus\ml}$.
            \begin{align}
                \left(\! 
                \begin{array}{cccc}
                    A_{22}&A_{22}&A_{23}&A_{24}\\
                    A_{22}&A_{22}&A_{23}&A_{24}\\
                    A_{23}&A_{23}&A_{33}&A_{34}\\
                    A_{24}&A_{24}&A_{34}&A_{44}\\
                \end{array} 
                \!\right)
                &\xrightarrow{\bm{P}_1}
                \left(\! 
                \begin{array}{cccc}
                    A_{22}&A_{22}&A_{24}&A_{23}\\
                    A_{22}&A_{22}&A_{24}&A_{23}\\
                    A_{23}&A_{23}&A_{34}&A_{33}\\
                    A_{24}&A_{24}&A_{44}&A_{34}\\
                \end{array} 
                \!\right)\nonumber\\
                &\mkern-80mu \xrightarrow{\bm{P}_2}
                \left(\! 
                \begin{array}{cccc}
                    \tikzmark{left1}{$A_{22}$}\tikzmark{right1}{$A_{22}$}&A_{22}&\tikzmark{left2}{$A_{24}$}\tikzmark{right2}{$A_{24}$}&A_{23}\\
                    A_{22}&A_{22}&A_{24}&A_{23}\\
                    \tikzmark{left3}{$A_{24}$}\tikzmark{right3}{$A_{24}$}&A_{24}&\tikzmark{left4}{$A_{44}$}\tikzmark{right4}{$A_{44}$}&A_{34}\\
                    A_{23}&A_{23}&A_{34}&A_{33}\\
                \end{array}
                \!\right).
                \Highlight[first]{1}
                \Highlight[first]{2}
                \Highlight[first]{3}
                \Highlight[first]{4}
                \label{eq:example_matrix_3}
            \end{align}
            As can be seen, $\bm{P}_1$ swaps the third and fourth columns, while $\bm{P}_2$ swaps the third a fourth rows.
            
            We now determine the photon detection pattern $\bm{t}$. First, let us define $\bar{\bm{t}}\equiv(\min(n_1,m_1),\dots, \min(n_M,m_M))$. In our example, this vector reads $\bar{\bm{t}}=(0,1)$. Next, we create a string with $\frac{1}{2}\sum_{k}|m_k-n_k|$ new modes, and we assign one photon to each of them. We then construct $\bm{t}$ by appending this string of ones to $\bar{\bm{t}}$. In our example, $\frac{1}{2}\sum_{k}|m_k-n_k|=1$, so we create only one new mode, and our final string takes the form $\bm{t}=(0,1,1)$.
        
            The logic behind the definition of $\bm{t}$ can be better understood by noticing that the highlighted sub-blocks in the second line of Eq.~\eqref{eq:example_matrix_3} correspond precisely to the different blocks of $A_{\bar{\bm{t}}\oplus\bar{\bm{t}}}$:
            \begin{equation}
                A_{\bar{\bm{t}}\oplus\bar{\bm{t}}} = 
                \left(\! 
                \begin{array}{c|c}
                    A_{22} & A_{24} \\
                    \hline
                    A_{24} & A_{44}
                \end{array} 
                \!\right).
                \label{eq:example_matrix_4}
            \end{equation}
            The string of ones attached to $\bar{\bm{t}}$ accounts for the remaining rows and columns of matrix $\bm{P}_2\bm{P}_1\bm{A}_{\nl\oplus\ml}\bm{P}_1^{\mathrm{T}}\bm{P}_2^{\mathrm{T}}$.
            
            We may now define the $6\times 6$ matrix 
            \begin{equation}
                \bm{A}' =
                \left(\! 
                \begin{array}{ccc|ccc}
                    \tikzmark{left1}{$A_{11}$}&A_{12}&B_1&\tikzmark{left2}{$A_{13}$}&A_{14}&B_2\\
                    A_{12}&\tikzmark{right1}{$A_{22}$}&A_{22}&A_{23}&\tikzmark{right2}{$A_{24}$}&A_{23}\\
                    B_1&A_{22}&A_{22}&B_2&A_{24}&A_{23}\\
                    \hline
                    \tikzmark{left3}{$A_{13}$}&A_{23}&B_2&\tikzmark{left4}{$A_{33}$}&A_{34}&B_3\\
                    A_{14}&\tikzmark{right3}{$A_{24}$}&A_{24}&A_{34}&\tikzmark{right4}{$A_{44}$}&A_{34}\\
                    B_2&A_{23}&A_{23}&B_3&A_{34}&A_{33}
                \end{array} 
                \!\right),
                \Highlight[second]{1}
                \Highlight[second]{2}
                \Highlight[second]{3}
                \Highlight[second]{4}
                \label{eq:example_matrix_5}
            \end{equation}
            where $B_1, B_2, B_3$ are arbitrary complex numbers. Notice that the highlighted sub-blocks of this matrix are precisely the blocks of matrix $\bm{A}$, defined in Eq.~\eqref{eq:example_matrix_1}. We can readily check that 
            \begin{equation}
                \bm{P}_2\bm{P}_1\bm{A}_{\nl\oplus\ml}\bm{P}_1^{\mathrm{T}}\bm{P}_2^{\mathrm{T}}=\bm{A}'_{\bm{t}\oplus\bm{t}}.
                \label{eq:example_reorganization}
            \end{equation}
        
            In a more general case, a $2M\times 2M$ adjacency matrix, $\bm{A}$, describing a Gaussian state can be written as a block matrix of the form~\cite{hamilton2017gaussian, kruse2019detailed} 
            \begingroup
            \renewcommand*{\arraystretch}{1.5}
            \begin{align}
                \bm{A} = 
                \left(\! 
                \begin{array}{c|c}
                    \bm{B}  &\bm{C}\\
                    \hline
                    \bm{C}^{\mathrm{T}}  &\bm{B}^*
                \end{array} 
                \!\right),
                \label{eq:adj_mat_block}
            \end{align}
            \endgroup
            where $\bm{B}$, $\bm{C}$ have size $M\times M$, and $\bm{B}$ is symmetric. The corresponding $\bm{A}_{\nl\oplus\ml}$ can also be written as a block matrix:
            \begingroup
            \renewcommand*{\arraystretch}{1.5}
            \begin{align}
                \bm{A}_{\nl\oplus\ml} = 
                \left(\! 
                \begin{array}{c|c}
                    \bm{B}_{\nl} & \bm{C}_{\nl,\ml}\\
                    \hline
                    [\bm{C}^{\mathrm{T}}]_{\ml,\nl} &\bm{B}^*_{\ml}
                \end{array} 
                \!\right).
                \label{eq:adj_mat_block_2}
            \end{align}
            \endgroup
            Here, $\bm{B}_{\bm{n}}$ has size $(\sum_kn_k)\times(\sum_kn_k)$, and it is constructed from $\bm{B}$ by repeating its $k$th row and column $n_k$ times. On the other hand, $\bm{C}_{\nl,\ml}$ has size $(\sum_kn_k)\times(\sum_km_k)$ and it is constructed from $\bm{C}$ by repeating its $k$th row $n_k$ times, and its $k$th column $m_k$ times. $\bm{B}_{\bm{m}}^*$ and $[\bm{C}^{\mathrm{T}}]_{\nl,\ml}$ have analogous definitions.
        
            By setting $\bar{\bm{t}}\equiv(\min(n_1,m_1),\dots,\min(n_M,m_M))$, we can see that we can always find a permutation matrix $\bm{P}$ such that
            \begin{align}
                \bm{P}\bm{A}_{\nl\oplus\ml}\bm{P}^{\mathrm{T}} = 
                \left(\!
                \begin{array}{cc|cc}
                    \quad\tikzmark{left1}{$\bm{B}_{\bar{\bm{t}}}$} \tikzmark{right1}{$\bm{B}_{\bar{\bm{t}}}$}& \quad\vdots & \quad\tikzmark{left2}{$\bm{C}_{\bar{\bm{t}}}$} \tikzmark{right2}{$\bm{C}_{\bar{\bm{t}}}$} & \quad \vdots\\
                    \quad\cdots & \quad\ddots & \quad\cdots & \quad\ddots \\
                    \hline
                    \quad\tikzmark{left3}{$\bm{C}^{\mathrm{T}}_{\bar{\bm{t}}}$} \tikzmark{right3}{$\bm{C}^{\mathrm{T}}_{\bar{\bm{t}}}$} & \quad\vdots & \quad \tikzmark{left4}{$\bm{B}^*_{\bar{\bm{t}}}$} \tikzmark{right4}{$\bm{B}^*_{\bar{\bm{t}}}$} & \quad \vdots\\
                    \quad\cdots & \quad\ddots & \quad\cdots & \quad\ddots \\
                \end{array}
                \!\right).
                \Highlightt[third]{1}
                \Highlightt[third]{2}
                \Highlightt[third]{3}
                \Highlightt[third]{4}
                \label{eq:adj_mat_block_3}
            \end{align}
            If $T=\sum_k(n_k+m_k)$ is even, the remaining $\sum_k|m_k-n_k|$ rows and columns (an even number of them) can be distributed symmetrically, so as to make the four blocks of $\bm{P}\bm{A}_{\nl\oplus\ml}\bm{P}^{\mathrm{T}}$ have equal size. If $T$ is odd, so will be $\sum_k|m_k-n_k|$, and we must add one row and one column to $\bm{P}\bm{A}_{\nl\oplus\ml}\bm{P}^{\mathrm{T}}$ in order to obtain a block matrix with four blocks of equal size. For reasons that will become clear shortly, we will add a row and a column whose last entries are equal to $1$, while the rest of them are equal to $0$, that is, we will write $[\bm{P}\bm{A}_{\nl\oplus\ml}\bm{P}^{\mathrm{T}}]\oplus(1)$. 
            
            We can now define a new detection pattern $\bm{t}$, as
            \begin{align}
                \bm{t} = \bar{\bm{t}}\oplus \left(1\right)^{\oplus\frac{1}{2}\sum_k|m_k-n_k|},\text{ if }T\text{ is even},\label{eq:new_pattern_def_even}\\
                \bm{t} = \bar{\bm{t}}\oplus \left(1\right)^{\oplus\frac{1}{2}\sum_k|m_k-n_k|+\frac{1}{2}},\text{ if }T\text{ is odd},
                \label{eq:new_pattern_def_odd}
            \end{align}
            and a $2M'\times 2M'$ symmetric matrix $\bm{A}'$ (with $2M'= 2M+\sum_k|m_k-n_k|$ for even $T$, or $2M'= 2M+\sum_k|m_k-n_k|+1$ for odd $T$) as
            \begin{align}
                \bm{A}' &={\renewcommand*{\arraystretch}{1.5}
                \left(\! 
                \begin{array}{c|c}
                    \bm{B}'  &\bm{C}'\\
                    \hline
                    (\bm{C}')^{\mathrm{T}}  &\bm{D'}
                \end{array} 
                \!\right)}\nonumber\\
                &=
                \left(\!
                \begin{array}{cc|cc}
                    \quad\tikzmark{left1}{$\bm{B}$} \tikzmark{right1}{$\bm{B}$}& \quad\vdots & \quad\tikzmark{left2}{$\bm{C}$} \tikzmark{right2}{$\bm{C}$} & \quad \vdots\\
                    \quad\cdots & \quad\ddots & \quad\cdots & \quad\ddots \\
                    \hline
                    \quad\tikzmark{left3}{$\bm{C}^{\mathrm{T}}$} \tikzmark{right3}{$\bm{C}^{\mathrm{T}}$} & \quad\vdots & \quad \tikzmark{left4}{$\bm{B}^*$} \tikzmark{right4}{$\bm{B}^*$} & \quad \vdots\\
                    \quad\cdots & \quad\ddots & \quad\cdots & \quad\ddots \\
                \end{array}
                \!\right),
                \Highlightt[fourth]{1}
                \Highlightt[fourth]{2}
                \Highlightt[fourth]{3}
                \Highlightt[fourth]{4}
                \label{eq:adj_mat_block_4}
            \end{align}
            where $\bm{B}'$, $\bm{C}'$, $\bm{D}'$ have size $M'\times M'$, and $\bm{B}'$, $\bm{D}'$ are symmetric. 
            
            Eqs.~\eqref{eq:permutation_trick_even} and~\eqref{eq:permutation_trick_odd} follow from choosing the elements in the rows and columns that complement the highlighted sub-blocks in Eq.~\eqref{eq:adj_mat_block_4} according to the structure of $\bm{P}\bm{A}_{\nl\oplus\ml}\bm{P}^{\mathrm{T}}$ or $[\bm{P}\bm{A}_{\nl\oplus\ml}\bm{P}^{\mathrm{T}}]\oplus(1)$, just as we did in the illustrative example. If some component of $\bar{\bm{t}}$ is equal to zero, we can fill the corresponding elements of $\bm{A}'$ in an arbitrary way, keeping in mind the symmetry of the whole matrix.
    
            Using a similar argument, we can show that there is a loop vector $\gl'$, of length $2M'$, such that 
            \begin{align}
                \bm{P}\gl_{\nl\oplus\ml}=\gl'_{\bm{t}\oplus\bm{t}}\text{  or  }[\bm{P}\gl_{\nl\oplus\ml}]\oplus(1)=\gl'_{\bm{t}\oplus\bm{t}}
                \label{eq:permutation_trick_loop_vector}
            \end{align}
            This equation, along with Eqs.~\eqref{eq:permutation_trick_even} and~\eqref{eq:permutation_trick_odd}, allow us to use the invariance of the loop Hafnian under permutations, Eq.~\eqref{eq:lhafP}, to see that
            \begin{align}
                \lhaf(\bm{A}_{\nl \oplus \ml}, \gl_{\nl \oplus \ml}) &= \lhaf(\bm{P}\bm{A}_{\nl \oplus \ml}\bm{P}^{\mathrm{T}}, \bm{P}\gl_{\nl \oplus \ml})\nonumber\\
                &=\lhaf(\bm{A}'_{\bm{t}\oplus\bm{t}},\gl'_{\bm{t}\oplus\bm{t}}).
                \label{eq:off_diag_lhaf_new_even}
            \end{align}
            for $T$ even, while
            \begin{align}
                \lhaf&(\bm{A}_{\nl \oplus \ml}, \gl_{\nl \oplus \ml}) = \lhaf(\bm{P}\bm{A}_{\nl \oplus \ml}\bm{P}^{\mathrm{T}}, \bm{P}\gl_{\nl \oplus \ml})\nonumber\\
                &= \lhaf\left([\bm{P}\bm{A}_{\nl \oplus \ml}\bm{P}^{\mathrm{T}}]\oplus(1), [\bm{P}\gl_{\nl \oplus \ml}]\oplus(1)\right)\nonumber\\
                &=\lhaf(\bm{A}'_{\bm{t}\oplus\bm{t}},\gl'_{\bm{t}\oplus\bm{t}})
                \label{eq:off_diag_lhaf_new_odd}
            \end{align}
            for $T$ odd. Here, we used the fact that $\lhaf[\bm{A}\oplus(1)]=\lhaf(\bm{A})\lhaf[(1)]=\lhaf(\bm{A})$.
            
            We have succeeded in expressing $\lhaf(\bm{A}_{\nl \oplus \ml}, \gl_{\nl \oplus \ml})$ in a form that is compatible with the loop Hafnian master theorem. However, we note that $\bm{A}'$ and $\gl'$ do not generally represent a valid Gaussian state, i.e., they do not have the form of Eqs.~\eqref{eq:defA} and~\eqref{eq:defgamma}. Since our  loop Hafnian master theorem derivation in Sec.~\ref{sec:master_theorem} used $\bm{A}$ and $\gl$ from valid Gaussian states, it is not clear immediately that our derivation is valid in this case.
            In a previous draft of this work, we had motivated that the loop Hafnian master theorem should hold in this more general setting, but had stated this as an assumption.
            Following the release of a preprint of this work, a proof which showed that assumption was correct and that our loop Hafnian master theorem is valid for any symmetric matrix $\bm{A}$ and any vector $\gl$ has been given~\cite{tarasov2025derivation}.
    
            Let $|\bm{t}|=\sum_{k=1}^{M'} t_k$ (i.e., $|\bm{t}|=\frac{1}{2}\sum_{k=1}^M(n_k+m_k)$ if $T$ is even, or $|\bm{t}|=\frac{1}{2}\sum_{k=1}^M(n_k+m_k)+\frac{1}{2}$ if $T$ is odd). Then, following Eq.~\eqref{eq:finite_diff_lhaf}, we can write
            \begin{align}
                \lhaf&(\bm{A}_{\nl \oplus \ml}, \gl_{\nl \oplus \ml})=\left( \prod_{k=1}^{M'} D_{z_k}^{(t_k)} \right) f_{|\bm{t}|}(\bm{A}', \gl', \zl),
                \label{eq:off_diag_lhaf_final}
            \end{align}
            where it is understood that the $\zl$ in this expression has length $M'$. 

            Let us now use Eq.~\eqref{eq:off_diag_lhaf_final} to compute the matrix elements of a state that is prepared by the partial measurement of a multimode Gaussian state, $\hat{\varrho}$, parametrized by $\bm{A}$ and $\gl$. 
            
            Suppose that we divide the set of modes $[M]$ into two subsets $[M] = \mathcal{G}\cup \mathcal{H}$, such that $\mathcal{G}\cap \mathcal{H}=\emptyset$. The subset $\mathcal{H}$ will correspond to the modes that we will measure in order to generate the state, i.e., the \textit{heralding modes}. Our generated or \textit{heralded} state will be supported by the modes in $\mathcal{G}$. Without loss of generality, let us assume that $\mathcal{H}$ corresponds to the first $H$ modes of the system, i.e., $\mathcal{H}=\{1,\dots,H\}$, while $\mathcal{G}$ will be associated to the remaining $M-H$ modes, i.e., $\mathcal{G}=\{H+1,\dots, M\}$.

            The unnormalized state $\hat{\varrho}_{\mathcal{G}}$ obtained after the measurement of the photon number outcome $\tilde{\bm{n}}=(\tilde{n}_1,\dots,\tilde{n}_H)$ in the heralding modes is given by~\cite{quesada2019simulating}
            \begin{equation}
                \hat{\varrho}_{\mathcal{G}} = \sum_{\bm{u},\bm{v}}\langle \tilde{\bm{n}},\bm{v}|\hat{\varrho}|\tilde{\bm{n}},\bm{u}\rangle |\bm{v}\rangle\langle\bm{u}|,
                \label{eq:heralded_state_no_internal}
            \end{equation}
            where $|\bm{u}\rangle=|u_1,\dots,u_{M-H}\rangle$, $|\bm{v}\rangle=|v_1,\dots,v_{M-H}\rangle$ (with $u_k$,$v_k$ non-negative integers for every $k$) are Fock basis elements associated to the modes in $\mathcal{G}$, and $|\tilde{\bm{n}},\bm{u}\rangle=|\tilde{n}_1,\dots,\tilde{n}_H,u_{1},\dots,u_{M-H}\rangle$.

            According to Eq.~\eqref{eq:lhaf_offd}, we can compute the different matrix elements defining $\hat{\varrho}_{\mathcal{G}}$ as 
            \begin{align}
        		\langle\bm{v}&|\hat{\varrho}_{\mathcal{G}}|\bm{u} \rangle = \langle\tilde{\nl},\bm{v}|\hat{\varrho}|\tilde{\nl},\bm{u}\rangle \nonumber\\
                &=\Pr(\bm{0} | \bm{A}, \gl)\left(\prod_{k=1}^H\frac{1}{\tilde{n}_k!}\right)\left(\prod_{k=1}^{M-H}\frac{1}{ \sqrt{u_k! v_k!}}\right) \nonumber\\
                &\quad\quad\quad\times\lhaf\left[\bm{A}_{(\tilde{\nl}\oplus\bm{u}) \oplus (\tilde{\nl}\oplus\bm{v})}, \gl_{(\tilde{\nl}\oplus\bm{u}) \oplus (\tilde{\nl}\oplus\bm{v})}\right].
        		\label{eq:lhaf_heralded_offd_1}
        	\end{align}

            When computing the loop Hafnian in this relation, we need to define the corresponding $\bm{t}$, $\bm{A}'$ and $\gl'$ using, in principle, the complete photon number patterns $\tilde{\nl}\oplus\bm{u}$ and $\tilde{\nl}\oplus\bm{v}$. However, for any pair $\bm{u}$, $\bm{v}$, the integer string $\bm{t}$ will take the form 
            $\bm{t} = \tilde{\nl}\oplus\bm{r}$, where $\bm{r}=(r_1,\dots,r_{M'-H})$ is computed following the prescription of Eqs.~\eqref{eq:new_pattern_def_even} or~\eqref{eq:new_pattern_def_odd}, but using only $\bm{u}$ and $\bm{v}$ (instead of $\bm{m}$ and $\bm{n}$). This also means that the variations in the structures of $\bm{A}'$ and $\gl'$ will depend only on $\bm{u}$ and $\bm{v}$.

            On account of these considerations, and combining Eqs.~\eqref{eq:off_diag_lhaf_final} and~\eqref{eq:lhaf_heralded_offd_1}, we may write the matrix elements of the heralded state as
            \begin{align}
        		&\langle\bm{v}|\hat{\varrho}_{\mathcal{G}}|\bm{u} \rangle=\Pr(\bm{0} | \bm{A}, \gl)\left(\prod_{k=1}^H\frac{1}{\tilde{n}_k!}\right)\left(\prod_{k=1}^{M-H}\frac{1}{ \sqrt{u_k! v_k!}}\right) \nonumber\\
                &\times\left( \prod_{k=1}^{H} D_{z_k}^{(\tilde{n}_k)} \right) \left( \prod_{k=H+1}^{M'} D_{z_k}^{(r_{k-H})} \right)f_{\tilde{N}+|\bm{r}|}(\bm{A}', \gl', \zl),
        		\label{eq:lhaf_heralded_offd_2}
        	\end{align}
            where $\tilde{N}=\sum_{k=1}^H\tilde{n}_k$, $|\bm{r}|=\sum_{k=1}^{M'-H}r_k$, and $M'= M+\frac{1}{2}\sum_k|v_k-u_k|$ for even $\sum_k(v_k+u_k)$, or $M'= M+\frac{1}{2}\sum_k|v_k-u_k|+\frac{1}{2}$ for odd $\sum_k(v_k+u_k)$. Here, it is understood that $\zl$ has length $M'$, $\gl'$ has length $2M'$ and $\bm{A}'$ has size $2M'\times2M'$.  
           
        \subsection{\label{sec:off_diag_internal}Grouped measurement case}

            Suppose now that we group the heralding modes according to a set partition $\tilde{\Lambda}=\{\tilde{\Lambda}_1,\dots,\tilde{\Lambda}_L\}$ of $\mathcal{H}$, and that we condition the heralded state to the measurement of the coarse-grained detection pattern $\tilde{\bm{b}}=(\tilde{b}_1,\dots,\tilde{b}_L)$. This detection event is represented by the operator $\hat{\Pi}(\tilde{\bm{b}})$, which can be computed by summing over all the projectors $\{|\tilde{\bm{n}}\rangle\langle\tilde{\bm{n}}|\}$ whose corresponding $\tilde{\bm{n}}$ are compatible with $\tilde{\bm{b}}$, in the sense that $\tilde{b}_j=\sum_{i\in \tilde{\Lambda}_j}\tilde{n}_{i}$:
            \begin{equation}
                \hat{\Pi}(\tilde{\bm{b}})=\sum_{\tilde{\bm{n}}\in\mathcal{K}_{\tilde{\Lambda}}^{\tilde{\bm{b}}}}|\tilde{\bm{n}}\rangle\langle\tilde{\bm{n}}|.
                \label{eq:projector_onto_heralding_outcome}
            \end{equation}
            Here, the set $\mathcal{K}_{\tilde{\Lambda}}^{\tilde{\bm{b}}}$ is defined according to Eq.~\eqref{eq:coarse_grained_subset}.

            Using the above definitions, we can write the unnormalized heralded state, $\hat{\varrho}_{\mathcal{G}(\tilde{\Lambda})}$, as
            \begin{align}
                \hat{\varrho}_{\mathcal{G}(\tilde{\Lambda})} &= \mathrm{Tr}_{\mathcal{H}}\left[\left(\hat{\Pi}(\tilde{\bm{b}})\otimes\hat{\mathbb{I}}_{\mathcal{G}}\right) \hat{\varrho}\right]\nonumber\\
                &=\sum_{{\bm{u}},{\bm{v}}}\left(\sum_{\tilde{\bm{n}}\in\mathcal{K}_{\tilde{\Lambda}}^{\tilde{\bm{b}}}}\langle \tilde{\bm{n}},{\bm{v}}|\hat{\varrho}|\tilde{\bm{n}},{\bm{u}}\rangle \right)|{\bm{v}}\rangle\langle{\bm{u}}|,
                \label{eq:heralded_state_internal}
            \end{align}
            where the partial trace is taken over all the heralding modes, and $\hat{\mathbb{I}}_{\mathcal{G}}$ stands for the identity operator over the Hilbert space corresponding to $\mathcal{G}$. According to Eq.~\eqref{eq:lhaf_offd}, the density matrix elements defining this state can be computed as
            \begin{align}
        		\langle\bm{v}&|\hat{\varrho}_{\mathcal{G}(\tilde{\Lambda})}|\bm{u} \rangle = \sum_{\tilde{\bm{n}}\in\mathcal{K}_{\tilde{\Lambda}}^{\tilde{\bm{b}}}}\langle \tilde{\bm{n}},{\bm{v}}|\hat{\varrho}|\tilde{\bm{n}},{\bm{u}}\rangle \nonumber\\
                &=\Pr(\bm{0} | \bm{A}, \gl)\left(\prod_{k=1}^{M-H}\frac{1}{ \sqrt{u_k! v_k!}}\right) \nonumber\\
                &\;\times\sum_{\tilde{\bm{n}}\in\mathcal{K}_{\tilde{\Lambda}}^{\tilde{\bm{b}}}}\frac{\lhaf\left[\bm{A}_{(\tilde{\nl}\oplus\bm{u}) \oplus (\tilde{\nl}\oplus\bm{v})}, \gl_{(\tilde{\nl}\oplus\bm{u}) \oplus (\tilde{\nl}\oplus\bm{v})}\right]}{\prod_{j=1}^L\prod_{i\in\tilde{\Lambda}_j}\tilde{n}_i!}.
        		\label{eq:lhaf_heralded_offd_internal_1}
        	\end{align}

            To solve the problem of dealing with off-diagonal elements, we can define from $\bm{u}$, $\bm{v}$, just as we did before, a vector $\bm{t}=\tilde{\nl}\oplus\bm{r}$, an adjacency matrix $\bm{A}'$, and a loop vector $\gl'$ such that
            \begin{align}
                \lhaf&\left[\bm{A}_{(\tilde{\nl}\oplus\bm{u}) \oplus (\tilde{\nl}\oplus\bm{v})}, \gl_{(\tilde{\nl}\oplus\bm{u}) \oplus (\tilde{\nl}\oplus\bm{v})}\right] \nonumber\\
                & = \lhaf\left[\bm{A}'_{(\tilde{\nl}\oplus\bm{r}) \oplus (\tilde{\nl}\oplus\bm{r})}, \gl'_{(\tilde{\nl}\oplus\bm{r}) \oplus (\tilde{\nl}\oplus\bm{r})}\right].
                \label{eq:equality_lhaf_internal}
            \end{align}
            The dimensions of these quantities are the same as those defined below Eq.~\eqref{eq:lhaf_heralded_offd_2}.

            Another problem arises when considering the sum over the elements in $\mathcal{K}_{\tilde{\Lambda}}^{\tilde{\bm{b}}}$. Indeed, this sum keeps us from directly applying the finite-difference sieve formula for the loop Hafnian in order to obtain a result similar to Eq.~\eqref{eq:lhaf_heralded_offd_2}. Moreover, since we are grouping together only a subset of modes, while we are letting the rest of them loose, this sum also keeps us from directly writing $\langle\bm{v}|\hat{\varrho}_{\mathcal{G}(\tilde{\Lambda})}|\bm{u} \rangle$ in terms of a blocked loop Hafnian. Luckily, there is a quick workaround to address this second issue. 

            Let us extend the partition $\tilde{\Lambda}$ in order to include the remaining modes (taking into account the new modes defined along $\bm{r}$, $\bm{A}'$ and $\gl'$) in a way that captures the idea of them being ``loose''. Note that $\{\{H+1\},\dots,\{M'\}\}$ is a set partition of $\{H+1,\dots,M'\}$. This implies that $\Lambda'=\tilde{\Lambda}\cup\{\{H+1\},\dots,\{M'\}\}$ will be a set partition of $\mathcal{H}\cup \{H+1,\dots,M'\}$ (since $\mathcal{H}\cap \{H+1,\dots,M'\}=\emptyset$). By defining a modified coarse-grained detection pattern $\bm{b}'=\tilde{\bm{b}}\oplus\bm{r}$, we can see that 
            \begin{equation}
                \mathcal{K}_{\Lambda'}^{\bm{b}'}=\left\{\nl'\,\Big|\,\nl'=\tilde{\nl}\oplus\bm{r}\,\text{ and }\,\tilde{\nl}\in\mathcal{K}_{\tilde{\Lambda}}^{\tilde{\bm{b}}}\right\},
                \label{eq:compatible_set_extended_partition}
            \end{equation}
            which, in turn, allows us to write
            \begin{align}
                &\sum_{\tilde{\bm{n}}\in\mathcal{K}_{\tilde{\Lambda}}^{\tilde{\bm{b}}}}\frac{\lhaf\left[\bm{A}_{(\tilde{\nl}\oplus\bm{u}) \oplus (\tilde{\nl}\oplus\bm{v})}, \gl_{(\tilde{\nl}\oplus\bm{u}) \oplus (\tilde{\nl}\oplus\bm{v})}\right]}{\prod_{j=1}^L\prod_{i\in\tilde{\Lambda}_j}\tilde{n}_i!}\nonumber\\
                &=\left(\prod_{j=1}^{M'-H}r_j!\right)\sum_{\nl'\in\mathcal{K}_{\Lambda'}^{\bm{b}'}}\frac{\lhaf\left[\bm{A}'_{\nl' \oplus \nl'}, \gl'_{\nl' \oplus \nl'}\right]}{\prod_{j=1}^{L+M'-H}\prod_{i\in\Lambda'_j}n'_i!},
                \label{eq:sum_lhaf_extended_equality}
            \end{align}
            where we took into account that 
            \begin{equation}
                \prod_{j=1}^{L+M'-H}\prod_{i\in\Lambda'_j}n'_i! = \left(\prod_{j=1}^L\prod_{i\in\tilde{\Lambda}_j}\tilde{n}_i!\right)\left(\prod_{j=1}^{M'-H}r_j!\right).
                \label{eq:extended_partition_product}
            \end{equation}

            Recalling the definition of the $\Lambda$-blocked loop Hafnian, Eq.~\eqref{eq:Lamb_lhaf_def}, we can recast Eq.~\eqref{eq:lhaf_heralded_offd_internal_1} as
            \begin{align}
        		\langle\bm{v}|\hat{\varrho}_{\mathcal{G}(\tilde{\Lambda})}|\bm{u}\rangle = \Pr(\bm{0} &| \bm{A}, \gl)\left(\prod_{k=1}^{M-H}\frac{1}{\sqrt{u_k! v_k!}}\right) \left(\prod_{k=1}^{L}\frac{1}{\tilde{b}_k!}\right)\nonumber\\
                &\times \lhaf_{\Lambda'}\left(\bm{A}',\gl',\tilde{\bm{b}}\oplus\bm{r}\right),
        		\label{eq:lhaf_heralded_offd_internal_2}
        	\end{align}
            where we noticed that 
            \begin{equation}
                \prod_{k=1}^{L+M'-H}b'_k!=\prod_{k=1}^L\tilde{b}_k!\prod_{k=1}^{M'-H}r_k!.
                \label{eq:another_extended_partition_product}    
            \end{equation}

            We may now apply the finite-difference sieve formula for the blocked loop Hafnian and write the following final expression for $\langle\bm{v}|\hat{\varrho}_{\mathcal{G}(\tilde{\Lambda})}|\bm{u}\rangle$:
            \begin{align}
        		&\langle\bm{v}|\hat{\varrho}_{\mathcal{G}(\tilde{\Lambda})}|\bm{u}\rangle = \Pr(\bm{0}| \bm{A}, \gl)\left(\prod_{k=1}^{L}\frac{1}{\tilde{b}_k!}\right)\left(\prod_{k=1}^{M-H}\frac{1}{\sqrt{u_k! v_k!}}\right)\nonumber\\
                &\times \left( \prod_{k=1}^{L} D_{w'_k}^{(\tilde{b}_k)} \right) \left( \prod_{k=L+1}^{L+M'-H} D_{w'_k}^{(r_{k-L})} \right)f_{\tilde{N}+|\bm{r}|}(\bm{A}', \gl', \bm{w}'),
        		\label{eq:lhaf_heralded_offd_internal_3}
        	\end{align}
            with $\tilde{N}=\sum_{j=1}^L\tilde{b}_j=\sum_{k=1}^H\tilde{n}_k$. Here, it is understood that $\bm{w}'$ has length $M'$, with $M'$ as defined below Eq.~\eqref{eq:lhaf_heralded_offd_2}, and that its components are computed according to the partition $\Lambda'$, following the prescription given below Eq~\eqref{eq:coarse_grained_subset}. 

            To conclude this section, let us write Eq.~\eqref{eq:lhaf_heralded_offd_internal_3} for the special case of a system with $M$ external modes and $K$ internal modes per external mode (like that shown in Fig.~\ref{fig:internal_modes_heralding}). Recall that a Gaussian state representing this system is parametrized by a $2MK\times2MK$ matrix $\bar{\bm{A}}$ and a $2MK$-length vector $\bar{\gl}$. 
            
            As before, consider that we have photon-number-resolving detectors that can only distinguish between the different external modes. We divide these external modes in two subsets
            $\mathcal{H}=\{1,\dots,H\}$ and $\mathcal{G}=\{H+1,\dots,M\}$. Conditioned to the measurement of the outcome $\tilde{\nl}=(\tilde{n}_1,\dots, \tilde{n}_H)$ in the external heralding modes (i.e., the external modes in $\mathcal{H}$), we will obtain an unnormalized heralded state, $\hat{\varrho}_{\mathcal{G}(K)}$, that will be supported by \textit{all the internal modes whose corresponding external modes are in $\mathcal{G}$}, i.e., the heralded state will be defined over a total of $(M-H)K$ modes. 
            
            Let $\tilde{\Lambda}(K)=\{\tilde{\Lambda}_1,\dots,\tilde{\Lambda}_H\}$ be a set partition of $[HK] = \{1,\dots,HK\}$, whose blocks are defined as $\tilde{\Lambda}_l = \{(l-1)K+1,\dots,lK\}$. Let $\bar{\bm{u}}=(\bar{u}_{H+1,1},\dots \bar{u}_{H+1, K})\oplus\cdots\oplus(\bar{u}_{M,1},\dots \bar{u}_{M, K})$ and $\bar{\bm{v}}=(\bar{v}_{H+1,1},\dots \bar{v}_{H+1, K})\oplus\cdots\oplus(\bar{v}_{M,1},\dots \bar{v}_{M, K})$. Following the prescription of Eqs.~\eqref{eq:new_pattern_def_even} or~\eqref{eq:new_pattern_def_odd}, using $\bar{\bm{u}}$ and $\bar{\bm{v}}$, we can define a new integer string $\bar{\bm{r}}=(\bar{r}_1,\dots,\bar{r}_{\bar{M}-HK})$, which, in turn, allows us to define an extended partition $\bar{\Lambda}(K)=\tilde{\Lambda}(K)\cup\{\{HK+1\},\dots,\{\bar{M}\}\}$, and an extended coarse-grained detection pattern $\tilde{\nl}\oplus\bar{\bm{r}}$. In these expressions, $\bar{M}= MK + \frac{1}{2}\sum_{k=H+1}^M\sum_{l=1}^K|\bar{v}_{k,l}-\bar{u}_{k,l}|$ if $\sum_{k=H+1}^M\sum_{l=1}^K(\bar{v}_{k,l}+\bar{u}_{k,l})$ is even, or $\bar{M}= MK + \frac{1}{2}\sum_{k=H+1}^M\sum_{l=1}^K|\bar{v}_{k,l}-\bar{u}_{k,l}|+\frac{1}{2}$ if $\sum_{k=H+1}^M\sum_{l=1}^K(\bar{v}_{k,l}+\bar{u}_{k,l})$ is odd. We also define the new $2\bar{M}\times2\bar{M}$ adjacency matrix $\bar{\bm{A}'}$, and $2\bar{M}$-length loop vector $\bar{\gl}'$ associated to $\bar{\bm{r}}$.

            With these definitions in place, we may write
            \begin{align}
        		&\frac{\langle\bar{\bm{v}}|\hat{\varrho}_{\mathcal{G}(K)}|\bar{\bm{u}}\rangle}{\Pr(\bm{0}| \bar{\bm{A}}, \bar{\gl})} = \left(\prod_{k=1}^{H}\frac{1}{\tilde{n}_k!}\right)\left(\prod_{k=H+1}^{M}\prod_{l=1}^K\frac{1}{\sqrt{\bar{u}_{k,l}! \bar{v}_{k,l}!}}\right)\nonumber\\
                &\times \left( \prod_{k=1}^{H} D_{w_k}^{(\tilde{n}_k)} \right) \left( \prod_{k=H+1}^{\bar{M}-H(K-1)} \!\!D_{w_k}^{(\bar{r}_{k-H})} \right)f_{\tilde{N}+|\bar{\bm{r}}|}(\bar{\bm{A}}', \bar{\gl}', \bar{\bm{w}}).
        		\label{eq:lhaf_heralded_offd_internal_special}
        	\end{align}
            In this relation $\bar{\bm{w}}$ has length $\bar{M}$, and it can be written as 
            \begin{align}
                \bar{\bm{w}}=(w_1,\dots,w_1)&\oplus\cdots\oplus(w_H,\dots,w_H)\nonumber\\
                &\oplus \left(w_{H+1},\dots, w_{\bar{M}-H(K-1)}\right),
                \label{eq:w_vec_internal}
            \end{align}
            where each repeated $w_l$ appears a total of $K$ times.
            
    \section{\label{sec:numerical}Numerical experiments}

        In this section we show some applications of the results we have presented so far. First, we study the generation of Fock states in the presence of internal modes. Then, we show an example of the heralding of approximate Gottesman-Kitaev-Preskill (GKP) states. Finally, we will compute the total photon number probability distributions for some Gaussian Boson Sampling (GBS) setups.

        \subsection{Internal mode Fock state generation}

            \begin{figure}[!t]
                {
                  \includegraphics[width=\columnwidth]{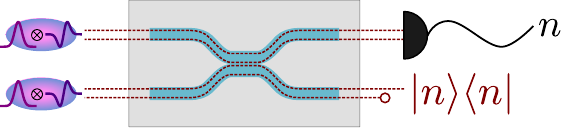}%
                }
                \caption{Setup for the generation of an arbitrary Fock state, $|n\rangle$. Two single-mode squeezed states are sent into a balanced beam-splitter. Upon the measurement of $n$ photons in one of the output \textit{ports} of the beam-splitter, the state $|n\rangle$ is heralded in remaining output port. Here, we consider that the input squeezed states are spectrally impure (as indicated by the multicolor ellipse representing the gaussian state), i.e., the light entering each port is a mixture of squeezed states in two mutually orthogonal spectral modes. As a consequence, our system has two internal (spectral) modes per external mode (ports of the beam-splitter). These internal modes are represented by dashed red lines. We filter out the second spectral mode in the heralding output port by tracing it out (as represented by the circle in the bottom red, dashed line).} 
                \label{fig:fock_herald_setup}
            \end{figure}

            \begin{figure*}[!t]
                {
                  \includegraphics[width=\textwidth]{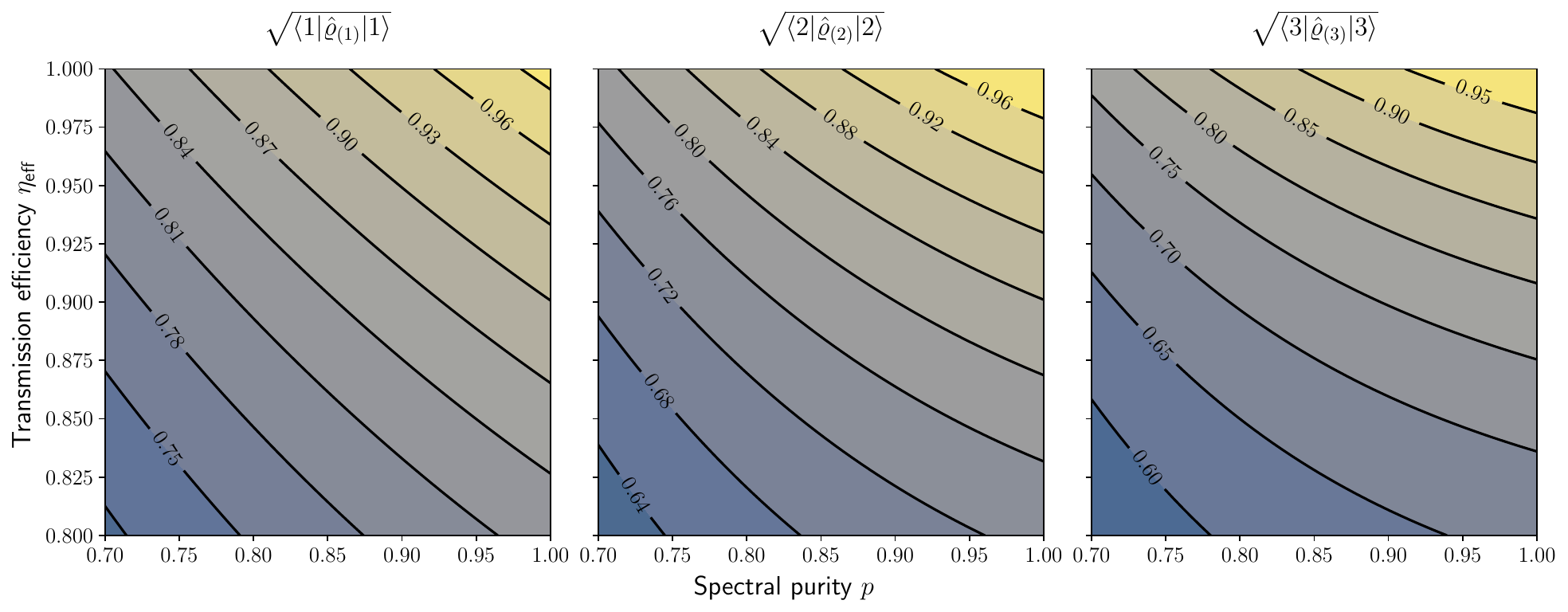}%
                }
                \caption{Fidelity between the noisy (including photon loss and spectral impurity) state, $\hat{\varrho}_{(n)}$, and the target state $|n\rangle\langle n|$, for $n=1,2,3$, as a function of the parameters $p$ and $\eta_{\text{eff}}$. The squeezing parameter of the input states was set to $\xi=1.1$. As can be seen in all cases, the fidelity decays more rapidly with increasing losses, while the generated states are more robust to spectral impurities. We can also see that the higher the photon number state, the more sensitive to both of these sources of noise.}
                \label{fig:fidelity_fock}
            \end{figure*}
            
            Suppose that we want to generate an arbitrary Fock state, $|n\rangle$, from the partial measurement of a two-mode squeezed state. In order to do this, we can use the following simple setup (see Fig.~\ref{fig:fock_herald_setup}). Two non-displaced single-mode squeezed states with the same real squeezing parameter, $\xi$, are sent into a balanced beam-splitter represented by the transmission matrix
            \begin{equation}
                \bm{U}=\frac{1}{\sqrt{2}}
                \begin{pmatrix}
                    1&i\\
                    i&1
                \end{pmatrix}.
                \label{eq:beam_splitter_unitary}
            \end{equation}
            The state at the output of this two-mode interferometer is the desired two-mode squeezed state. Upon the measurement of $n$ photons in \textit{any} of the output ports of the beam-splitter, the state $|n\rangle$ is heralded in remaining output port.

            Suppose now that the input single-mode squeezed states are not \textit{spectrally pure}, i.e., the state of the light in each input port of the beam-splitter corresponds to a mixture of squeezed states in different (mutually orthogonal) spectral modes. These spectral modes will correspond to the internal modes of the system. For simplicity, we will assume that there are only two of these modes. Furthermore, suppose that the light suffers from uniform losses during its transition through the beam-splitter. This implies that the corresponding transmission matrix will have the form $\bm{T}=\sqrt{\eta_{\text{eff}}}\bm{U}$, with $\eta_{\text{eff}}\in[0,1]$ a parameter describing the transmission efficiency. We will simulate how the spectral impurity of the input light, and the presence of losses affects the generation of the desired Fock state.

            We will model the spectral impurity of the input light by assuming that we have four input single-mode squeezed states, two for each spectral mode. The corresponding squeezing parameters are $(\xi_1, \xi_2, \xi_1',\xi_2')$. For this particular example, we have $\xi_1=\xi_2=\xi$, $\xi_1'=\xi_2'=\xi'$. The parameters $\{\xi_k'\}$ correspond to the second spectral mode, and the states with squeezing parameters $(\xi_k,\, \xi_k')$ will be associated to the same port (external mode) $k$. The different $\{\xi_k'\}$ are related to the $\{\xi_k\}$ according to the relation
            \begin{equation}
                \tanh^2(\xi_k')=\left(\frac{1-p}{p}\right)\tanh^2(\xi_k)
                \label{eq:relation_between_sqz_par},
            \end{equation}
            where $p\in(0, 1]$ is a parameter characterizing the spectral purity of the input light. We also assume that all $\xi_k$, $\xi_k'$ are positive.

            The rationale behind the definition of the $\{\xi_k'\}$ is the following. Suppose that $|\xi_k\rangle$, $|\xi_k'\rangle$ are two-mode squeezed states of the form
            \[|\xi_k\rangle=\sum_{n}\chi_{kn}|n\rangle_1|n\rangle_2\,,\;|\xi_k'\rangle=\sum_{m}\chi_{km}'|m\rangle_3|m\rangle_4,\] 
            with $\chi_{kn}=i^n\tanh^n(\xi_k)/\cosh(\xi_k)$ ($\chi_{km}'$ is similarly defined). The subscript $j$ in $|n\rangle_j$ labels one of the modes in which the squeezed states are defined. Upon the measurement of one photon in \textit{either} mode 2 \textit{or} mode 4, and after tracing out mode $3$, the state of the light in mode 1 takes the form $p|1\rangle\langle1| + (1-p)|0\rangle\langle0|$, with $p = \tanh^2(\xi_k) / [\tanh^2(\xi_k) + \tanh^2(\xi_k')]$. If we think of the state $|\xi_k\rangle|\xi_k'\rangle$ as a spectrally impure two-mode squeezed state, where the \textit{internal} (spectral) modes $(1,3)$ correspond to an \textit{external} mode $\bar{1}$, while $(2,4)$ correspond to a external mode $\bar{2}$, we can interpret $p$ as the probability of successfully heralding one photon in the external mode $\bar{1}$ and the spectral mode 1, when one photon is detected in the external mode $\bar{2}$. We can see that only when $\xi_k'=0$ (i.e. when we have a spectrally pure squeezed state) is $p=1.0$, while $p=0.5$ if $\xi_k=\xi_k'$.

            The covariance matrix of the input light can be written as a direct sum of the covariance matrices of the first and second spectral modes $\tilde{\bm{\Sigma}}_0 = \bm{\Sigma}_0\oplus\bm{\Sigma}_0'$. Since the two internal modes are mutually orthogonal, there will not be any interference between them after the propagation through the beam-splitter. Thus, we can write the transmission matrix of the four-mode system as $\tilde{\bm{T}}=\bm{T}\oplus\bm{T}$. By defining $\bm{W}=\tilde{\bm{T}}^*\oplus\tilde{\bm{T}}$, we can write the covariance matrix of the output light as
            \begin{equation}
                \bm{\Sigma} = \bm{P}\left[\bm{W}\tilde{\bm{\Sigma}}_0\bm{W}^{\dagger}+\frac{1}{2}\left(\mathbb{I}-\bm{W}\bm{W}^\dagger\right)\right]\bm{P}^{\mathrm{T}},
                \label{eq:output_cov_internal}
            \end{equation}
            where $\bm{P}$ is a permutation matrix that allows us to change the ordering of modes from $(\xi_1, \xi_2, \xi_1',\xi_2')$ to $(\xi_1, \xi_1', \xi_2, \xi_2')$. In this way, we have made the spectral modes corresponding to the same external mode adjacent to each other.

            With the definition of $\bm{\Sigma}$ in place, the definition of the adjacency matrix, $\bm{A}$, of the four mode system readily follows. Since the input light is non-displaced, the loop vector of the output light, $\bm{\gamma}$, identically vanishes.

            Given the ordering of modes dictated by Eq.~\eqref{eq:output_cov_internal}, and labeling them as $\{1,2,3,4\}$, we can define the partition $\bar{\Lambda}=\{\{1,2\},\{3\},\{4\}\}$ (which means that we choose to measure the first output port of the beam-splitter), and use Eq.~\eqref{eq:lhaf_heralded_offd_internal_special}, along with $\tilde{\bm{n}}=(n)$, to compute the state in modes (3, 4). We want to herald the desired Fock state in the first spectral mode of the second output port (i.e., mode 3), so we will filter out mode 4 by tracing it out. Let us call $\hat{\varrho}_{(n)}$ the (normalized) heralded state in mode 3. We compute the density matrix elements $\langle m'|\hat{\varrho}_{(n)}|n'\rangle$ with $m'$, $n'$ up to $n_{\text{cutoff}} = 15$ (taking into account that the presence of losses in the heralding modes can make the generated state have support over Fock states with $n'>n$). 
            
            To assess the difference between $\hat{\varrho}_{(n)}$ and the target Fock state $|n\rangle\langle n|$ we compute the fidelity
            \begin{equation}
                F_{\text{Fock}} = \sqrt{\langle n|\hat{\varrho}_{(n)}|n\rangle}
                \label{eq:fidelity_between_states_fock}
            \end{equation}
            as a function of $p\in[0.7, 1.0]$ and $\eta_{\text{eff}}\in[0.8, 1.0]$ (considering 50 samples for each of these parameters). The results of this computation for $\xi=1.1$ are shown in Fig.\ref{fig:fidelity_fock}. We carried out all our computations using Eq.~\eqref{eq:lhaf_heralded_offd_internal_special}.

            As can be seen, the generation of Fock states is more sensitive to photon losses that to spectral impurities. Indeed, for $p=1.0$ and $\eta_{\text{eff}}$ as high as $0.9$, we have that $F_{\text{Fock}}=\{0.90, 0.84, 0.79\}$ for $n=\{1,2,3\}$, while for $p=0.9$ and $\eta_{\text{eff}}=1.0$, we obtain $F_{\text{Fock}}= \{0.949, 0.943, 0.943\}$ for the same values of $n$. Moreover, the results reveal that the higher the photon-number state, the more sensitive to both photon losses and spectral impurities. Take for instance the case $p=0.7$, $\eta_{\text{eff}}$. For $n=1$ we have $F_{\text{Fock}} = 0.714$; for $n=2$, $F_{\text{Fock}} = 0.616$; and for $n=3$, $F_{\text{Fock}} = 0.560$. This suggest that there is a special need to reduce these sources of noise at the moment of generating photonic states having significant support over elements of the Fock state basis with large values of $n$.

            This illustrative example shows that we can apply our techniques to study how the presence of internal modes affects the generation of \textit{discrete-variable} photonic states. Indeed, we can think of the two-mode system we presented here as an imperfect source of photon-number states. Coupling many of these sources to other passive linear optical elements, we can study the influence of noise on the generation of larger photonic states, such as Bell or GHZ states~\cite{zhang2008demonstration, wiesner2024influence}. We believe that the speed-up offered by our techniques for computing blocked loop Hafnians paves the way to studying these type discrete-variable systems in more effective way.

        \subsection{Internal mode GKP state generation}

            Consider now the preparation of the approximate GKP state
            \begin{align}
                |\bar{0}_{\Delta}\rangle &= \hat{S}(0.196)\left(0.661|0\rangle-0.343|2\rangle + 0.253 |4\rangle\right.\nonumber\\
                &\left.-0.368|6\rangle +0.377|8\rangle+0.323|10\rangle+0.365|12\rangle\right),
                \label{eq:target_gkp_state}
            \end{align}
            where $\hat{S}(r)$ is the single-mode squeezing operator. This state has 96.9\% fidelity to the \textit{normalizable} GKP state $|0_{\Delta}\rangle$ with $\Delta=10\, \text{dB}$~\cite{tzitrin2024gkp} (see Ref.~\cite{tzitrin2020progress} for a definition of the normalizable GKP states).
            
            To generate this state, we will use a photonic system with three \textit{external} modes following a setup similar to the one described in Ref.~\cite{tzitrin2024gkp}. A set of three non-displaced, single-mode squeezed states, with real squeezing parameters $(\xi_1, \xi_2, \xi_3)$, are sent into a linear interferometer described by a $3\times 3$ transmission matrix $\bm{T}$. The output light of the interferometer is partially measured in the first two ports of the interferometer. This measurement heralds the desired GKP state in the third output port (see Fig.~\ref{fig:gkp_herald_setup} for an illustration of this setup). The specific values of the squeezing parameters and transmission matrix that are needed to generate the desired GKP state in the ideal (lossless) case (for which $\bm{T}=\bm{U}$ is a unitary matrix), as well as the covariance matrix of the corresponding Gaussian state before measurement, can be found in Ref.~\cite{tzitrin2024gkp}. The detection pattern in the heralding external modes is $\tilde{\bm{n}}=(5, 7)$.

            \begin{figure}[!t]
                {
                  \includegraphics[width=\columnwidth]{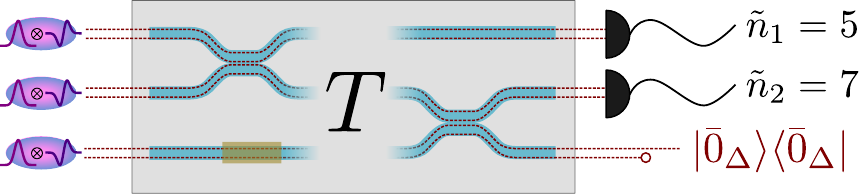}%
                }
                \caption{Setup for the generation of the approximate GKP state $|\bar{0}_{\Delta}\rangle$. A set of three single-mode squeezed states is sent into a linear interferometer represented by a transmission matrix $\bm{T}$. Upon the measurement of the photon-number pattern $\tilde{\bm{n}}=(5,7)$ in the first and second output ports of the interferometer, a state close in fidelity to $|\bar{0}_{\Delta}\rangle$ is heralded in the third output port. In our simulation, we consider that the input squeezed states are spectrally impure (as indicated by the multicolor ellipse representing the gaussian state), i.e., the light entering each port is a mixture of squeezed states in two mutually orthogonal spectral modes. As a consequence, our system has two internal (spectral) modes per external mode (ports of the interferometer). These internal modes are represented by dashed red lines. We filter out the second spectral mode in the third output port by tracing it out (as represented by the circle in the bottom red, dashed line).} 
                \label{fig:gkp_herald_setup}
            \end{figure}    

            Just as in the previous section, we will assume that the input squeezed states are not spectrally pure (they may be in a mixture of two mutually orthogonal spectral modes). Moreover, we will assume that the interferometer has uniform losses, so the transmission matrix will have the form $\bm{T}=\sqrt{\eta_{\text{eff}}}\bm{U}$, with $\bm{U}$ the unitary transmission matrix of the ideal case. We will simulate how the spectral impurity of the input light, and the presence of losses affects the generation of the approximate GKP state. 

            We will model the spectral impurity of the input light by assuming that we have six input single-mode squeezed states, three for each spectral mode, with squeezing parameters $(\xi_1, \xi_1', \xi_2, \xi_2',\xi_3, \xi_3')$. The relation between the parameters $\{\xi_k'\}$ and $\{\xi_k\}$ is dictated by Eq.~\eqref{eq:relation_between_sqz_par}. The adjacency matrix, $\bm{A}$, of the Gaussian state before measurement can be computed following a procedure analogous to the one described in the previous section. In this case, the loop vector of the output Gaussian state, $\gl$, also vanishes.  
            
            Labeling the modes of the system as $\{1,2,3,4,5,6\}$, and keeping in mind the ordering $(\xi_1, \xi_1', \xi_2, \xi_2',\xi_3, \xi_3')$, we can define the partition $\bar{\Lambda}=\{\{1,2\},\{3,4\},\{5\},\{6\}\}$ and use Eq.~\eqref{eq:lhaf_heralded_offd_internal_special}, with $\tilde{\bm{n}}=(5, 7)$, to compute the density matrix elements of the heralded state in modes (5, 6). We are interested in heralding the GKP state in the first spectral mode (i.e., mode 5), and so, in our simulations, we filter out mode 6 by tracing it out. Let $\hat{\varrho}_{\text{GKP}}$ be the heralded state in mode 5. We will compute the density matrix elements $\langle m|\hat{\varrho}_{\text{GKP}}|n\rangle$ with $m$, $n$ up to $n_{\text{cutoff}} = 26$ (as was done in Ref.~\cite{tzitrin2024gkp}).

            \begin{figure}[!t]
                {
                  \includegraphics[width=\columnwidth]{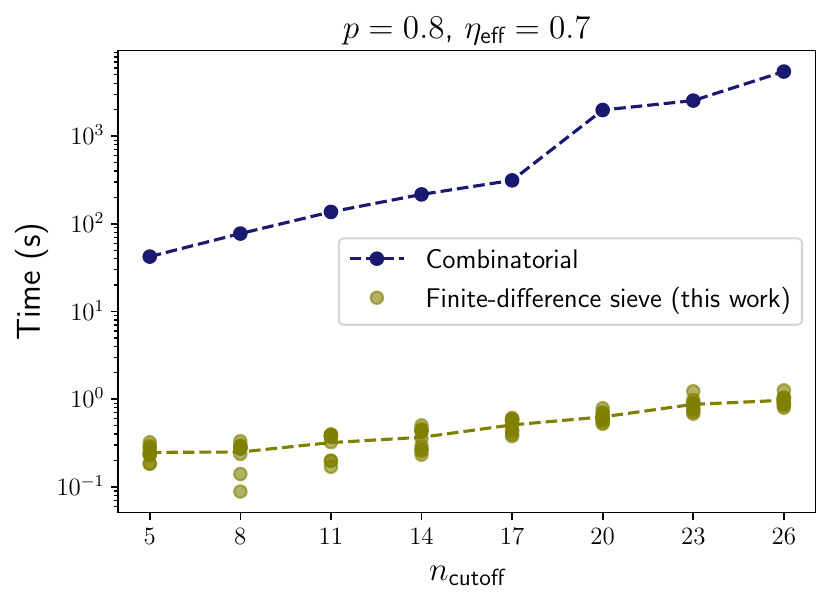}%
                }
                \caption{Time of computation of $\hat{\varrho}_{\text{GKP}}$ as a function of the photon number cutoff, $n_{\text{cutoff}}$, for $p = 0.8$ and $\eta_{\text{eff}}=0.7$. The dark blue, dashed line with circles corresponds to the combinatorial definition of the blocked loop Hafnian. The olive green circles correspond to ten runs of the finite-difference sieve expression for the density matrix elements of $\hat{\varrho}_{\text{GKP}}$. The dashed, olive green line indicates the average value of the ten runs. The vertical axis is in logarithmic scale. We can see that there is a clear advantage in using the finite-difference sieve version of the blocked loop Hafnian.} 
                \label{fig:gkp_time_comparison}
            \end{figure}
            
            Before showing how the presence of spectral impurities and losses affect the generation of approximate GKP states, it is worth noting that Eq.~\eqref{eq:lhaf_heralded_offd_internal_2} suggests that we can use $n_{\text{cutoff}}$ to parametrize the time it takes to calculate $\hat{\varrho}_{\text{GKP}}$. This is not only due to the fact that $n_{\text{cutoff}}$ determines the number of entries in the density matrix, but also because this cutoff determines the time complexity of the blocked loop Hafnians that take the longest to compute. On this account, let us illustrate the advantage of using the finite-difference sieve version of the blocked loop Hafnian (like that appearing in Eq.~\eqref{eq:lhaf_heralded_offd_internal_special}) over the original definition of Eq.~\eqref{eq:Lamb_lhaf_def} (which we will refer to as the \textit{combinatorial} blocked loop Hafnian in the remaining of this section) in order to compute $\hat{\varrho}_{\text{GKP}}$.

            We arbitrarily set $p = 0.8$, $\eta_{\text{eff}}=0.7$ and compute $\hat{\varrho}_{\text{GKP}}$ using both the combinatorial and finite-difference sieve expressions of the blocked loop Hafnian. The comparison of the time it takes to compute the density matrix using these methods for $n_{\text{cutoff}}\in\{5, 8, 11, 14, 17, 20, 23, 26\}$ is shown in Fig.~\ref{fig:gkp_time_comparison}. We carried out our computations in a 64-core CPU with two AMD EPYC 7532 (Zen 2) processors with 2.4 GHz clock speed. The computation of the different density matrix elements for the combinatorial case were done in  parallel using up to 50 cores, and the calculation of the loop and blocked loop Hafnians involved was further optimized using the \texttt{numba} library~\cite{lam2015numba}. The density matrices obtained with both methods were equal within an absolute tolerance of $10^{-8}$. 
            
            As can be seen in Fig.~\ref{fig:gkp_time_comparison}, there is a clear advantage in using the finite-difference sieve version of the blocked loop Hafnian. For $n_{\text{cutoff}}=5$, which corresponds to the fastest computation of the density matrix, the combinatorial method takes around $42.29\,\text{s}$ to compute, while the finite-difference sieve method takes on average $0.25\,\text{s}$. For $n_{\text{cutoff}}=26$ the difference is even more striking; the combinatorial method takes close to 1.6 hours, while finite-difference sieve method takes on average $0.97\, \text{s}$. This represents a reduction of three orders of magnitude in the time of computation.

            \begin{figure}[!ht]
                {
                  \includegraphics[width=0.9\columnwidth]{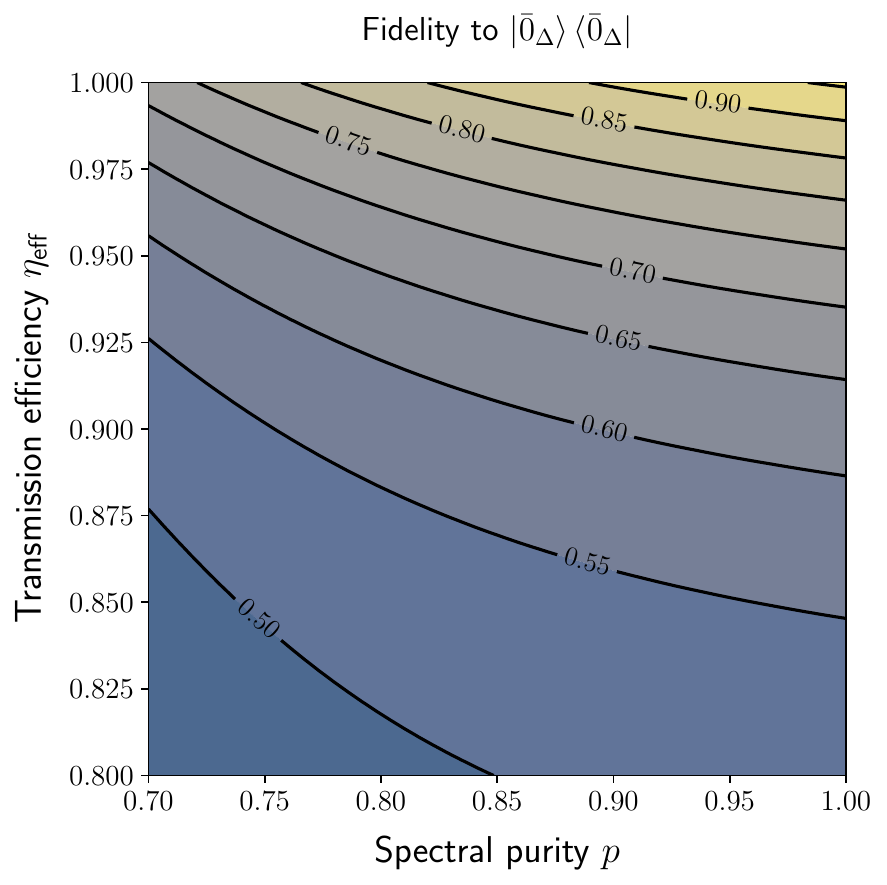}%
                }
                \caption{Fidelity between the noisy (including photon loss and spectral impurity) state $\hat{\varrho}_{\text{GKP}}$, and the target state $|\bar{0}_{\Delta}\rangle\langle\bar{0}_{\Delta}|$ as a function of the parameters $p$ and $\eta_{\text{eff}}$. As can be seen, the fidelity decays sharply with increasing losses, while the generated states seem to be much more robust to spectral impurities.}
                \label{fig:fidelity_gkp}
            \end{figure}

            We move on now to illustrate how the presence of losses and spectral impurities affects the generated GKP states. To do this, we compute the fidelity between the (now normalized) heralded state $\hat{\varrho}_{\text{GKP}}$, for $p\in[0.7,1.0]$ and $\eta_{\text{eff}}\in[0.8,1.0]$, and the target (normalized) state $|\bar{0}_{\Delta}\rangle\langle\bar{0}_{\Delta}|$,
            \begin{equation}
                F_{\text{GKP}} = \sqrt{\langle\bar{0}_{\Delta}|\hat{\varrho}_{\text{GKP}}|\bar{0}_{\Delta}\rangle}.
                \label{eq:fidelity_between_states_gkp}
            \end{equation}
            
            Fig.~\ref{fig:fidelity_gkp} shows the result of this computation for $n_{\text{cutoff}}=26$. We considered 50 samples for each of the parameters $p$, $\eta_{\text{eff}}$. All the density matrices involved in the calculation of the fidelity were obtained using Eq.~\eqref{eq:lhaf_heralded_offd_internal_special}.

            The maximum value of the fidelity is $F_{\text{GKP}}=0.96$ for $p=1.0$,  $\eta_{\text{eff}}=1.0$. For decreasing transmission, or increasing photon loss, we find that the fidelity sharply decreases for all values of $p$. For instance, for $p=1.0$ and $\eta_{\text{eff}}=0.97$, the fidelity drops to $F_{\text{GKP}}=0.81$. For $p=0.8$ and $\eta_{\text{eff}}$ as high as $0.8$ the fidelity is as low as $F_{\text{GKP}}=0.49$. On the other hand, the heralded states seem to be much more robust to the presence of spectral impurities. For $\eta_{\text{eff}}=1.0$ and $p=0.9$ we find that $F_{\text{GKP}}=0.91$, while for $p$ as low as $0.8$ the fidelity takes values close to $F_{\text{GKP}}=0.83$. This represents a gentler decrease than that obtained with increasing photon losses.

            To further demonstrate the effect of losses and spectral impurities on the generation of the targeted GKP states, Fig.~\ref{fig:wigner_func_gkp} shows the Wigner function of $\hat{\varrho}_{\text{GKP}}$ for $p\in\{1.0,0.9,0.7\}$ and $\eta\in\{1.0,0.95\}$. These distributions were computed using $n_{\text{cutoff}}=26$, and $300$ samples for both the $P$ and $Q$-quadratures.
            
            We can confirm that the generated states are much more robust to spectral impurities than to transmission losses. This is revealed by noticing the reduction in the absolute value of the maximum and minimum of the Wigner function when $p$, $\eta$ are modified. For $\eta=1.0$, $p=1.0$, we have a maximum of $0.317$, and a minimum of $-0.188$. For $\eta=1.0$, $p=0.9$, these values decrease to $0.279$ and $-0.165$, respectively. This represents a reduction of about $12\%$. For $\eta=1.0$, $p=0.7$, the maximum takes the value $0.178$, and the minimum is $-0.088$, which represents a reduction in the absolute values of $44\%$ and $53\%$, respectively. On the other hand, when $p=1.0$ and $\eta=0.95$, we already see the maximum drop to $0.126$, and the minimum take the value of $-0.092$, representing changes of $60\%$ and $51\%$, respectively. Keeping $\eta=0.95$ and decreasing $p$ only worsens the reduction in the absolute values.  

            \begin{figure*}[!ht]
                {
                  \includegraphics[width=0.9\textwidth]{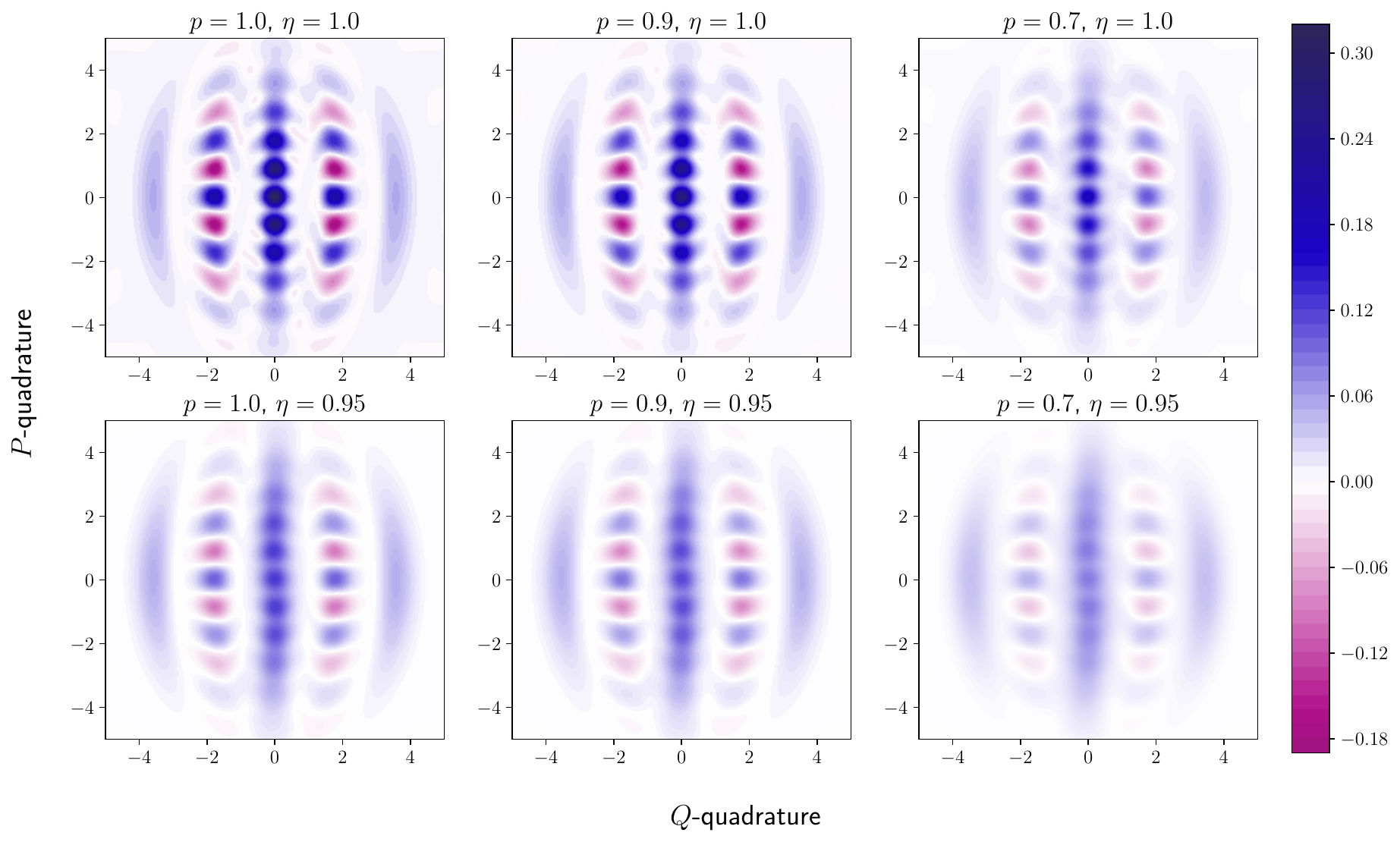}%
                }
                \caption{Wigner function of the noisy (including photon loss and spectral impurity) state $\hat{\varrho}_{\text{GKP}}$ for $p\in\{1.0,0.9,0.7\}$ and $\eta\in\{1.0,0.95\}$. As can be seen, the generated state is highly sensitive to photon losses, while they seem to be much more robust to spectral impurities.}
                \label{fig:wigner_func_gkp}
            \end{figure*}

            Previous studies (see for instance Ref.~\cite{tzitrin2020progress}) have also found a sharp decrease in the fidelity of the approximate GKP states obtained using lossy photonic architectures similar to that Fig.~\ref{fig:gkp_herald_setup}. It has been suggested that using physical models including photon losses, and subsequently optimizing over the various parameters that define these photonic circuits (e.g., squeezing parameters or transfer matrices) could be a strategy to alleviate the drastic reduction of fidelity. Even though the generated states seem to be more robust to spectral impurities than to photon losses, one could also use a similar optimization strategy to account for the effects of having spectrally impure squeezed light as input of the circuits generating the targeted GKP states. We believe that our results open the door for carrying out this type of optimization more efficiently, even for systems with large number of modes.
            
        \subsection{GBS total photon number distributions}

            Consider a GBS setup similar to the one depicted in Fig.~\ref{fig:coarse_grained_probs}. A set of $M$ input non-displaced single-mode squeezed states (this time considered to be spectrally pure) are sent into a linear interferometer represented by a $M\times M$ transmission matrix $\bm{T}$. This matrix is generally non-unitary as it may include the effects of losses and other sources of noise in the interferometer. The output light of the interferometer is described by a $M$-mode Gaussian state with adjacency matrix $\bm{A}$ and loop vector $\bm{\gamma}=\bm{0}$, and is measured using photon-number-resolving detectors. We want to compute the probability of measuring a total of $N$ photons across all detectors (which can be viewed as grouping the detectors according to the set partition $\Lambda=\{\{1,\dots,M\}\}$).
            
            Previous studies on the estimation of these type of distributions have successfully avoided the need of computing a blocked loop Hafnian directly from its combinatorial definition (see for instance~\cite{bressanini2023gaussian, dellios2024validation}). For Fock boson sampling this is also the case, as it has been shown that there exist efficient algorithms to compute binned photon-number distributions~\cite{seron2022efficient}, circumventing the need to deal with sums of a combinatorial number of elements.            
            
            Notably, phase space simulation methods have been applied to the computation of total number of photons and total number of clicks probability distributions of GBS implementations using photon-number-resolving and threshold detectors, respectively~\cite{dellios2022simulating, drummond2022simulating, dellios2024validation}. These strategies consist in expressing the desired distributions as an expectation value over a suitable phase space. Their computation require the definition of a normally ordered operator representing the detection event, and a \textit{positive} quasiprobability distribution that can be randomly sampled. Then, one can estimate photon-number distributions using averages of functions of these samples. For the case of the GBS setup described above, one needs to use the \textit{positive $P$-distribution}, defined over a complex phase space. We refer the interested reader to Refs.~\cite{ dellios2022validation, drummond2022simulating, dellios2024validation} for details about the definition of this quasiprobability distribution. 

            We are interested in comparing the performance of computing total photon-number distributions using the strategies we presented in Sec.~\ref{sec:total_ph_prob} with respect to phase space techniques. In particular, we will consider the methods described in Ref.~\cite{dellios2024validation}.

            Suppose that the input squeezed states have real squeezing parameters $\{\xi_j\}$. In terms of the positive $P$-distribution, $P(\bm{\alpha},\bm{\beta})$, the probability of measuring $N$ photons across all detectors can be written as 
            \begin{align}
                \Pr(N|\bm{A}) = \mathrm{Re}\left[\int\frac{\left(n'\right)^{N}}{N!}e^{-n'}P(\bm{\alpha},\bm{\beta})\,d\mu(\bm{\alpha},\bm{\beta})\right],
                \label{eq:coarse_grained_positive_P}
            \end{align}
             where $d\mu(\bm{\alpha}, \bm{\beta})$ is an integration measure over the $2M$-dimensional complex space of variables $\bm{\alpha}=(\alpha_1,\dots,\alpha_M)$ and $\bm{\beta}=(\beta_1,\dots,\beta_M)$, and $n' = \sum_{j=1}^M\alpha_j'(\beta_j')^*$, with $\bm{\alpha}'=\bm{T}\bm{\alpha}$, $\bm{\beta}'=\bm{T}\bm{\beta}$. Random samples following distribution $P(\bm{\alpha},\bm{\beta})$ can be generated by writing
             \begin{align}
			         \alpha_j = u_j\sqrt{\tfrac{1}{2}(\bar{n}_j+\bar{m}_j)} + iv_j\sqrt{\tfrac{1}{2}(\bar{n}_j-\bar{m}_j)},\\
			         \beta_j = u_j\sqrt{\tfrac{1}{2}(\bar{n}_j+\bar{m}_j)} - iv_j\sqrt{\tfrac{1}{2}(\bar{n}_j-\bar{m}_j)},
		          \label{eq:sampling_amplitudes}
	        \end{align}
            where $u_j$, $v_j$ are real Gaussian random variables satisfying $\mathbb{E}[u_jv_k] = 0$, $\mathbb{E}[u_ju_k] = \mathbb{E}[v_jv_k] = \delta_{jk}$ (with $\mathbb{E}[\cdot]$ indicating an expectation value with respect to the Gaussian distribution), and $\bar{n}_j=\sinh^2\xi_j$, $\bar{m}_j=\frac{1}{2}\sinh(2\xi_j)$.

            In the following, we will suppose that all the input squeezed states have the same squeezing parameter $\xi$, and that the transmission matrix takes the form $\bm{T}=\sqrt{\eta_{\text{eff}}}\bm{U}$, where $\eta_{\text{eff}}$ is a parameter describing the transmission efficiency of the interferometer and $\bm{U}$ is a \textit{Haar-random} unitary matrix. The computation of $\Pr(N|\bm{A})$ for these type of setups can be done analytically, which gives us a reference to to test our numerical results (see Ref.~\cite{deshpande2022quantum} for the complete expression of this distribution). 

            From Eq.~\eqref{eq:coarse_grained_positive_P} it is apparent that the time complexity of estimating $\Pr(N|\bm{A})$ heavily relies on the number of samples $(\bm{\alpha}, \bm{\beta})$ that are needed to ensure convergence. According to Ref.~\cite{dellios2024validation}, we need a number of samples of the order of $10^{6}$ to attain relative errors lower than $10^{-3}$. However, we cannot assess the performance of phase space methods based on the number of samples needed to make a good enough estimation. This is because the methods we developed in the previous sections do not rely on random sampling. This can be seen as an advantage, as it ensures that the time complexity of our techniques rely only on parameters of the physical system (including those related to the grouping of detectors and the number of detected photons in the samples), as well as relieving the computed probabilities from the presence of statistical errors.

            \begin{figure}[!ht]
                {
                  \includegraphics[width=\columnwidth]{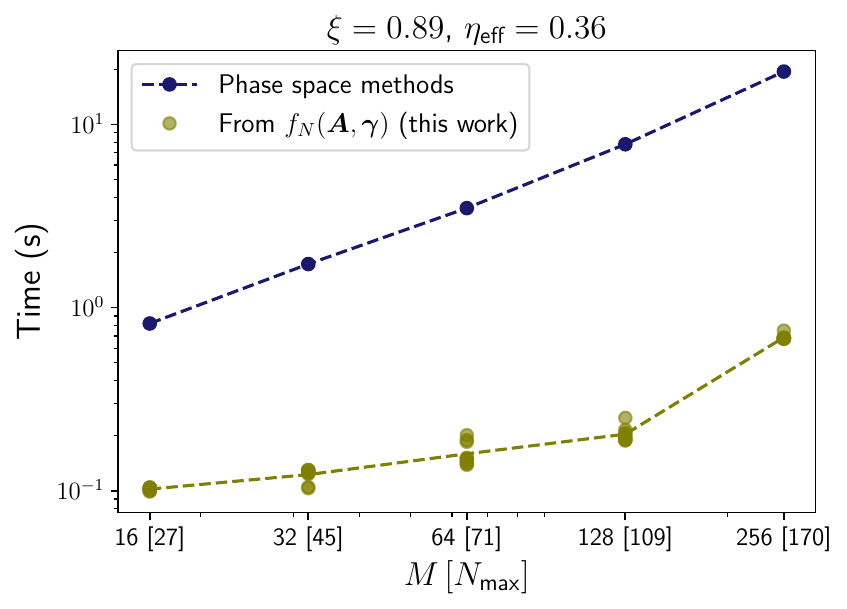}%
                }
                \caption{Time of computation of $\Pr(N|\bm{A})$ as a function of the number of modes, $M$, and its corresponding photon number cutoff, $N_{\text{max}}$ (the value of the cutoff is indicated within square brackets). The squeezing parameter was set to $\xi = 0.89$ and the transmission efficiency to $\eta_{\text{eff}}=0.36$. The dark blue, dashed line with circles corresponds to a calculation using phase space techniques. The olive green circles correspond to ten runs of Eq.~\eqref{eq:total_photon_number_prob} with $Y=[M]$. The dashed, olive green line indicates the average value of the ten runs. Both the vertical and horizontal axes are in logarithmic scale. We can see that there is a clear advantage in computing the function $f_N(\bm{A},\gl)$ directly over estimating the distribution using phase space techniques.} 
                \label{fig:total_probs_time_comparison}
            \end{figure}

            A better option to evaluate the performance of these methods is to parametrize the time of computation in terms of the total number of photons, $N$. Indeed, the time complexity of phase space methods relies on this value via the operator $e^{-n'}(n')^{N}/N!$,
            which needs to be computed recursively in order to avoid numerical instabilities induced by the term $N!$ when the number of detected photons is large. The time complexity of our methods depends on $N$ via the function $f_N(\bm{A},\gl)$ (defined in Eq.~\eqref{eq:f_function}). Thus, it is reasonable to use this value to parametrize the time it takes to compute $\Pr(N|\bm{A})$. Even more, we can parametrize the runtime of computation in terms of a cutoff, $N_{\text{max}}$, that is chosen so that $\Pr(N > N_{\text{max}}|\bm{A}) < \epsilon$ with $\epsilon >0$. This means that we would be comparing the runtimes of computation of the distribution $\Pr(N|\bm{A})$ in the most relevant part of its support.

            \begin{figure}[!t]
                {
                  \includegraphics[width=\columnwidth]{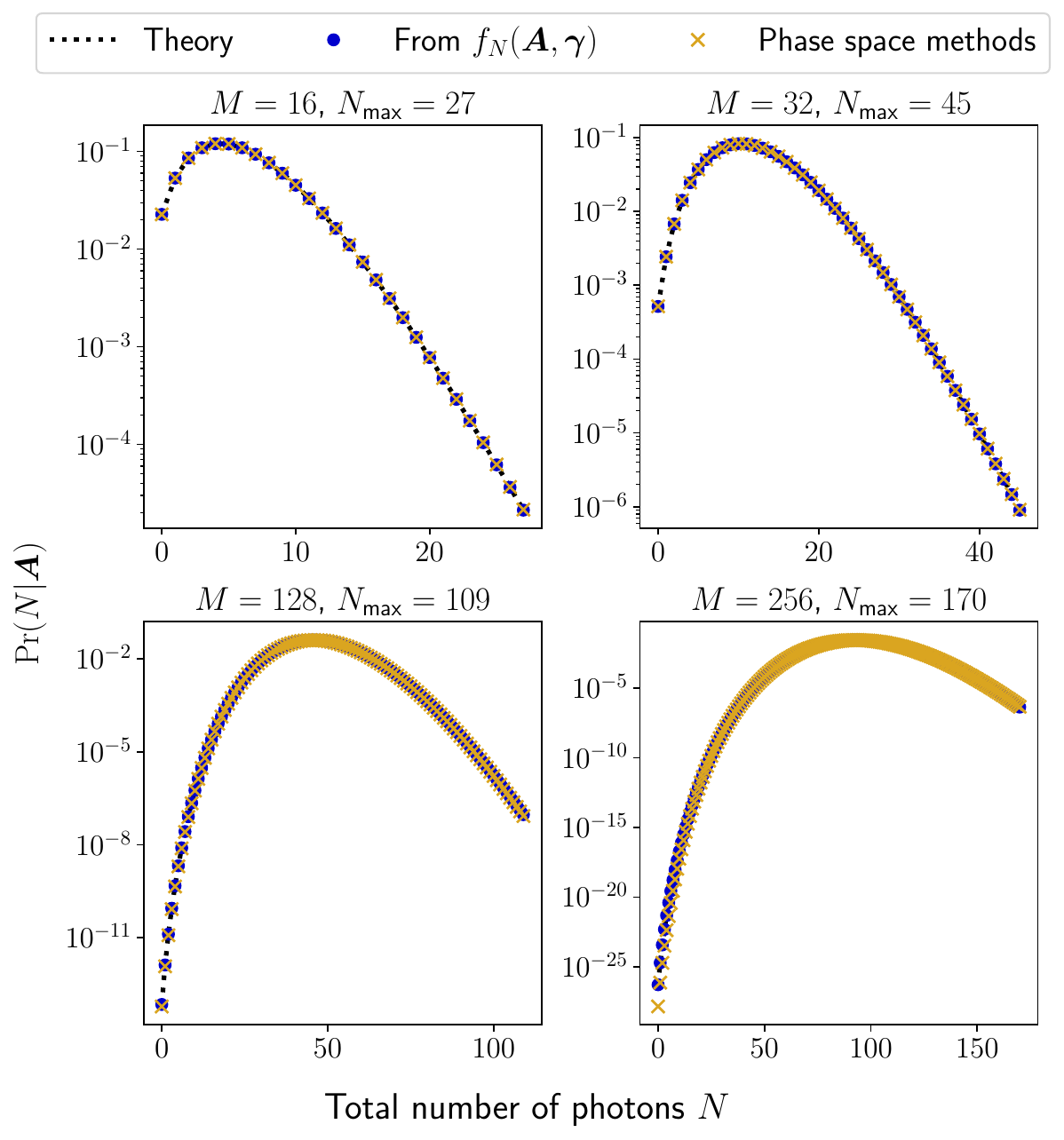}%
                }
                \caption{Total photon number distribution $\Pr(N|\bm{A})$ as a function of $N$ for $M\in \{16, 32, 128, 256\}$. The corresponding values of $N_{\text{max}}$ are $\{27, 45, 109, 170\}$. The vertical axis is in logarithmic scale. The blue circles are the results of computing Eq.~\eqref{eq:total_photon_number_prob} with $Y=[M]$. The yellow crosses were estimated using phase space methods with $2.4\times10^6$ samples. The black dotted line corresponds to the theoretical predictions. For all distributions $\xi=0.89$ and $\eta_{\text{eff}}=0.36$. We can see that there is an excellent agreement between the numerical computations and the theoretical predictions.} 
                \label{fig:total_photon_distributions}
            \end{figure}
            
            Since we are considering a system where we have input light in every mode of the interferometer, and all the squeezed states have the same squeezing parameter, $\xi$, we can change $N_{\text{max}}$ by either varying $\xi$ or modifying the number of modes, $M$. Here, we will keep $\xi$ fixed, and modify $N_{\text{max}}$ by changing the number of modes. 
            
            Fig.~\ref{fig:total_probs_time_comparison} shows the comparison of the time it takes to compute $\Pr(N|\bm{A})$ using both phase space methods and Eq.~\eqref{eq:total_photon_number_prob} (we show the results of ten runs of this equation), for $\xi=0.89$ and $\eta_{\text{eff}}=0.36$. We used the values $M\in\{16,32,64,128,256\}$, which correspond to $N_{\text{max}}\in\{27, 45, 71, 109, 170\}$. The different $N_{\text{max}}$ were computed by setting $\epsilon$ to the order of $10^{-7}$. The estimation of $\Pr(N|\bm{A})$ using phase space methods was done using the parallelized \texttt{MATLAB} implementation of Ref.~\cite{dellios2024validation}. The number of samples used was $2.4\times10^{6}$ for all values of $M$.
            Eq.~\eqref{eq:total_photon_number_prob} was computed using a \texttt{Python} implementation optimized using the \texttt{numba} library. We carried out our calculations using a 8-core, 16-thread CPU with a AMD Ryzen 9 8945HS processor with 4.0 GHz clock speed. Some of the resulting probability distributions are shown in Fig.~\ref{fig:total_photon_distributions}.

            We can readily notice that there is an advantage in using Eq.~\eqref{eq:total_photon_number_prob} over phase space methods for computing $\Pr(N|\bm{A})$. For $N_{\max}=27$ ($M=16$), phase space methods take around $0.819\,\text{s}$, while calculating $f_N(\bm{A},\gl)$ takes on average $0.102\,\text{s}$. The worst case runtime for computing this function for this value of $N_{\text{max}}$ is $0.104\,\text{s}$.
            For $N_{\max}=170$ ($M=256$) the phase space technique takes around $19.45\,\text{s}$, while Eq.~\eqref{eq:total_photon_number_prob} takes on average $0.69\,\text{s}$, with a worst case runtime of $0.75\,\text{s}$. This represents a speedup of about one order of magnitude. 

            \begin{figure}[!t]
                {
                  \includegraphics[width=\columnwidth]{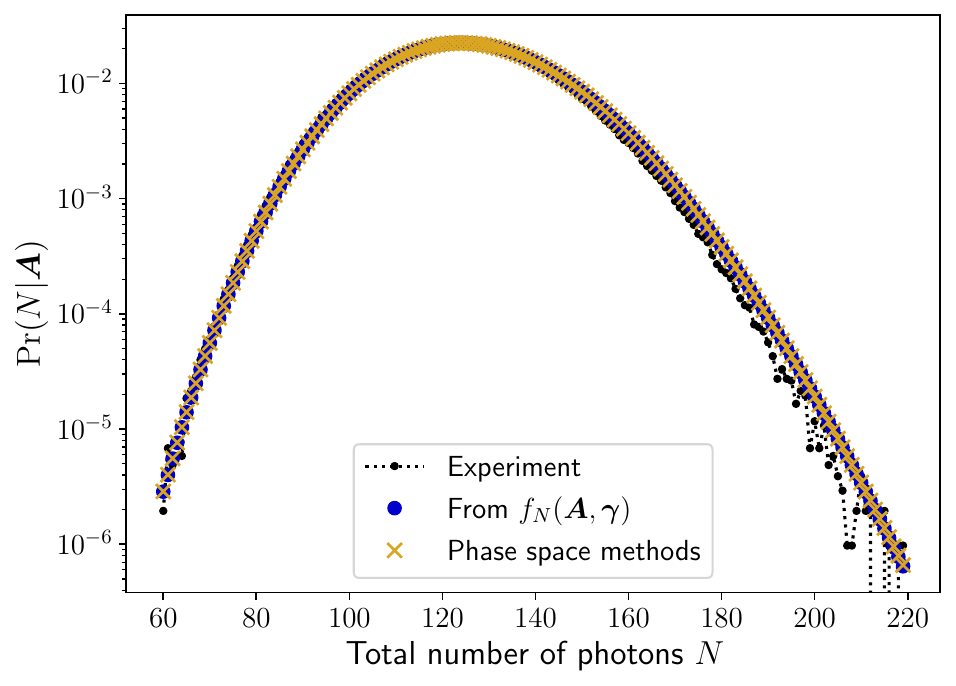}%
                }
                \caption{Total photon-number distribution for one of the configurations of the Borealis GBS experiment. This configuration corresponds to a system of 216 modes. The squeezing parameters take values between 1.09 and 1.12. The mean number of output photons is 124.35. The dotted line with points corresponds to the estimated distribution using close to $10^6$ experimental samples. The blue circles correspond to the computation using Eq.~\eqref{eq:total_photon_number_prob}, while the yellow crosses indicate the results of the estimation using phase space methods. The numerical computations were made using the squeezing parameters and the transfer matrix made publicly available by the authors of the experiment. The number of samples used for the estimation of $\Pr(N|\bm{A})$ using phase space methods was $2.4\times 10^6$. Here, we can see an excellent agreement between the two numerical methods.} 
                \label{fig:total_photon_distribution_borealis}
            \end{figure}

            Let us now apply our methods to a more general situation: the computation of the total photon-number distribution of one of the configurations of the Borealis experiment~\cite{madsen2022quantum}, a GBS implementation that has claimed to have achieved quantum computational advantage. One of the validation tests that can be used to build confidence in the fact that a GBS experiment truly demonstrates quantum advantage is to show that the experimentally estimated total photon-number distribution is sufficiently close to that of the expected theoretical model of the experiment, which is usually known as the ground truth of the experiment. This theoretical model takes into account all the noise and imperfections that the experimenters can account for, and it is usually defined by a either a covariance matrix, or a set of input squeezing parameters and a transmission matrix. Here, we will compute the ground truth total photon-number distribution using the information about the squeezing parameters and the transfer matrix of the experiment that the authors of the experiment made publicly available (see Ref.~\cite{madsen2022quantum} for a link to these data).

            We choose a configuration of Borealis with 216 modes, and input single-mode squeezed states in every one of them. The squeezing parameters take values between 1.09 and 1.12, which leads to a mean number of output photons of 124.35. The resulting photon number distribution has significant support up to $N_{\text{max}}=219$. Fig.~\ref{fig:total_photon_distribution_borealis} shows the computed $\Pr(N|\bm{A})$ using both phase space methods (with $2.4\times10^6$ samples) and Eq.~\eqref{eq:total_photon_number_prob}, as well as the estimated experimental distribution.

            As we can see, there is a very good agreement between the distributions obtained with phase space techniques, and by computing function $f_N(\bm{A},\bm{\gamma})$. There is however, a clear difference between the runtimes of these methods; the computation using Eq.~\eqref{eq:total_photon_number_prob} took about $1.05\,\text{s}$, while the phase space estimation took around $15.64\,\text{s}$. In addition, it is worth mentioning once more that computing $f_N(\bm{A},\bm{\gamma})$ directly shields the computed distribution from statistical uncertainties proper of phase space techniques.
    
    \section{\label{sec:conclusion} Conclusion}
    
        In this work we have shown that it is possible to significantly reduce the time complexity of computing the Fock state representation of both Gaussian and non-Gaussian states with internal modes. Equivalently, we have proven the reduction of complexity of computing coarse-grained photon-number distributions of Gaussian states. 
        
        More precisely, according to their direct definition, the computation of these type of photon-number statistics for a system with $M$ modes, and for detection events with $N$ total detected photons, is expected to have time complexity
        \begin{equation*}
            \mathcal{O}\left(N^32^N\prod_{l=1}^L\binom{K_l + b_l + 1}{b_l}\right),    
        \end{equation*}
        where $L$ is the number of bins (or blocks) that define the coarse-grained detection event, $K_l$ is the size of each bin, and $b_l$ is the number of photons detected in each bin. In the case of Gaussian states with $K$ internal modes per external mode, $L=M$ and $K_l=K$ for all $l$. We have shown that we can reduce the time of computation to 
        \begin{equation*}
            \mathcal{O}\left(\left(NM^3 + N^2\log N\right)\prod_{l=1}^L(b_l+1)\right). 
        \end{equation*}
        Our results are likely to be optimal up to improvements to the polynomial prefactor, as otherwise one could devise algorithms that would give an unexpected speedup to the computation of quantities like the permanent, which have been studied for many decades in computational complexity theory.

        To illustrate the aforementioned reduction of complexity, we simulated the generation of approximate GKP states using a three-mode photonic circuit in the presence of photon losses and spectral impurities of the light sources. Supposing that the state of the light sources could be in a mixture of two spectral modes, we found that our techniques compute the density matrices of the desired GKP states (up to 26 elements of the Fock states basis) approximately $10^3$ times faster than combinatorial methods. This speedup allowed us to map the reduction of fidelity of the generated GKP states caused by these common sources of noise. Even though this example corresponds to a system with moderate size, we can readily apply our results to the simulation of systems with a larger number of modes. 

        In addition, we computed the total photon-number distribution of a number of GBS implementations, including one of the configurations of the Borealis experiment. We compared the performance of our methods with respect to phase space strategies, which are known to avoid the need of dealing with sums of a combinatorial number of elements. We found that our methods are approximately $10$ times faster than phase space techniques at the task of computing the total photon-number distribution of GBS up to 219 detected photons.  

        We have also illustrated the way our methods can be applied to study how the presence of internal modes affects the generation of discrete-variable photonic states. Indeed, we have shown that our techniques can be used to study the influence of imperfect squeezed light sources on the generation of Fock states. In future work, we will scale-up this approach, and show that we may also simulate the influence of imperfect light sources on the generation of other types of photonic states, such as Bell or GHZ states. 
        
        From a more applied point of view, our results pave the way for performing constrained optimization of \emph{realistic} photonic circuits that contain photon losses and mode-mismatch. While artificial intelligence software has been used to solve idealized versions of these type of problems~\cite{ruiz2023digital}, their results are not necessarily that surprising~\cite{rudolph2023terry} and, perhaps more importantly, these idealized photonic circuits fail to incorporate the complexities and imperfections of real experiments.
        As argued by T. Rudoph~\cite{rudolph2023terry}, the main use of computational techniques in aiding the design of quantum photonic components and states is for them to perform \emph{constrained and robust} optimization that takes into account the physical restrictions associated with \textit{``imperfect multimode sources, losses, manufacturing imperfections in passive interferometers and so on''}.
        
        The theory and tools we introduced here constitute a first step in this direction, as they provide optimal algorithms to solve the problem of determining, given a set of parameters describing imperfections and tunable knobs, what is the outcome of a given heralding experiment. Our tools can then be immediately coupled with gradient free optimizers to perform inverse design and optimize any tunable parameters to maximize the performance of a given photonic circuit in preparing a target state. We believe that they can also be generalized to be differentiable~\cite{yao2024riemannian}, thus allowing the use of gradient based searches. This constitutes the first step in a line of future work.

            \emph{Addendum ---} While preparing this manuscript we became aware of the related and very recent work by Kleinpa{\ss} et al. \cite{kleinpass2024frequency} and Steinmetz et al. ~\cite{steinmetz2024simulating}. The former truncates the equivalent of Eq.~\eqref{eq:lhaf_master1} to leading order to deal with photon pair problems using differentiable programming techniques, while the latter uses first quantization techniques for dealing with internal modes for discrete variable circuits.

    \section*{Code availability}
        The code used in the computation of the results shown in Figs.~\ref{fig:gkp_time_comparison} to~\ref{fig:total_photon_distribution_borealis} is available at \url{https://github.com/polyquantique/blocked_hafnian}.

    \section*{Acknowledgments}
        J.M.-C. and N.Q. acknowledge the support from the Ministère de l’Économie et de l’Innovation du Québec and the Natural Sciences and Engineering Research Council of Canada. They also thank A.S. Dellios, M.D. Reid and P.D. Drummond  for sharing their positive-P phase-space code and for insightful discussions. J.F.F.B. and N.Q. thank D. Phillips for insightful discussions.

    \bibliography{bib}

\begin{thebibliography}{87}%
\makeatletter
\providecommand \@ifxundefined [1]{%
 \@ifx{#1\undefined}
}%
\providecommand \@ifnum [1]{%
 \ifnum #1\expandafter \@firstoftwo
 \else \expandafter \@secondoftwo
 \fi
}%
\providecommand \@ifx [1]{%
 \ifx #1\expandafter \@firstoftwo
 \else \expandafter \@secondoftwo
 \fi
}%
\providecommand \natexlab [1]{#1}%
\providecommand \enquote  [1]{``#1''}%
\providecommand \bibnamefont  [1]{#1}%
\providecommand \bibfnamefont [1]{#1}%
\providecommand \citenamefont [1]{#1}%
\providecommand \href@noop [0]{\@secondoftwo}%
\providecommand \href [0]{\begingroup \@sanitize@url \@href}%
\providecommand \@href[1]{\@@startlink{#1}\@@href}%
\providecommand \@@href[1]{\endgroup#1\@@endlink}%
\providecommand \@sanitize@url [0]{\catcode `\\12\catcode `\$12\catcode `\&12\catcode `\#12\catcode `\^12\catcode `\_12\catcode `\%12\relax}%
\providecommand \@@startlink[1]{}%
\providecommand \@@endlink[0]{}%
\providecommand \url  [0]{\begingroup\@sanitize@url \@url }%
\providecommand \@url [1]{\endgroup\@href {#1}{\urlprefix }}%
\providecommand \urlprefix  [0]{URL }%
\providecommand \Eprint [0]{\href }%
\providecommand \doibase [0]{https://doi.org/}%
\providecommand \selectlanguage [0]{\@gobble}%
\providecommand \bibinfo  [0]{\@secondoftwo}%
\providecommand \bibfield  [0]{\@secondoftwo}%
\providecommand \translation [1]{[#1]}%
\providecommand \BibitemOpen [0]{}%
\providecommand \bibitemStop [0]{}%
\providecommand \bibitemNoStop [0]{.\EOS\space}%
\providecommand \EOS [0]{\spacefactor3000\relax}%
\providecommand \BibitemShut  [1]{\csname bibitem#1\endcsname}%
\let\auto@bib@innerbib\@empty
\bibitem [{\citenamefont {Serafini}(2017)}]{serafini2017quantum}%
  \BibitemOpen
  \bibfield  {author} {\bibinfo {author} {\bibfnamefont {A.}~\bibnamefont {Serafini}},\ }\href@noop {} {\emph {\bibinfo {title} {Quantum continuous variables: a primer of theoretical methods}}}\ (\bibinfo  {publisher} {CRC press},\ \bibinfo {year} {2017})\BibitemShut {NoStop}%
\bibitem [{\citenamefont {Ferraro}\ \emph {et~al.}(2005)\citenamefont {Ferraro}, \citenamefont {Olivares},\ and\ \citenamefont {Paris}}]{ferraro2005gaussian}%
  \BibitemOpen
  \bibfield  {author} {\bibinfo {author} {\bibfnamefont {A.}~\bibnamefont {Ferraro}}, \bibinfo {author} {\bibfnamefont {S.}~\bibnamefont {Olivares}},\ and\ \bibinfo {author} {\bibfnamefont {M.~G.}\ \bibnamefont {Paris}},\ }\bibfield  {title} {\bibinfo {title} {Gaussian states in continuous variable quantum information},\ }\href@noop {} {\bibfield  {journal} {\bibinfo  {journal} {arXiv preprint quant-ph/0503237}\ } (\bibinfo {year} {2005})}\BibitemShut {NoStop}%
\bibitem [{\citenamefont {Weedbrook}\ \emph {et~al.}(2012)\citenamefont {Weedbrook}, \citenamefont {Pirandola}, \citenamefont {Garc{\'\i}a-Patr{\'o}n}, \citenamefont {Cerf}, \citenamefont {Ralph}, \citenamefont {Shapiro},\ and\ \citenamefont {Lloyd}}]{weedbrook2012gaussian}%
  \BibitemOpen
  \bibfield  {author} {\bibinfo {author} {\bibfnamefont {C.}~\bibnamefont {Weedbrook}}, \bibinfo {author} {\bibfnamefont {S.}~\bibnamefont {Pirandola}}, \bibinfo {author} {\bibfnamefont {R.}~\bibnamefont {Garc{\'\i}a-Patr{\'o}n}}, \bibinfo {author} {\bibfnamefont {N.~J.}\ \bibnamefont {Cerf}}, \bibinfo {author} {\bibfnamefont {T.~C.}\ \bibnamefont {Ralph}}, \bibinfo {author} {\bibfnamefont {J.~H.}\ \bibnamefont {Shapiro}},\ and\ \bibinfo {author} {\bibfnamefont {S.}~\bibnamefont {Lloyd}},\ }\bibfield  {title} {\bibinfo {title} {Gaussian quantum information},\ }\href@noop {} {\bibfield  {journal} {\bibinfo  {journal} {Rev. Mod. Phys.}\ }\textbf {\bibinfo {volume} {84}},\ \bibinfo {pages} {621} (\bibinfo {year} {2012})}\BibitemShut {NoStop}%
\bibitem [{\citenamefont {Asavanant}\ and\ \citenamefont {Furusawa}(2022)}]{asavanant2022optical}%
  \BibitemOpen
  \bibfield  {author} {\bibinfo {author} {\bibfnamefont {W.}~\bibnamefont {Asavanant}}\ and\ \bibinfo {author} {\bibfnamefont {A.}~\bibnamefont {Furusawa}},\ }\href@noop {} {\emph {\bibinfo {title} {Optical Quantum Computers}}}\ (\bibinfo  {publisher} {AIP Publishing LLC},\ \bibinfo {year} {2022})\BibitemShut {NoStop}%
\bibitem [{\citenamefont {Barnett}\ and\ \citenamefont {Radmore}(2002)}]{barnett2002methods}%
  \BibitemOpen
  \bibfield  {author} {\bibinfo {author} {\bibfnamefont {S.}~\bibnamefont {Barnett}}\ and\ \bibinfo {author} {\bibfnamefont {P.~M.}\ \bibnamefont {Radmore}},\ }\href@noop {} {\emph {\bibinfo {title} {Methods in theoretical quantum optics}}},\ Vol.~\bibinfo {volume} {15}\ (\bibinfo  {publisher} {Oxford University Press},\ \bibinfo {year} {2002})\BibitemShut {NoStop}%
\bibitem [{\citenamefont {Gerry}\ and\ \citenamefont {Knight}(2005)}]{gerry2005introductory}%
  \BibitemOpen
  \bibfield  {author} {\bibinfo {author} {\bibfnamefont {C.}~\bibnamefont {Gerry}}\ and\ \bibinfo {author} {\bibfnamefont {P.~L.}\ \bibnamefont {Knight}},\ }\href@noop {} {\emph {\bibinfo {title} {Introductory quantum optics}}}\ (\bibinfo  {publisher} {Cambridge university press},\ \bibinfo {year} {2005})\BibitemShut {NoStop}%
\bibitem [{\citenamefont {Dodonov}\ \emph {et~al.}(1994)\citenamefont {Dodonov}, \citenamefont {Man’ko},\ and\ \citenamefont {Man’ko}}]{dodonov1994multidimensional}%
  \BibitemOpen
  \bibfield  {author} {\bibinfo {author} {\bibfnamefont {V.}~\bibnamefont {Dodonov}}, \bibinfo {author} {\bibfnamefont {O.}~\bibnamefont {Man’ko}},\ and\ \bibinfo {author} {\bibfnamefont {V.}~\bibnamefont {Man’ko}},\ }\bibfield  {title} {\bibinfo {title} {Multidimensional hermite polynomials and photon distribution for polymode mixed light},\ }\href@noop {} {\bibfield  {journal} {\bibinfo  {journal} {Phys. Rev. A}\ }\textbf {\bibinfo {volume} {50}},\ \bibinfo {pages} {813} (\bibinfo {year} {1994})}\BibitemShut {NoStop}%
\bibitem [{\citenamefont {Krenn}\ \emph {et~al.}(2017)\citenamefont {Krenn}, \citenamefont {Gu},\ and\ \citenamefont {Zeilinger}}]{krenn2017quantum}%
  \BibitemOpen
  \bibfield  {author} {\bibinfo {author} {\bibfnamefont {M.}~\bibnamefont {Krenn}}, \bibinfo {author} {\bibfnamefont {X.}~\bibnamefont {Gu}},\ and\ \bibinfo {author} {\bibfnamefont {A.}~\bibnamefont {Zeilinger}},\ }\bibfield  {title} {\bibinfo {title} {Quantum experiments and graphs: Multiparty states as coherent superpositions of perfect matchings},\ }\href@noop {} {\bibfield  {journal} {\bibinfo  {journal} {Phys. Rev. Lett.}\ }\textbf {\bibinfo {volume} {119}},\ \bibinfo {pages} {240403} (\bibinfo {year} {2017})}\BibitemShut {NoStop}%
\bibitem [{\citenamefont {Gu}\ \emph {et~al.}(2019)\citenamefont {Gu}, \citenamefont {Erhard}, \citenamefont {Zeilinger},\ and\ \citenamefont {Krenn}}]{gu2019quantum}%
  \BibitemOpen
  \bibfield  {author} {\bibinfo {author} {\bibfnamefont {X.}~\bibnamefont {Gu}}, \bibinfo {author} {\bibfnamefont {M.}~\bibnamefont {Erhard}}, \bibinfo {author} {\bibfnamefont {A.}~\bibnamefont {Zeilinger}},\ and\ \bibinfo {author} {\bibfnamefont {M.}~\bibnamefont {Krenn}},\ }\bibfield  {title} {\bibinfo {title} {Quantum experiments and graphs ii: Quantum interference, computation, and state generation},\ }\href@noop {} {\bibfield  {journal} {\bibinfo  {journal} {Proceedings of the National Academy of Sciences}\ }\textbf {\bibinfo {volume} {116}},\ \bibinfo {pages} {4147} (\bibinfo {year} {2019})}\BibitemShut {NoStop}%
\bibitem [{\citenamefont {Quesada}(2019)}]{quesada2019franck}%
  \BibitemOpen
  \bibfield  {author} {\bibinfo {author} {\bibfnamefont {N.}~\bibnamefont {Quesada}},\ }\bibfield  {title} {\bibinfo {title} {Franck-condon factors by counting perfect matchings of graphs with loops},\ }\href@noop {} {\bibfield  {journal} {\bibinfo  {journal} {The Journal of chemical physics}\ }\textbf {\bibinfo {volume} {150}},\ \bibinfo {pages} {164113} (\bibinfo {year} {2019})}\BibitemShut {NoStop}%
\bibitem [{\citenamefont {Bulmer}\ \emph {et~al.}(2022{\natexlab{a}})\citenamefont {Bulmer}, \citenamefont {Bell}, \citenamefont {Chadwick}, \citenamefont {Jones}, \citenamefont {Moise}, \citenamefont {Rigazzi}, \citenamefont {Thorbecke}, \citenamefont {Haus}, \citenamefont {Van~Vaerenbergh}, \citenamefont {Patel} \emph {et~al.}}]{bulmer2022boundary}%
  \BibitemOpen
  \bibfield  {author} {\bibinfo {author} {\bibfnamefont {J.~F.}\ \bibnamefont {Bulmer}}, \bibinfo {author} {\bibfnamefont {B.~A.}\ \bibnamefont {Bell}}, \bibinfo {author} {\bibfnamefont {R.~S.}\ \bibnamefont {Chadwick}}, \bibinfo {author} {\bibfnamefont {A.~E.}\ \bibnamefont {Jones}}, \bibinfo {author} {\bibfnamefont {D.}~\bibnamefont {Moise}}, \bibinfo {author} {\bibfnamefont {A.}~\bibnamefont {Rigazzi}}, \bibinfo {author} {\bibfnamefont {J.}~\bibnamefont {Thorbecke}}, \bibinfo {author} {\bibfnamefont {U.-U.}\ \bibnamefont {Haus}}, \bibinfo {author} {\bibfnamefont {T.}~\bibnamefont {Van~Vaerenbergh}}, \bibinfo {author} {\bibfnamefont {R.~B.}\ \bibnamefont {Patel}}, \emph {et~al.},\ }\bibfield  {title} {\bibinfo {title} {The boundary for quantum advantage in gaussian boson sampling},\ }\href@noop {} {\bibfield  {journal} {\bibinfo  {journal} {Science advances}\ }\textbf {\bibinfo {volume} {8}},\ \bibinfo {pages} {eabl9236} (\bibinfo {year} {2022}{\natexlab{a}})}\BibitemShut {NoStop}%
\bibitem [{\citenamefont {Cardin}\ and\ \citenamefont {Quesada}(2024)}]{cardin2022photon}%
  \BibitemOpen
  \bibfield  {author} {\bibinfo {author} {\bibfnamefont {Y.}~\bibnamefont {Cardin}}\ and\ \bibinfo {author} {\bibfnamefont {N.}~\bibnamefont {Quesada}},\ }\bibfield  {title} {\bibinfo {title} {Photon-number moments and cumulants of {G}aussian states},\ }\href@noop {} {\bibfield  {journal} {\bibinfo  {journal} {{Quantum}}\ }\textbf {\bibinfo {volume} {8}},\ \bibinfo {pages} {1521} (\bibinfo {year} {2024})}\BibitemShut {NoStop}%
\bibitem [{\citenamefont {Roeland}\ \emph {et~al.}(2021)\citenamefont {Roeland}, \citenamefont {Kaali}, \citenamefont {Rodriguez}, \citenamefont {Treps},\ and\ \citenamefont {Parigi}}]{roeland2021mode}%
  \BibitemOpen
  \bibfield  {author} {\bibinfo {author} {\bibfnamefont {G.}~\bibnamefont {Roeland}}, \bibinfo {author} {\bibfnamefont {S.}~\bibnamefont {Kaali}}, \bibinfo {author} {\bibfnamefont {V.~R.}\ \bibnamefont {Rodriguez}}, \bibinfo {author} {\bibfnamefont {N.}~\bibnamefont {Treps}},\ and\ \bibinfo {author} {\bibfnamefont {V.}~\bibnamefont {Parigi}},\ }\bibfield  {title} {\bibinfo {title} {Mode-selective single-photon addition to a multimode quantum field},\ }\href@noop {} {\bibfield  {journal} {\bibinfo  {journal} {arXiv preprint arXiv:2110.14488}\ } (\bibinfo {year} {2021})}\BibitemShut {NoStop}%
\bibitem [{\citenamefont {Thomas}\ \emph {et~al.}(2021)\citenamefont {Thomas}, \citenamefont {McCutcheon},\ and\ \citenamefont {McCutcheon}}]{thomas2021general}%
  \BibitemOpen
  \bibfield  {author} {\bibinfo {author} {\bibfnamefont {O.~F.}\ \bibnamefont {Thomas}}, \bibinfo {author} {\bibfnamefont {W.}~\bibnamefont {McCutcheon}},\ and\ \bibinfo {author} {\bibfnamefont {D.~P.}\ \bibnamefont {McCutcheon}},\ }\bibfield  {title} {\bibinfo {title} {A general framework for multimode gaussian quantum optics and photo-detection: Application to hong--ou--mandel interference with filtered heralded single photon sources},\ }\href@noop {} {\bibfield  {journal} {\bibinfo  {journal} {APL Photonics}\ }\textbf {\bibinfo {volume} {6}},\ \bibinfo {pages} {040801} (\bibinfo {year} {2021})}\BibitemShut {NoStop}%
\bibitem [{\citenamefont {Quesada}(2015)}]{quesada2015very}%
  \BibitemOpen
  \bibfield  {author} {\bibinfo {author} {\bibfnamefont {N.}~\bibnamefont {Quesada}},\ }\emph {\bibinfo {title} {Very Nonlinear Quantum Optics}},\ \href@noop {} {Ph.D. thesis},\ \bibinfo  {school} {University of Toronto} (\bibinfo {year} {2015})\BibitemShut {NoStop}%
\bibitem [{\citenamefont {Tiedau}\ \emph {et~al.}(2019)\citenamefont {Tiedau}, \citenamefont {Bartley}, \citenamefont {Harder}, \citenamefont {Lita}, \citenamefont {Nam}, \citenamefont {Gerrits},\ and\ \citenamefont {Silberhorn}}]{tiedau2019scalability}%
  \BibitemOpen
  \bibfield  {author} {\bibinfo {author} {\bibfnamefont {J.}~\bibnamefont {Tiedau}}, \bibinfo {author} {\bibfnamefont {T.~J.}\ \bibnamefont {Bartley}}, \bibinfo {author} {\bibfnamefont {G.}~\bibnamefont {Harder}}, \bibinfo {author} {\bibfnamefont {A.~E.}\ \bibnamefont {Lita}}, \bibinfo {author} {\bibfnamefont {S.~W.}\ \bibnamefont {Nam}}, \bibinfo {author} {\bibfnamefont {T.}~\bibnamefont {Gerrits}},\ and\ \bibinfo {author} {\bibfnamefont {C.}~\bibnamefont {Silberhorn}},\ }\bibfield  {title} {\bibinfo {title} {Scalability of parametric down-conversion for generating higher-order fock states},\ }\href@noop {} {\bibfield  {journal} {\bibinfo  {journal} {Phys. Rev. A}\ }\textbf {\bibinfo {volume} {100}},\ \bibinfo {pages} {041802} (\bibinfo {year} {2019})}\BibitemShut {NoStop}%
\bibitem [{\citenamefont {Sonoyama}\ \emph {et~al.}(2024)\citenamefont {Sonoyama}, \citenamefont {Takahashi}, \citenamefont {Sano}, \citenamefont {Suzuki}, \citenamefont {Nomura}, \citenamefont {Yabuno}, \citenamefont {Miki}, \citenamefont {Terai}, \citenamefont {Takse}, \citenamefont {Asavanant} \emph {et~al.}}]{sonoyama2024generation}%
  \BibitemOpen
  \bibfield  {author} {\bibinfo {author} {\bibfnamefont {T.}~\bibnamefont {Sonoyama}}, \bibinfo {author} {\bibfnamefont {K.}~\bibnamefont {Takahashi}}, \bibinfo {author} {\bibfnamefont {T.}~\bibnamefont {Sano}}, \bibinfo {author} {\bibfnamefont {T.}~\bibnamefont {Suzuki}}, \bibinfo {author} {\bibfnamefont {T.}~\bibnamefont {Nomura}}, \bibinfo {author} {\bibfnamefont {M.}~\bibnamefont {Yabuno}}, \bibinfo {author} {\bibfnamefont {S.}~\bibnamefont {Miki}}, \bibinfo {author} {\bibfnamefont {H.}~\bibnamefont {Terai}}, \bibinfo {author} {\bibfnamefont {K.}~\bibnamefont {Takse}}, \bibinfo {author} {\bibfnamefont {W.}~\bibnamefont {Asavanant}}, \emph {et~al.},\ }\bibfield  {title} {\bibinfo {title} {Generation of multi-photon fock states at telecommunication wavelength using picosecond pulsed light},\ }\href@noop {} {\bibfield  {journal} {\bibinfo  {journal} {arXiv preprint arXiv:2405.06567}\ } (\bibinfo {year} {2024})}\BibitemShut {NoStop}%
\bibitem [{\citenamefont {Konno}\ \emph {et~al.}(2024)\citenamefont {Konno}, \citenamefont {Asavanant}, \citenamefont {Hanamura}, \citenamefont {Nagayoshi}, \citenamefont {Fukui}, \citenamefont {Sakaguchi}, \citenamefont {Ide}, \citenamefont {China}, \citenamefont {Yabuno}, \citenamefont {Miki} \emph {et~al.}}]{konno2024logical}%
  \BibitemOpen
  \bibfield  {author} {\bibinfo {author} {\bibfnamefont {S.}~\bibnamefont {Konno}}, \bibinfo {author} {\bibfnamefont {W.}~\bibnamefont {Asavanant}}, \bibinfo {author} {\bibfnamefont {F.}~\bibnamefont {Hanamura}}, \bibinfo {author} {\bibfnamefont {H.}~\bibnamefont {Nagayoshi}}, \bibinfo {author} {\bibfnamefont {K.}~\bibnamefont {Fukui}}, \bibinfo {author} {\bibfnamefont {A.}~\bibnamefont {Sakaguchi}}, \bibinfo {author} {\bibfnamefont {R.}~\bibnamefont {Ide}}, \bibinfo {author} {\bibfnamefont {F.}~\bibnamefont {China}}, \bibinfo {author} {\bibfnamefont {M.}~\bibnamefont {Yabuno}}, \bibinfo {author} {\bibfnamefont {S.}~\bibnamefont {Miki}}, \emph {et~al.},\ }\bibfield  {title} {\bibinfo {title} {Logical states for fault-tolerant quantum computation with propagating light},\ }\href@noop {} {\bibfield  {journal} {\bibinfo  {journal} {Science}\ }\textbf {\bibinfo {volume} {383}},\ \bibinfo {pages} {289} (\bibinfo {year} {2024})}\BibitemShut {NoStop}%
\bibitem [{\citenamefont {Endo}\ \emph {et~al.}(2024)\citenamefont {Endo}, \citenamefont {Takahashi}, \citenamefont {Nomura}, \citenamefont {Sonoyama}, \citenamefont {Yabuno}, \citenamefont {Miki}, \citenamefont {Terai}, \citenamefont {Kashiwazaki}, \citenamefont {Inoue}, \citenamefont {Umeki} \emph {et~al.}}]{endo2024optically}%
  \BibitemOpen
  \bibfield  {author} {\bibinfo {author} {\bibfnamefont {M.}~\bibnamefont {Endo}}, \bibinfo {author} {\bibfnamefont {K.}~\bibnamefont {Takahashi}}, \bibinfo {author} {\bibfnamefont {T.}~\bibnamefont {Nomura}}, \bibinfo {author} {\bibfnamefont {T.}~\bibnamefont {Sonoyama}}, \bibinfo {author} {\bibfnamefont {M.}~\bibnamefont {Yabuno}}, \bibinfo {author} {\bibfnamefont {S.}~\bibnamefont {Miki}}, \bibinfo {author} {\bibfnamefont {H.}~\bibnamefont {Terai}}, \bibinfo {author} {\bibfnamefont {T.}~\bibnamefont {Kashiwazaki}}, \bibinfo {author} {\bibfnamefont {A.}~\bibnamefont {Inoue}}, \bibinfo {author} {\bibfnamefont {T.}~\bibnamefont {Umeki}}, \emph {et~al.},\ }\bibfield  {title} {\bibinfo {title} {Optically-sampled superconducting-nanostrip photon-number resolving detector for non-classical quantum state generation},\ }\href@noop {} {\bibfield  {journal} {\bibinfo  {journal} {arXiv preprint arXiv:2405.06901}\ } (\bibinfo {year} {2024})}\BibitemShut {NoStop}%
\bibitem [{\citenamefont {Asavanant}\ and\ \citenamefont {Furusawa}(2024)}]{asavanant2024multipartite}%
  \BibitemOpen
  \bibfield  {author} {\bibinfo {author} {\bibfnamefont {W.}~\bibnamefont {Asavanant}}\ and\ \bibinfo {author} {\bibfnamefont {A.}~\bibnamefont {Furusawa}},\ }\bibfield  {title} {\bibinfo {title} {Multipartite continuous-variable optical quantum entanglement: Generation and application},\ }\href@noop {} {\bibfield  {journal} {\bibinfo  {journal} {Physical Review A}\ }\textbf {\bibinfo {volume} {109}},\ \bibinfo {pages} {040101} (\bibinfo {year} {2024})}\BibitemShut {NoStop}%
\bibitem [{\citenamefont {Kawasaki}\ \emph {et~al.}(2024)\citenamefont {Kawasaki}, \citenamefont {Ide}, \citenamefont {Brunel}, \citenamefont {Suzuki}, \citenamefont {Nehra}, \citenamefont {Nakashima}, \citenamefont {Kashiwazaki}, \citenamefont {Inoue}, \citenamefont {Umeki}, \citenamefont {China} \emph {et~al.}}]{kawasaki2024high}%
  \BibitemOpen
  \bibfield  {author} {\bibinfo {author} {\bibfnamefont {A.}~\bibnamefont {Kawasaki}}, \bibinfo {author} {\bibfnamefont {R.}~\bibnamefont {Ide}}, \bibinfo {author} {\bibfnamefont {H.}~\bibnamefont {Brunel}}, \bibinfo {author} {\bibfnamefont {T.}~\bibnamefont {Suzuki}}, \bibinfo {author} {\bibfnamefont {R.}~\bibnamefont {Nehra}}, \bibinfo {author} {\bibfnamefont {K.}~\bibnamefont {Nakashima}}, \bibinfo {author} {\bibfnamefont {T.}~\bibnamefont {Kashiwazaki}}, \bibinfo {author} {\bibfnamefont {A.}~\bibnamefont {Inoue}}, \bibinfo {author} {\bibfnamefont {T.}~\bibnamefont {Umeki}}, \bibinfo {author} {\bibfnamefont {F.}~\bibnamefont {China}}, \emph {et~al.},\ }\bibfield  {title} {\bibinfo {title} {High-rate generation and state tomography of non-gaussian quantum states for ultra-fast clock frequency quantum processors},\ }\href@noop {} {\bibfield  {journal} {\bibinfo  {journal} {arXiv preprint arXiv:2402.17408}\ } (\bibinfo {year} {2024})}\BibitemShut {NoStop}%
\bibitem [{\citenamefont {Wiesner}\ \emph {et~al.}(2024)\citenamefont {Wiesner}, \citenamefont {Chrzanowski}, \citenamefont {Pieplow}, \citenamefont {Schr{\"o}der}, \citenamefont {Pappa},\ and\ \citenamefont {Wolters}}]{wiesner2024influence}%
  \BibitemOpen
  \bibfield  {author} {\bibinfo {author} {\bibfnamefont {F.}~\bibnamefont {Wiesner}}, \bibinfo {author} {\bibfnamefont {H.~M.}\ \bibnamefont {Chrzanowski}}, \bibinfo {author} {\bibfnamefont {G.}~\bibnamefont {Pieplow}}, \bibinfo {author} {\bibfnamefont {T.}~\bibnamefont {Schr{\"o}der}}, \bibinfo {author} {\bibfnamefont {A.}~\bibnamefont {Pappa}},\ and\ \bibinfo {author} {\bibfnamefont {J.}~\bibnamefont {Wolters}},\ }\bibfield  {title} {\bibinfo {title} {The influence of experimental imperfections on photonic ghz state generation},\ }\href@noop {} {\bibfield  {journal} {\bibinfo  {journal} {arXiv preprint arXiv:2406.18257}\ } (\bibinfo {year} {2024})}\BibitemShut {NoStop}%
\bibitem [{\citenamefont {Ra}\ \emph {et~al.}(2020)\citenamefont {Ra}, \citenamefont {Dufour}, \citenamefont {Walschaers}, \citenamefont {Jacquard}, \citenamefont {Michel}, \citenamefont {Fabre},\ and\ \citenamefont {Treps}}]{ra2020non}%
  \BibitemOpen
  \bibfield  {author} {\bibinfo {author} {\bibfnamefont {Y.-S.}\ \bibnamefont {Ra}}, \bibinfo {author} {\bibfnamefont {A.}~\bibnamefont {Dufour}}, \bibinfo {author} {\bibfnamefont {M.}~\bibnamefont {Walschaers}}, \bibinfo {author} {\bibfnamefont {C.}~\bibnamefont {Jacquard}}, \bibinfo {author} {\bibfnamefont {T.}~\bibnamefont {Michel}}, \bibinfo {author} {\bibfnamefont {C.}~\bibnamefont {Fabre}},\ and\ \bibinfo {author} {\bibfnamefont {N.}~\bibnamefont {Treps}},\ }\bibfield  {title} {\bibinfo {title} {Non-gaussian quantum states of a multimode light field},\ }\href@noop {} {\bibfield  {journal} {\bibinfo  {journal} {Nature Physics}\ }\textbf {\bibinfo {volume} {16}},\ \bibinfo {pages} {144} (\bibinfo {year} {2020})}\BibitemShut {NoStop}%
\bibitem [{\citenamefont {Kruse}\ \emph {et~al.}(2019)\citenamefont {Kruse}, \citenamefont {Hamilton}, \citenamefont {Sansoni}, \citenamefont {Barkhofen}, \citenamefont {Silberhorn},\ and\ \citenamefont {Jex}}]{kruse2019detailed}%
  \BibitemOpen
  \bibfield  {author} {\bibinfo {author} {\bibfnamefont {R.}~\bibnamefont {Kruse}}, \bibinfo {author} {\bibfnamefont {C.~S.}\ \bibnamefont {Hamilton}}, \bibinfo {author} {\bibfnamefont {L.}~\bibnamefont {Sansoni}}, \bibinfo {author} {\bibfnamefont {S.}~\bibnamefont {Barkhofen}}, \bibinfo {author} {\bibfnamefont {C.}~\bibnamefont {Silberhorn}},\ and\ \bibinfo {author} {\bibfnamefont {I.}~\bibnamefont {Jex}},\ }\bibfield  {title} {\bibinfo {title} {Detailed study of gaussian boson sampling},\ }\href@noop {} {\bibfield  {journal} {\bibinfo  {journal} {Physical Review A}\ }\textbf {\bibinfo {volume} {100}},\ \bibinfo {pages} {032326} (\bibinfo {year} {2019})}\BibitemShut {NoStop}%
\bibitem [{\citenamefont {Hamilton}\ \emph {et~al.}(2017)\citenamefont {Hamilton}, \citenamefont {Kruse}, \citenamefont {Sansoni}, \citenamefont {Barkhofen}, \citenamefont {Silberhorn},\ and\ \citenamefont {Jex}}]{hamilton2017gaussian}%
  \BibitemOpen
  \bibfield  {author} {\bibinfo {author} {\bibfnamefont {C.~S.}\ \bibnamefont {Hamilton}}, \bibinfo {author} {\bibfnamefont {R.}~\bibnamefont {Kruse}}, \bibinfo {author} {\bibfnamefont {L.}~\bibnamefont {Sansoni}}, \bibinfo {author} {\bibfnamefont {S.}~\bibnamefont {Barkhofen}}, \bibinfo {author} {\bibfnamefont {C.}~\bibnamefont {Silberhorn}},\ and\ \bibinfo {author} {\bibfnamefont {I.}~\bibnamefont {Jex}},\ }\bibfield  {title} {\bibinfo {title} {Gaussian boson sampling},\ }\href@noop {} {\bibfield  {journal} {\bibinfo  {journal} {Physical review letters}\ }\textbf {\bibinfo {volume} {119}},\ \bibinfo {pages} {170501} (\bibinfo {year} {2017})}\BibitemShut {NoStop}%
\bibitem [{\citenamefont {Grier}\ \emph {et~al.}(2022)\citenamefont {Grier}, \citenamefont {Brod}, \citenamefont {Arrazola}, \citenamefont {de~Andrade~Alonso},\ and\ \citenamefont {Quesada}}]{grier2022complexity}%
  \BibitemOpen
  \bibfield  {author} {\bibinfo {author} {\bibfnamefont {D.}~\bibnamefont {Grier}}, \bibinfo {author} {\bibfnamefont {D.~J.}\ \bibnamefont {Brod}}, \bibinfo {author} {\bibfnamefont {J.~M.}\ \bibnamefont {Arrazola}}, \bibinfo {author} {\bibfnamefont {M.~B.}\ \bibnamefont {de~Andrade~Alonso}},\ and\ \bibinfo {author} {\bibfnamefont {N.}~\bibnamefont {Quesada}},\ }\bibfield  {title} {\bibinfo {title} {The complexity of bipartite gaussian boson sampling},\ }\href@noop {} {\bibfield  {journal} {\bibinfo  {journal} {Quantum}\ }\textbf {\bibinfo {volume} {6}},\ \bibinfo {pages} {863} (\bibinfo {year} {2022})}\BibitemShut {NoStop}%
\bibitem [{\citenamefont {Deshpande}\ \emph {et~al.}(2022)\citenamefont {Deshpande}, \citenamefont {Mehta}, \citenamefont {Vincent}, \citenamefont {Quesada}, \citenamefont {Hinsche}, \citenamefont {Ioannou}, \citenamefont {Madsen}, \citenamefont {Lavoie}, \citenamefont {Qi}, \citenamefont {Eisert} \emph {et~al.}}]{deshpande2022quantum}%
  \BibitemOpen
  \bibfield  {author} {\bibinfo {author} {\bibfnamefont {A.}~\bibnamefont {Deshpande}}, \bibinfo {author} {\bibfnamefont {A.}~\bibnamefont {Mehta}}, \bibinfo {author} {\bibfnamefont {T.}~\bibnamefont {Vincent}}, \bibinfo {author} {\bibfnamefont {N.}~\bibnamefont {Quesada}}, \bibinfo {author} {\bibfnamefont {M.}~\bibnamefont {Hinsche}}, \bibinfo {author} {\bibfnamefont {M.}~\bibnamefont {Ioannou}}, \bibinfo {author} {\bibfnamefont {L.}~\bibnamefont {Madsen}}, \bibinfo {author} {\bibfnamefont {J.}~\bibnamefont {Lavoie}}, \bibinfo {author} {\bibfnamefont {H.}~\bibnamefont {Qi}}, \bibinfo {author} {\bibfnamefont {J.}~\bibnamefont {Eisert}}, \emph {et~al.},\ }\bibfield  {title} {\bibinfo {title} {Quantum computational advantage via high-dimensional gaussian boson sampling},\ }\href@noop {} {\bibfield  {journal} {\bibinfo  {journal} {Science advances}\ }\textbf {\bibinfo {volume} {8}},\ \bibinfo {pages} {eabi7894} (\bibinfo {year} {2022})}\BibitemShut {NoStop}%
\bibitem [{\citenamefont {Go}\ \emph {et~al.}(2024)\citenamefont {Go}, \citenamefont {Oh}, \citenamefont {Jiang},\ and\ \citenamefont {Jeong}}]{go2024exploring}%
  \BibitemOpen
  \bibfield  {author} {\bibinfo {author} {\bibfnamefont {B.}~\bibnamefont {Go}}, \bibinfo {author} {\bibfnamefont {C.}~\bibnamefont {Oh}}, \bibinfo {author} {\bibfnamefont {L.}~\bibnamefont {Jiang}},\ and\ \bibinfo {author} {\bibfnamefont {H.}~\bibnamefont {Jeong}},\ }\bibfield  {title} {\bibinfo {title} {Exploring shallow-depth boson sampling: Toward a scalable quantum advantage},\ }\href@noop {} {\bibfield  {journal} {\bibinfo  {journal} {Physical Review A}\ }\textbf {\bibinfo {volume} {109}},\ \bibinfo {pages} {052613} (\bibinfo {year} {2024})}\BibitemShut {NoStop}%
\bibitem [{\citenamefont {Braunstein}(2005)}]{braunstein2005squeezing}%
  \BibitemOpen
  \bibfield  {author} {\bibinfo {author} {\bibfnamefont {S.~L.}\ \bibnamefont {Braunstein}},\ }\bibfield  {title} {\bibinfo {title} {Squeezing as an irreducible resource},\ }\href@noop {} {\bibfield  {journal} {\bibinfo  {journal} {Physical Review A—Atomic, Molecular, and Optical Physics}\ }\textbf {\bibinfo {volume} {71}},\ \bibinfo {pages} {055801} (\bibinfo {year} {2005})}\BibitemShut {NoStop}%
\bibitem [{\citenamefont {Thekkadath}\ \emph {et~al.}(2024)\citenamefont {Thekkadath}, \citenamefont {Houde}, \citenamefont {England}, \citenamefont {Bustard}, \citenamefont {Bouchard}, \citenamefont {Quesada},\ and\ \citenamefont {Sussman}}]{thekkadath2024gain}%
  \BibitemOpen
  \bibfield  {author} {\bibinfo {author} {\bibfnamefont {G.}~\bibnamefont {Thekkadath}}, \bibinfo {author} {\bibfnamefont {M.}~\bibnamefont {Houde}}, \bibinfo {author} {\bibfnamefont {D.}~\bibnamefont {England}}, \bibinfo {author} {\bibfnamefont {P.}~\bibnamefont {Bustard}}, \bibinfo {author} {\bibfnamefont {F.}~\bibnamefont {Bouchard}}, \bibinfo {author} {\bibfnamefont {N.}~\bibnamefont {Quesada}},\ and\ \bibinfo {author} {\bibfnamefont {B.}~\bibnamefont {Sussman}},\ }\bibfield  {title} {\bibinfo {title} {Gain-induced group delay in spontaneous parametric down-conversion},\ }\href@noop {} {\bibfield  {journal} {\bibinfo  {journal} {Physical Review Letters}\ }\textbf {\bibinfo {volume} {133}},\ \bibinfo {pages} {203601} (\bibinfo {year} {2024})}\BibitemShut {NoStop}%
\bibitem [{\citenamefont {Houde}\ and\ \citenamefont {Quesada}(2023)}]{houde2023waveguided}%
  \BibitemOpen
  \bibfield  {author} {\bibinfo {author} {\bibfnamefont {M.}~\bibnamefont {Houde}}\ and\ \bibinfo {author} {\bibfnamefont {N.}~\bibnamefont {Quesada}},\ }\bibfield  {title} {\bibinfo {title} {Waveguided sources of consistent, single-temporal-mode squeezed light: The good, the bad, and the ugly},\ }\href@noop {} {\bibfield  {journal} {\bibinfo  {journal} {AVS Quantum Science}\ }\textbf {\bibinfo {volume} {5}} (\bibinfo {year} {2023})}\BibitemShut {NoStop}%
\bibitem [{\citenamefont {de~Guise}\ \emph {et~al.}(2014)\citenamefont {de~Guise}, \citenamefont {Tan}, \citenamefont {Poulin},\ and\ \citenamefont {Sanders}}]{de2014coincidence}%
  \BibitemOpen
  \bibfield  {author} {\bibinfo {author} {\bibfnamefont {H.}~\bibnamefont {de~Guise}}, \bibinfo {author} {\bibfnamefont {S.-H.}\ \bibnamefont {Tan}}, \bibinfo {author} {\bibfnamefont {I.~P.}\ \bibnamefont {Poulin}},\ and\ \bibinfo {author} {\bibfnamefont {B.~C.}\ \bibnamefont {Sanders}},\ }\bibfield  {title} {\bibinfo {title} {Coincidence landscapes for three-channel linear optical networks},\ }\href@noop {} {\bibfield  {journal} {\bibinfo  {journal} {Physical Review A}\ }\textbf {\bibinfo {volume} {89}},\ \bibinfo {pages} {063819} (\bibinfo {year} {2014})}\BibitemShut {NoStop}%
\bibitem [{\citenamefont {Villalobos}\ \emph {et~al.}(2023)\citenamefont {Villalobos}, \citenamefont {Valencia},\ and\ \citenamefont {Botero}}]{villalobos2023quantum}%
  \BibitemOpen
  \bibfield  {author} {\bibinfo {author} {\bibfnamefont {M.~E.}\ \bibnamefont {Villalobos}}, \bibinfo {author} {\bibfnamefont {A.}~\bibnamefont {Valencia}},\ and\ \bibinfo {author} {\bibfnamefont {A.}~\bibnamefont {Botero}},\ }\bibfield  {title} {\bibinfo {title} {The quantum beam splitter with many partially indistinguishable photons: multiphotonic interference and asymptotic classical correspondence},\ }\href@noop {} {\bibfield  {journal} {\bibinfo  {journal} {arXiv preprint arXiv:2312.16774}\ } (\bibinfo {year} {2023})}\BibitemShut {NoStop}%
\bibitem [{\citenamefont {Nawrath}\ \emph {et~al.}(2019)\citenamefont {Nawrath}, \citenamefont {Olbrich}, \citenamefont {Paul}, \citenamefont {Portalupi}, \citenamefont {Jetter},\ and\ \citenamefont {Michler}}]{nawrath2019coherence}%
  \BibitemOpen
  \bibfield  {author} {\bibinfo {author} {\bibfnamefont {C.}~\bibnamefont {Nawrath}}, \bibinfo {author} {\bibfnamefont {F.}~\bibnamefont {Olbrich}}, \bibinfo {author} {\bibfnamefont {M.}~\bibnamefont {Paul}}, \bibinfo {author} {\bibfnamefont {S.}~\bibnamefont {Portalupi}}, \bibinfo {author} {\bibfnamefont {M.}~\bibnamefont {Jetter}},\ and\ \bibinfo {author} {\bibfnamefont {P.}~\bibnamefont {Michler}},\ }\bibfield  {title} {\bibinfo {title} {Coherence and indistinguishability of highly pure single photons from non-resonantly and resonantly excited telecom c-band quantum dots},\ }\href@noop {} {\bibfield  {journal} {\bibinfo  {journal} {Applied Physics Letters}\ }\textbf {\bibinfo {volume} {115}} (\bibinfo {year} {2019})}\BibitemShut {NoStop}%
\bibitem [{\citenamefont {Nawrath}\ \emph {et~al.}(2023)\citenamefont {Nawrath}, \citenamefont {Joos}, \citenamefont {Kolatschek}, \citenamefont {Bauer}, \citenamefont {Pruy}, \citenamefont {Hornung}, \citenamefont {Fischer}, \citenamefont {Huang}, \citenamefont {Vijayan}, \citenamefont {Sittig} \emph {et~al.}}]{nawrath2023bright}%
  \BibitemOpen
  \bibfield  {author} {\bibinfo {author} {\bibfnamefont {C.}~\bibnamefont {Nawrath}}, \bibinfo {author} {\bibfnamefont {R.}~\bibnamefont {Joos}}, \bibinfo {author} {\bibfnamefont {S.}~\bibnamefont {Kolatschek}}, \bibinfo {author} {\bibfnamefont {S.}~\bibnamefont {Bauer}}, \bibinfo {author} {\bibfnamefont {P.}~\bibnamefont {Pruy}}, \bibinfo {author} {\bibfnamefont {F.}~\bibnamefont {Hornung}}, \bibinfo {author} {\bibfnamefont {J.}~\bibnamefont {Fischer}}, \bibinfo {author} {\bibfnamefont {J.}~\bibnamefont {Huang}}, \bibinfo {author} {\bibfnamefont {P.}~\bibnamefont {Vijayan}}, \bibinfo {author} {\bibfnamefont {R.}~\bibnamefont {Sittig}}, \emph {et~al.},\ }\bibfield  {title} {\bibinfo {title} {Bright source of purcell-enhanced, triggered, single photons in the telecom c-band},\ }\href@noop {} {\bibfield  {journal} {\bibinfo  {journal} {Advanced Quantum Technologies}\ }\textbf {\bibinfo {volume} {6}},\ \bibinfo {pages} {2300111} (\bibinfo {year} {2023})}\BibitemShut {NoStop}%
\bibitem [{\citenamefont {Yu}\ \emph {et~al.}(2023)\citenamefont {Yu}, \citenamefont {Liu}, \citenamefont {Lee}, \citenamefont {Michler}, \citenamefont {Reitzenstein}, \citenamefont {Srinivasan}, \citenamefont {Waks},\ and\ \citenamefont {Liu}}]{yu2023telecom}%
  \BibitemOpen
  \bibfield  {author} {\bibinfo {author} {\bibfnamefont {Y.}~\bibnamefont {Yu}}, \bibinfo {author} {\bibfnamefont {S.}~\bibnamefont {Liu}}, \bibinfo {author} {\bibfnamefont {C.-M.}\ \bibnamefont {Lee}}, \bibinfo {author} {\bibfnamefont {P.}~\bibnamefont {Michler}}, \bibinfo {author} {\bibfnamefont {S.}~\bibnamefont {Reitzenstein}}, \bibinfo {author} {\bibfnamefont {K.}~\bibnamefont {Srinivasan}}, \bibinfo {author} {\bibfnamefont {E.}~\bibnamefont {Waks}},\ and\ \bibinfo {author} {\bibfnamefont {J.}~\bibnamefont {Liu}},\ }\bibfield  {title} {\bibinfo {title} {Telecom-band quantum dot technologies for long-distance quantum networks},\ }\href@noop {} {\bibfield  {journal} {\bibinfo  {journal} {Nature Nanotechnology}\ }\textbf {\bibinfo {volume} {18}},\ \bibinfo {pages} {1389} (\bibinfo {year} {2023})}\BibitemShut {NoStop}%
\bibitem [{\citenamefont {Phillips}\ \emph {et~al.}(2024)\citenamefont {Phillips}, \citenamefont {Brash}, \citenamefont {Godsland}, \citenamefont {Martin}, \citenamefont {Foster}, \citenamefont {Tomlinson}, \citenamefont {Dost}, \citenamefont {Babazadeh}, \citenamefont {Sala}, \citenamefont {Wilson} \emph {et~al.}}]{phillips2024purcell}%
  \BibitemOpen
  \bibfield  {author} {\bibinfo {author} {\bibfnamefont {C.~L.}\ \bibnamefont {Phillips}}, \bibinfo {author} {\bibfnamefont {A.~J.}\ \bibnamefont {Brash}}, \bibinfo {author} {\bibfnamefont {M.}~\bibnamefont {Godsland}}, \bibinfo {author} {\bibfnamefont {N.~J.}\ \bibnamefont {Martin}}, \bibinfo {author} {\bibfnamefont {A.}~\bibnamefont {Foster}}, \bibinfo {author} {\bibfnamefont {A.}~\bibnamefont {Tomlinson}}, \bibinfo {author} {\bibfnamefont {R.}~\bibnamefont {Dost}}, \bibinfo {author} {\bibfnamefont {N.}~\bibnamefont {Babazadeh}}, \bibinfo {author} {\bibfnamefont {E.~M.}\ \bibnamefont {Sala}}, \bibinfo {author} {\bibfnamefont {L.}~\bibnamefont {Wilson}}, \emph {et~al.},\ }\bibfield  {title} {\bibinfo {title} {Purcell-enhanced single photons at telecom wavelengths from a quantum dot in a photonic crystal cavity},\ }\href@noop {} {\bibfield  {journal} {\bibinfo  {journal} {Scientific Reports}\ }\textbf {\bibinfo {volume} {14}},\ \bibinfo {pages} {4450} (\bibinfo {year} {2024})}\BibitemShut {NoStop}%
\bibitem [{\citenamefont {Dusanowski}\ \emph {et~al.}(2022)\citenamefont {Dusanowski}, \citenamefont {Nawrath}, \citenamefont {Portalupi}, \citenamefont {Jetter}, \citenamefont {Huber}, \citenamefont {Klembt}, \citenamefont {Michler},\ and\ \citenamefont {H{\"o}fling}}]{dusanowski2022optical}%
  \BibitemOpen
  \bibfield  {author} {\bibinfo {author} {\bibfnamefont {{\L}.}~\bibnamefont {Dusanowski}}, \bibinfo {author} {\bibfnamefont {C.}~\bibnamefont {Nawrath}}, \bibinfo {author} {\bibfnamefont {S.~L.}\ \bibnamefont {Portalupi}}, \bibinfo {author} {\bibfnamefont {M.}~\bibnamefont {Jetter}}, \bibinfo {author} {\bibfnamefont {T.}~\bibnamefont {Huber}}, \bibinfo {author} {\bibfnamefont {S.}~\bibnamefont {Klembt}}, \bibinfo {author} {\bibfnamefont {P.}~\bibnamefont {Michler}},\ and\ \bibinfo {author} {\bibfnamefont {S.}~\bibnamefont {H{\"o}fling}},\ }\bibfield  {title} {\bibinfo {title} {Optical charge injection and coherent control of a quantum-dot spin-qubit emitting at telecom wavelengths},\ }\href@noop {} {\bibfield  {journal} {\bibinfo  {journal} {Nature communications}\ }\textbf {\bibinfo {volume} {13}},\ \bibinfo {pages} {748} (\bibinfo {year} {2022})}\BibitemShut {NoStop}%
\bibitem [{\citenamefont {Aaronson}\ and\ \citenamefont {Arkhipov}(2011)}]{aaronson2011computational}%
  \BibitemOpen
  \bibfield  {author} {\bibinfo {author} {\bibfnamefont {S.}~\bibnamefont {Aaronson}}\ and\ \bibinfo {author} {\bibfnamefont {A.}~\bibnamefont {Arkhipov}},\ }\bibfield  {title} {\bibinfo {title} {The computational complexity of linear optics},\ }in\ \href@noop {} {\emph {\bibinfo {booktitle} {Proceedings of the forty-third annual ACM symposium on Theory of computing}}}\ (\bibinfo {year} {2011})\ pp.\ \bibinfo {pages} {333--342}\BibitemShut {NoStop}%
\bibitem [{\citenamefont {Lund}\ \emph {et~al.}(2014)\citenamefont {Lund}, \citenamefont {Laing}, \citenamefont {Rahimi-Keshari}, \citenamefont {Rudolph}, \citenamefont {O’Brien},\ and\ \citenamefont {Ralph}}]{lund2014boson}%
  \BibitemOpen
  \bibfield  {author} {\bibinfo {author} {\bibfnamefont {A.~P.}\ \bibnamefont {Lund}}, \bibinfo {author} {\bibfnamefont {A.}~\bibnamefont {Laing}}, \bibinfo {author} {\bibfnamefont {S.}~\bibnamefont {Rahimi-Keshari}}, \bibinfo {author} {\bibfnamefont {T.}~\bibnamefont {Rudolph}}, \bibinfo {author} {\bibfnamefont {J.~L.}\ \bibnamefont {O’Brien}},\ and\ \bibinfo {author} {\bibfnamefont {T.~C.}\ \bibnamefont {Ralph}},\ }\bibfield  {title} {\bibinfo {title} {Boson sampling from a gaussian state},\ }\href@noop {} {\bibfield  {journal} {\bibinfo  {journal} {Physical review letters}\ }\textbf {\bibinfo {volume} {113}},\ \bibinfo {pages} {100502} (\bibinfo {year} {2014})}\BibitemShut {NoStop}%
\bibitem [{\citenamefont {Seron}\ \emph {et~al.}(2022)\citenamefont {Seron}, \citenamefont {Novo}, \citenamefont {Arkhipov},\ and\ \citenamefont {Cerf}}]{seron2022efficient}%
  \BibitemOpen
  \bibfield  {author} {\bibinfo {author} {\bibfnamefont {B.}~\bibnamefont {Seron}}, \bibinfo {author} {\bibfnamefont {L.}~\bibnamefont {Novo}}, \bibinfo {author} {\bibfnamefont {A.}~\bibnamefont {Arkhipov}},\ and\ \bibinfo {author} {\bibfnamefont {N.~J.}\ \bibnamefont {Cerf}},\ }\bibfield  {title} {\bibinfo {title} {Efficient validation of boson sampling from binned photon-number distributions},\ }\href@noop {} {\bibfield  {journal} {\bibinfo  {journal} {arXiv preprint arXiv:2212.09643}\ } (\bibinfo {year} {2022})}\BibitemShut {NoStop}%
\bibitem [{\citenamefont {Bressanini}\ \emph {et~al.}(2023)\citenamefont {Bressanini}, \citenamefont {Seron}, \citenamefont {Novo}, \citenamefont {Cerf},\ and\ \citenamefont {Kim}}]{bressanini2023gaussian}%
  \BibitemOpen
  \bibfield  {author} {\bibinfo {author} {\bibfnamefont {G.}~\bibnamefont {Bressanini}}, \bibinfo {author} {\bibfnamefont {B.}~\bibnamefont {Seron}}, \bibinfo {author} {\bibfnamefont {L.}~\bibnamefont {Novo}}, \bibinfo {author} {\bibfnamefont {N.~J.}\ \bibnamefont {Cerf}},\ and\ \bibinfo {author} {\bibfnamefont {M.}~\bibnamefont {Kim}},\ }\bibfield  {title} {\bibinfo {title} {Gaussian boson sampling validation via detector binning},\ }\href@noop {} {\bibfield  {journal} {\bibinfo  {journal} {arXiv preprint arXiv:2310.18113}\ } (\bibinfo {year} {2023})}\BibitemShut {NoStop}%
\bibitem [{\citenamefont {Madsen}\ \emph {et~al.}(2022)\citenamefont {Madsen}, \citenamefont {Laudenbach}, \citenamefont {Askarani}, \citenamefont {Rortais}, \citenamefont {Vincent}, \citenamefont {Bulmer}, \citenamefont {Miatto}, \citenamefont {Neuhaus}, \citenamefont {Helt}, \citenamefont {Collins} \emph {et~al.}}]{madsen2022quantum}%
  \BibitemOpen
  \bibfield  {author} {\bibinfo {author} {\bibfnamefont {L.~S.}\ \bibnamefont {Madsen}}, \bibinfo {author} {\bibfnamefont {F.}~\bibnamefont {Laudenbach}}, \bibinfo {author} {\bibfnamefont {M.~F.}\ \bibnamefont {Askarani}}, \bibinfo {author} {\bibfnamefont {F.}~\bibnamefont {Rortais}}, \bibinfo {author} {\bibfnamefont {T.}~\bibnamefont {Vincent}}, \bibinfo {author} {\bibfnamefont {J.~F.}\ \bibnamefont {Bulmer}}, \bibinfo {author} {\bibfnamefont {F.~M.}\ \bibnamefont {Miatto}}, \bibinfo {author} {\bibfnamefont {L.}~\bibnamefont {Neuhaus}}, \bibinfo {author} {\bibfnamefont {L.~G.}\ \bibnamefont {Helt}}, \bibinfo {author} {\bibfnamefont {M.~J.}\ \bibnamefont {Collins}}, \emph {et~al.},\ }\bibfield  {title} {\bibinfo {title} {Quantum computational advantage with a programmable photonic processor},\ }\href@noop {} {\bibfield  {journal} {\bibinfo  {journal} {Nature}\ }\textbf {\bibinfo {volume} {606}},\ \bibinfo {pages} {75} (\bibinfo {year} {2022})}\BibitemShut {NoStop}%
\bibitem [{\citenamefont {Avagyan}\ \emph {et~al.}(2023)\citenamefont {Avagyan}, \citenamefont {Knill},\ and\ \citenamefont {Glancy}}]{avagyan2023multi}%
  \BibitemOpen
  \bibfield  {author} {\bibinfo {author} {\bibfnamefont {A.}~\bibnamefont {Avagyan}}, \bibinfo {author} {\bibfnamefont {E.}~\bibnamefont {Knill}},\ and\ \bibinfo {author} {\bibfnamefont {S.}~\bibnamefont {Glancy}},\ }\bibfield  {title} {\bibinfo {title} {Multi-mode gaussian state analysis with total-photon counting},\ }\href@noop {} {\bibfield  {journal} {\bibinfo  {journal} {Journal of Physics B: Atomic, Molecular and Optical Physics}\ }\textbf {\bibinfo {volume} {56}},\ \bibinfo {pages} {145501} (\bibinfo {year} {2023})}\BibitemShut {NoStop}%
\bibitem [{\citenamefont {Avagyan}\ \emph {et~al.}(2024)\citenamefont {Avagyan}, \citenamefont {Knill},\ and\ \citenamefont {Glancy}}]{avagyan2024multi}%
  \BibitemOpen
  \bibfield  {author} {\bibinfo {author} {\bibfnamefont {A.}~\bibnamefont {Avagyan}}, \bibinfo {author} {\bibfnamefont {E.}~\bibnamefont {Knill}},\ and\ \bibinfo {author} {\bibfnamefont {S.}~\bibnamefont {Glancy}},\ }\bibfield  {title} {\bibinfo {title} {Multi-mode gaussian state analysis with one bounded photon counter},\ }\href@noop {} {\bibfield  {journal} {\bibinfo  {journal} {arXiv preprint arXiv:2404.08957}\ } (\bibinfo {year} {2024})}\BibitemShut {NoStop}%
\bibitem [{\citenamefont {Dellios}\ \emph {et~al.}(2022{\natexlab{a}})\citenamefont {Dellios}, \citenamefont {Drummond}, \citenamefont {Opanchuk}, \citenamefont {Teh},\ and\ \citenamefont {Reid}}]{dellios2022simulating}%
  \BibitemOpen
  \bibfield  {author} {\bibinfo {author} {\bibfnamefont {A.}~\bibnamefont {Dellios}}, \bibinfo {author} {\bibfnamefont {P.~D.}\ \bibnamefont {Drummond}}, \bibinfo {author} {\bibfnamefont {B.}~\bibnamefont {Opanchuk}}, \bibinfo {author} {\bibfnamefont {R.~Y.}\ \bibnamefont {Teh}},\ and\ \bibinfo {author} {\bibfnamefont {M.~D.}\ \bibnamefont {Reid}},\ }\bibfield  {title} {\bibinfo {title} {Simulating macroscopic quantum correlations in linear networks},\ }\href@noop {} {\bibfield  {journal} {\bibinfo  {journal} {Physics Letters A}\ }\textbf {\bibinfo {volume} {429}},\ \bibinfo {pages} {127911} (\bibinfo {year} {2022}{\natexlab{a}})}\BibitemShut {NoStop}%
\bibitem [{\citenamefont {Dellios}\ \emph {et~al.}(2022{\natexlab{b}})\citenamefont {Dellios}, \citenamefont {Opanchuk}, \citenamefont {Reid},\ and\ \citenamefont {Drummond}}]{dellios2022validation}%
  \BibitemOpen
  \bibfield  {author} {\bibinfo {author} {\bibfnamefont {A.}~\bibnamefont {Dellios}}, \bibinfo {author} {\bibfnamefont {B.}~\bibnamefont {Opanchuk}}, \bibinfo {author} {\bibfnamefont {M.}~\bibnamefont {Reid}},\ and\ \bibinfo {author} {\bibfnamefont {P.}~\bibnamefont {Drummond}},\ }\bibfield  {title} {\bibinfo {title} {Validation tests for gbs quantum computers using grouped count probabilities (2023)},\ }\href@noop {} {\bibfield  {journal} {\bibinfo  {journal} {arXiv preprint arXiv:2211.03480}\ } (\bibinfo {year} {2022}{\natexlab{b}})}\BibitemShut {NoStop}%
\bibitem [{\citenamefont {Drummond}\ \emph {et~al.}(2022)\citenamefont {Drummond}, \citenamefont {Opanchuk}, \citenamefont {Dellios},\ and\ \citenamefont {Reid}}]{drummond2022simulating}%
  \BibitemOpen
  \bibfield  {author} {\bibinfo {author} {\bibfnamefont {P.~D.}\ \bibnamefont {Drummond}}, \bibinfo {author} {\bibfnamefont {B.}~\bibnamefont {Opanchuk}}, \bibinfo {author} {\bibfnamefont {A.}~\bibnamefont {Dellios}},\ and\ \bibinfo {author} {\bibfnamefont {M.~D.}\ \bibnamefont {Reid}},\ }\bibfield  {title} {\bibinfo {title} {Simulating complex networks in phase space: Gaussian boson sampling},\ }\href@noop {} {\bibfield  {journal} {\bibinfo  {journal} {Physical Review A}\ }\textbf {\bibinfo {volume} {105}},\ \bibinfo {pages} {012427} (\bibinfo {year} {2022})}\BibitemShut {NoStop}%
\bibitem [{\citenamefont {Mart{\'\i}nez-Cifuentes}\ \emph {et~al.}(2023)\citenamefont {Mart{\'\i}nez-Cifuentes}, \citenamefont {Fonseca-Romero},\ and\ \citenamefont {Quesada}}]{martinez2023classical}%
  \BibitemOpen
  \bibfield  {author} {\bibinfo {author} {\bibfnamefont {J.}~\bibnamefont {Mart{\'\i}nez-Cifuentes}}, \bibinfo {author} {\bibfnamefont {K.}~\bibnamefont {Fonseca-Romero}},\ and\ \bibinfo {author} {\bibfnamefont {N.}~\bibnamefont {Quesada}},\ }\bibfield  {title} {\bibinfo {title} {Classical models may be a better explanation of the jiuzhang 1.0 gaussian boson sampler than its targeted squeezed light model},\ }\href@noop {} {\bibfield  {journal} {\bibinfo  {journal} {Quantum}\ }\textbf {\bibinfo {volume} {7}},\ \bibinfo {pages} {1076} (\bibinfo {year} {2023})}\BibitemShut {NoStop}%
\bibitem [{\citenamefont {Dellios}\ \emph {et~al.}(2024)\citenamefont {Dellios}, \citenamefont {Reid},\ and\ \citenamefont {Drummond}}]{dellios2024validation}%
  \BibitemOpen
  \bibfield  {author} {\bibinfo {author} {\bibfnamefont {A.~S.}\ \bibnamefont {Dellios}}, \bibinfo {author} {\bibfnamefont {M.~D.}\ \bibnamefont {Reid}},\ and\ \bibinfo {author} {\bibfnamefont {P.~D.}\ \bibnamefont {Drummond}},\ }\bibfield  {title} {\bibinfo {title} {Validation tests of gaussian boson samplers with photon-number resolving detectors},\ }\href@noop {} {\bibfield  {journal} {\bibinfo  {journal} {arXiv preprint arXiv:2411.11228}\ } (\bibinfo {year} {2024})}\BibitemShut {NoStop}%
\bibitem [{\citenamefont {Mart{\'\i}nez-Cifuentes}\ \emph {et~al.}(2024)\citenamefont {Mart{\'\i}nez-Cifuentes}, \citenamefont {de~Guise},\ and\ \citenamefont {Quesada}}]{martinez2024linear}%
  \BibitemOpen
  \bibfield  {author} {\bibinfo {author} {\bibfnamefont {J.}~\bibnamefont {Mart{\'\i}nez-Cifuentes}}, \bibinfo {author} {\bibfnamefont {H.}~\bibnamefont {de~Guise}},\ and\ \bibinfo {author} {\bibfnamefont {N.}~\bibnamefont {Quesada}},\ }\bibfield  {title} {\bibinfo {title} {Linear cross-entropy certification of quantum computational advantage in gaussian boson sampling},\ }\href@noop {} {\bibfield  {journal} {\bibinfo  {journal} {PRX Quantum}\ }\textbf {\bibinfo {volume} {5}},\ \bibinfo {pages} {040312} (\bibinfo {year} {2024})}\BibitemShut {NoStop}%
\bibitem [{\citenamefont {Quesada}\ \emph {et~al.}(2019)\citenamefont {Quesada}, \citenamefont {Helt}, \citenamefont {Izaac}, \citenamefont {Arrazola}, \citenamefont {Shahrokhshahi}, \citenamefont {Myers},\ and\ \citenamefont {Sabapathy}}]{quesada2019simulating}%
  \BibitemOpen
  \bibfield  {author} {\bibinfo {author} {\bibfnamefont {N.}~\bibnamefont {Quesada}}, \bibinfo {author} {\bibfnamefont {L.}~\bibnamefont {Helt}}, \bibinfo {author} {\bibfnamefont {J.}~\bibnamefont {Izaac}}, \bibinfo {author} {\bibfnamefont {J.}~\bibnamefont {Arrazola}}, \bibinfo {author} {\bibfnamefont {R.}~\bibnamefont {Shahrokhshahi}}, \bibinfo {author} {\bibfnamefont {C.}~\bibnamefont {Myers}},\ and\ \bibinfo {author} {\bibfnamefont {K.}~\bibnamefont {Sabapathy}},\ }\bibfield  {title} {\bibinfo {title} {Simulating realistic non-gaussian state preparation},\ }\href@noop {} {\bibfield  {journal} {\bibinfo  {journal} {Physical Review A}\ }\textbf {\bibinfo {volume} {100}},\ \bibinfo {pages} {022341} (\bibinfo {year} {2019})}\BibitemShut {NoStop}%
\bibitem [{\citenamefont {Bj{\"o}rklund}\ \emph {et~al.}(2019)\citenamefont {Bj{\"o}rklund}, \citenamefont {Gupt},\ and\ \citenamefont {Quesada}}]{bjorklund2019faster}%
  \BibitemOpen
  \bibfield  {author} {\bibinfo {author} {\bibfnamefont {A.}~\bibnamefont {Bj{\"o}rklund}}, \bibinfo {author} {\bibfnamefont {B.}~\bibnamefont {Gupt}},\ and\ \bibinfo {author} {\bibfnamefont {N.}~\bibnamefont {Quesada}},\ }\bibfield  {title} {\bibinfo {title} {A faster hafnian formula for complex matrices and its benchmarking on a supercomputer},\ }\href@noop {} {\bibfield  {journal} {\bibinfo  {journal} {Journal of Experimental Algorithmics (JEA)}\ }\textbf {\bibinfo {volume} {24}},\ \bibinfo {pages} {1} (\bibinfo {year} {2019})}\BibitemShut {NoStop}%
\bibitem [{\citenamefont {Sylvester}(1883)}]{sylvester1883xxxix}%
  \BibitemOpen
  \bibfield  {author} {\bibinfo {author} {\bibfnamefont {J.~J.}\ \bibnamefont {Sylvester}},\ }\bibfield  {title} {\bibinfo {title} {Xxxix. on the equation to the secular inequalities in the planetary theory},\ }\href@noop {} {\bibfield  {journal} {\bibinfo  {journal} {The London, Edinburgh, and Dublin Philosophical Magazine and Journal of Science}\ }\textbf {\bibinfo {volume} {16}},\ \bibinfo {pages} {267} (\bibinfo {year} {1883})}\BibitemShut {NoStop}%
\bibitem [{\citenamefont {Pozrikidis}(2014)}]{pozrikidis2014introduction}%
  \BibitemOpen
  \bibfield  {author} {\bibinfo {author} {\bibfnamefont {C.}~\bibnamefont {Pozrikidis}},\ }\href@noop {} {\emph {\bibinfo {title} {An introduction to grids, graphs, and networks}}}\ (\bibinfo  {publisher} {Oxford University Press, USA},\ \bibinfo {year} {2014})\BibitemShut {NoStop}%
\bibitem [{\citenamefont {Kocharovsky}\ \emph {et~al.}(2022)\citenamefont {Kocharovsky}, \citenamefont {Kocharovsky},\ and\ \citenamefont {Tarasov}}]{kocharovsky2022hafnian}%
  \BibitemOpen
  \bibfield  {author} {\bibinfo {author} {\bibfnamefont {V.~V.}\ \bibnamefont {Kocharovsky}}, \bibinfo {author} {\bibfnamefont {V.~V.}\ \bibnamefont {Kocharovsky}},\ and\ \bibinfo {author} {\bibfnamefont {S.~V.}\ \bibnamefont {Tarasov}},\ }\bibfield  {title} {\bibinfo {title} {The hafnian master theorem},\ }\href@noop {} {\bibfield  {journal} {\bibinfo  {journal} {Linear Algebra and its Applications}\ } (\bibinfo {year} {2022})}\BibitemShut {NoStop}%
\bibitem [{\citenamefont {{The Walrus}}(2024)}]{lhaf2024walrus}%
  \BibitemOpen
  \bibfield  {author} {\bibinfo {author} {\bibnamefont {{The Walrus}}},\ }\href@noop {} {\bibinfo {title} {loophafnian}} (\bibinfo {year} {2024}),\ \bibinfo {note} {\url{https://github.com/XanaduAI/thewalrus/blob/master/octave/loophafnian.m}}\BibitemShut {NoStop}%
\bibitem [{\citenamefont {Cygan}\ and\ \citenamefont {Pilipczuk}(2015)}]{cygan2015faster}%
  \BibitemOpen
  \bibfield  {author} {\bibinfo {author} {\bibfnamefont {M.}~\bibnamefont {Cygan}}\ and\ \bibinfo {author} {\bibfnamefont {M.}~\bibnamefont {Pilipczuk}},\ }\bibfield  {title} {\bibinfo {title} {Faster exponential-time algorithms in graphs of bounded average degree},\ }\href@noop {} {\bibfield  {journal} {\bibinfo  {journal} {Information and Computation}\ }\textbf {\bibinfo {volume} {243}},\ \bibinfo {pages} {75} (\bibinfo {year} {2015})}\BibitemShut {NoStop}%
\bibitem [{\citenamefont {Bulmer}\ \emph {et~al.}(2022{\natexlab{b}})\citenamefont {Bulmer}, \citenamefont {Paesani}, \citenamefont {Chadwick},\ and\ \citenamefont {Quesada}}]{bulmer2022threshold}%
  \BibitemOpen
  \bibfield  {author} {\bibinfo {author} {\bibfnamefont {J.~F.}\ \bibnamefont {Bulmer}}, \bibinfo {author} {\bibfnamefont {S.}~\bibnamefont {Paesani}}, \bibinfo {author} {\bibfnamefont {R.~S.}\ \bibnamefont {Chadwick}},\ and\ \bibinfo {author} {\bibfnamefont {N.}~\bibnamefont {Quesada}},\ }\bibfield  {title} {\bibinfo {title} {Threshold detection statistics of bosonic states},\ }\href@noop {} {\bibfield  {journal} {\bibinfo  {journal} {Physical Review A}\ }\textbf {\bibinfo {volume} {106}},\ \bibinfo {pages} {043712} (\bibinfo {year} {2022}{\natexlab{b}})}\BibitemShut {NoStop}%
\bibitem [{\citenamefont {Bax}(1998)}]{bax1998finite}%
  \BibitemOpen
  \bibfield  {author} {\bibinfo {author} {\bibfnamefont {E.~T.}\ \bibnamefont {Bax}},\ }\emph {\bibinfo {title} {Finite-difference algorithms for counting problems}},\ \href@noop {} {Ph.D. thesis},\ \bibinfo  {school} {California Institute of Technology} (\bibinfo {year} {1998})\BibitemShut {NoStop}%
\bibitem [{\citenamefont {Ryser}(1963)}]{ryser1963combinatorial}%
  \BibitemOpen
  \bibfield  {author} {\bibinfo {author} {\bibfnamefont {H.~J.}\ \bibnamefont {Ryser}},\ }\href@noop {} {\emph {\bibinfo {title} {Combinatorial mathematics}}},\ Vol.~\bibinfo {volume} {14}\ (\bibinfo  {publisher} {American Mathematical Soc.},\ \bibinfo {year} {1963})\BibitemShut {NoStop}%
\bibitem [{\citenamefont {Glynn}(2010)}]{glynn2010permanent}%
  \BibitemOpen
  \bibfield  {author} {\bibinfo {author} {\bibfnamefont {D.~G.}\ \bibnamefont {Glynn}},\ }\bibfield  {title} {\bibinfo {title} {The permanent of a square matrix},\ }\href@noop {} {\bibfield  {journal} {\bibinfo  {journal} {European Journal of Combinatorics}\ }\textbf {\bibinfo {volume} {31}},\ \bibinfo {pages} {1887} (\bibinfo {year} {2010})}\BibitemShut {NoStop}%
\bibitem [{\citenamefont {Shchesnovich}(2020)}]{shchesnovich2020classical}%
  \BibitemOpen
  \bibfield  {author} {\bibinfo {author} {\bibfnamefont {V.}~\bibnamefont {Shchesnovich}},\ }\bibfield  {title} {\bibinfo {title} {On the classical complexity of sampling from quantum interference of indistinguishable bosons},\ }\href@noop {} {\bibfield  {journal} {\bibinfo  {journal} {International Journal of Quantum Information}\ }\textbf {\bibinfo {volume} {18}},\ \bibinfo {pages} {2050044} (\bibinfo {year} {2020})}\BibitemShut {NoStop}%
\bibitem [{\citenamefont {Kan}(2008)}]{kan2008moments}%
  \BibitemOpen
  \bibfield  {author} {\bibinfo {author} {\bibfnamefont {R.}~\bibnamefont {Kan}},\ }\bibfield  {title} {\bibinfo {title} {From moments of sum to moments of product},\ }\href@noop {} {\bibfield  {journal} {\bibinfo  {journal} {Journal of Multivariate Analysis}\ }\textbf {\bibinfo {volume} {99}},\ \bibinfo {pages} {542} (\bibinfo {year} {2008})}\BibitemShut {NoStop}%
\bibitem [{\citenamefont {Shi}\ and\ \citenamefont {Byrnes}(2022)}]{shi2022effect}%
  \BibitemOpen
  \bibfield  {author} {\bibinfo {author} {\bibfnamefont {J.}~\bibnamefont {Shi}}\ and\ \bibinfo {author} {\bibfnamefont {T.}~\bibnamefont {Byrnes}},\ }\bibfield  {title} {\bibinfo {title} {Effect of partial distinguishability on quantum supremacy in gaussian boson sampling},\ }\href@noop {} {\bibfield  {journal} {\bibinfo  {journal} {npj Quantum Information}\ }\textbf {\bibinfo {volume} {8}},\ \bibinfo {pages} {54} (\bibinfo {year} {2022})}\BibitemShut {NoStop}%
\bibitem [{\citenamefont {Ursell}(1927)}]{ursell_1927}%
  \BibitemOpen
  \bibfield  {author} {\bibinfo {author} {\bibfnamefont {H.~D.}\ \bibnamefont {Ursell}},\ }\bibfield  {title} {\bibinfo {title} {The evaluation of {G}ibbs' phase-integral for imperfect gases},\ }\href@noop {} {\bibfield  {journal} {\bibinfo  {journal} {Math. Proc. Cambridge Philos. Soc.}\ }\textbf {\bibinfo {volume} {23}},\ \bibinfo {pages} {685–697} (\bibinfo {year} {1927})}\BibitemShut {NoStop}%
\bibitem [{\citenamefont {Fisher}\ and\ \citenamefont {Wishart}(1932)}]{fisher1932derivation}%
  \BibitemOpen
  \bibfield  {author} {\bibinfo {author} {\bibfnamefont {R.~A.}\ \bibnamefont {Fisher}}\ and\ \bibinfo {author} {\bibfnamefont {J.}~\bibnamefont {Wishart}},\ }\bibfield  {title} {\bibinfo {title} {{The Derivation of the Pattern Formulae of Two-Way Partitions from those of Simpler Patterns}},\ }\href@noop {} {\bibfield  {journal} {\bibinfo  {journal} {Proc. London Math. Soc.}\ }\textbf {\bibinfo {volume} {s2-33}},\ \bibinfo {pages} {195} (\bibinfo {year} {1932})}\BibitemShut {NoStop}%
\bibitem [{\citenamefont {Su}\ \emph {et~al.}(2019)\citenamefont {Su}, \citenamefont {Myers},\ and\ \citenamefont {Sabapathy}}]{su2019conversion}%
  \BibitemOpen
  \bibfield  {author} {\bibinfo {author} {\bibfnamefont {D.}~\bibnamefont {Su}}, \bibinfo {author} {\bibfnamefont {C.~R.}\ \bibnamefont {Myers}},\ and\ \bibinfo {author} {\bibfnamefont {K.~K.}\ \bibnamefont {Sabapathy}},\ }\bibfield  {title} {\bibinfo {title} {Conversion of gaussian states to non-gaussian states using photon-number-resolving detectors},\ }\href@noop {} {\bibfield  {journal} {\bibinfo  {journal} {Physical Review A}\ }\textbf {\bibinfo {volume} {100}},\ \bibinfo {pages} {052301} (\bibinfo {year} {2019})}\BibitemShut {NoStop}%
\bibitem [{\citenamefont {Sabapathy}\ \emph {et~al.}(2019)\citenamefont {Sabapathy}, \citenamefont {Qi}, \citenamefont {Izaac},\ and\ \citenamefont {Weedbrook}}]{sabapathy2019production}%
  \BibitemOpen
  \bibfield  {author} {\bibinfo {author} {\bibfnamefont {K.~K.}\ \bibnamefont {Sabapathy}}, \bibinfo {author} {\bibfnamefont {H.}~\bibnamefont {Qi}}, \bibinfo {author} {\bibfnamefont {J.}~\bibnamefont {Izaac}},\ and\ \bibinfo {author} {\bibfnamefont {C.}~\bibnamefont {Weedbrook}},\ }\bibfield  {title} {\bibinfo {title} {Production of photonic universal quantum gates enhanced by machine learning},\ }\href@noop {} {\bibfield  {journal} {\bibinfo  {journal} {Physical Review A}\ }\textbf {\bibinfo {volume} {100}},\ \bibinfo {pages} {012326} (\bibinfo {year} {2019})}\BibitemShut {NoStop}%
\bibitem [{\citenamefont {Tarasov}(2025)}]{tarasov2025derivation}%
  \BibitemOpen
  \bibfield  {author} {\bibinfo {author} {\bibfnamefont {S.~V.}\ \bibnamefont {Tarasov}},\ }\bibfield  {title} {\bibinfo {title} {Derivation of the loop hafnian generating function for arbitrary symmetric matrices via gaussian integration},\ }\href@noop {} {\bibfield  {journal} {\bibinfo  {journal} {arXiv preprint arXiv:2507.16100}\ } (\bibinfo {year} {2025})}\BibitemShut {NoStop}%
\bibitem [{\citenamefont {Zhang}\ \emph {et~al.}(2008)\citenamefont {Zhang}, \citenamefont {Bao}, \citenamefont {Lu}, \citenamefont {Zhou}, \citenamefont {Yang}, \citenamefont {Rudolph},\ and\ \citenamefont {Pan}}]{zhang2008demonstration}%
  \BibitemOpen
  \bibfield  {author} {\bibinfo {author} {\bibfnamefont {Q.}~\bibnamefont {Zhang}}, \bibinfo {author} {\bibfnamefont {X.-H.}\ \bibnamefont {Bao}}, \bibinfo {author} {\bibfnamefont {C.-Y.}\ \bibnamefont {Lu}}, \bibinfo {author} {\bibfnamefont {X.-Q.}\ \bibnamefont {Zhou}}, \bibinfo {author} {\bibfnamefont {T.}~\bibnamefont {Yang}}, \bibinfo {author} {\bibfnamefont {T.}~\bibnamefont {Rudolph}},\ and\ \bibinfo {author} {\bibfnamefont {J.-W.}\ \bibnamefont {Pan}},\ }\bibfield  {title} {\bibinfo {title} {Demonstration of a scheme for the generation of “event-ready” entangled photon pairs from a single-photon source},\ }\href@noop {} {\bibfield  {journal} {\bibinfo  {journal} {Physical Review A—Atomic, Molecular, and Optical Physics}\ }\textbf {\bibinfo {volume} {77}},\ \bibinfo {pages} {062316} (\bibinfo {year} {2008})}\BibitemShut {NoStop}%
\bibitem [{\citenamefont {Tzitrin}\ and\ \citenamefont {Bourassa}(2024)}]{tzitrin2024gkp}%
  \BibitemOpen
  \bibfield  {author} {\bibinfo {author} {\bibfnamefont {I.}~\bibnamefont {Tzitrin}}\ and\ \bibinfo {author} {\bibfnamefont {J.~E.}\ \bibnamefont {Bourassa}},\ }\href@noop {} {\bibinfo {title} {{GKP} states}},\ \bibinfo {howpublished} {\url{https://the-walrus.readthedocs.io/en/latest/gallery/gkp.html}} (\bibinfo {year} {2024}),\ \bibinfo {note} {accessed: 2024-12-14}\BibitemShut {NoStop}%
\bibitem [{\citenamefont {Tzitrin}\ \emph {et~al.}(2020)\citenamefont {Tzitrin}, \citenamefont {Bourassa}, \citenamefont {Menicucci},\ and\ \citenamefont {Sabapathy}}]{tzitrin2020progress}%
  \BibitemOpen
  \bibfield  {author} {\bibinfo {author} {\bibfnamefont {I.}~\bibnamefont {Tzitrin}}, \bibinfo {author} {\bibfnamefont {J.~E.}\ \bibnamefont {Bourassa}}, \bibinfo {author} {\bibfnamefont {N.~C.}\ \bibnamefont {Menicucci}},\ and\ \bibinfo {author} {\bibfnamefont {K.~K.}\ \bibnamefont {Sabapathy}},\ }\bibfield  {title} {\bibinfo {title} {Progress towards practical qubit computation using approximate gottesman-kitaev-preskill codes},\ }\href@noop {} {\bibfield  {journal} {\bibinfo  {journal} {Physical Review A}\ }\textbf {\bibinfo {volume} {101}},\ \bibinfo {pages} {032315} (\bibinfo {year} {2020})}\BibitemShut {NoStop}%
\bibitem [{\citenamefont {Lam}\ \emph {et~al.}(2015)\citenamefont {Lam}, \citenamefont {Pitrou},\ and\ \citenamefont {Seibert}}]{lam2015numba}%
  \BibitemOpen
  \bibfield  {author} {\bibinfo {author} {\bibfnamefont {S.~K.}\ \bibnamefont {Lam}}, \bibinfo {author} {\bibfnamefont {A.}~\bibnamefont {Pitrou}},\ and\ \bibinfo {author} {\bibfnamefont {S.}~\bibnamefont {Seibert}},\ }\bibfield  {title} {\bibinfo {title} {Numba: A llvm-based python jit compiler},\ }in\ \href@noop {} {\emph {\bibinfo {booktitle} {Proceedings of the Second Workshop on the LLVM Compiler Infrastructure in HPC}}}\ (\bibinfo {year} {2015})\ pp.\ \bibinfo {pages} {1--6}\BibitemShut {NoStop}%
\bibitem [{\citenamefont {Ruiz-Gonzalez}\ \emph {et~al.}(2023)\citenamefont {Ruiz-Gonzalez}, \citenamefont {Arlt}, \citenamefont {Petermann}, \citenamefont {Sayyad}, \citenamefont {Jaouni}, \citenamefont {Karimi}, \citenamefont {Tischler}, \citenamefont {Gu},\ and\ \citenamefont {Krenn}}]{ruiz2023digital}%
  \BibitemOpen
  \bibfield  {author} {\bibinfo {author} {\bibfnamefont {C.}~\bibnamefont {Ruiz-Gonzalez}}, \bibinfo {author} {\bibfnamefont {S.}~\bibnamefont {Arlt}}, \bibinfo {author} {\bibfnamefont {J.}~\bibnamefont {Petermann}}, \bibinfo {author} {\bibfnamefont {S.}~\bibnamefont {Sayyad}}, \bibinfo {author} {\bibfnamefont {T.}~\bibnamefont {Jaouni}}, \bibinfo {author} {\bibfnamefont {E.}~\bibnamefont {Karimi}}, \bibinfo {author} {\bibfnamefont {N.}~\bibnamefont {Tischler}}, \bibinfo {author} {\bibfnamefont {X.}~\bibnamefont {Gu}},\ and\ \bibinfo {author} {\bibfnamefont {M.}~\bibnamefont {Krenn}},\ }\bibfield  {title} {\bibinfo {title} {Digital discovery of 100 diverse quantum experiments with pytheus},\ }\href@noop {} {\bibfield  {journal} {\bibinfo  {journal} {Quantum}\ }\textbf {\bibinfo {volume} {7}},\ \bibinfo {pages} {1204} (\bibinfo {year} {2023})}\BibitemShut {NoStop}%
\bibitem [{\citenamefont {Rudolph}(2023)}]{rudolph2023terry}%
  \BibitemOpen
  \bibfield  {author} {\bibinfo {author} {\bibfnamefont {T.}~\bibnamefont {Rudolph}},\ }\bibfield  {title} {\bibinfo {title} {Terry vs an ai, round 1: Heralding single-rail (approximate?) 4-ghz state from squeezed sources},\ }\href@noop {} {\bibfield  {journal} {\bibinfo  {journal} {arXiv preprint arXiv:2303.05514}\ } (\bibinfo {year} {2023})}\BibitemShut {NoStop}%
\bibitem [{\citenamefont {Yao}\ \emph {et~al.}(2024)\citenamefont {Yao}, \citenamefont {Miatto},\ and\ \citenamefont {Quesada}}]{yao2024riemannian}%
  \BibitemOpen
  \bibfield  {author} {\bibinfo {author} {\bibfnamefont {Y.}~\bibnamefont {Yao}}, \bibinfo {author} {\bibfnamefont {F.}~\bibnamefont {Miatto}},\ and\ \bibinfo {author} {\bibfnamefont {N.}~\bibnamefont {Quesada}},\ }\bibfield  {title} {\bibinfo {title} {Riemannian optimization of photonic quantum circuits in phase and fock space},\ }\href@noop {} {\bibfield  {journal} {\bibinfo  {journal} {SciPost Physics}\ }\textbf {\bibinfo {volume} {17}},\ \bibinfo {pages} {082} (\bibinfo {year} {2024})}\BibitemShut {NoStop}%
\bibitem [{\citenamefont {Kleinpa{\ss}}\ \emph {et~al.}(2024)\citenamefont {Kleinpa{\ss}}, \citenamefont {Walther},\ and\ \citenamefont {Fitzke}}]{kleinpass2024frequency}%
  \BibitemOpen
  \bibfield  {author} {\bibinfo {author} {\bibfnamefont {P.}~\bibnamefont {Kleinpa{\ss}}}, \bibinfo {author} {\bibfnamefont {T.}~\bibnamefont {Walther}},\ and\ \bibinfo {author} {\bibfnamefont {E.}~\bibnamefont {Fitzke}},\ }\bibfield  {title} {\bibinfo {title} {Frequency-resolved simulations of highly entangled biphoton states: Beyond the single-pair approximation. ii. application to entanglement-based quantum key distribution},\ }\href@noop {} {\bibfield  {journal} {\bibinfo  {journal} {arXiv preprint arXiv:2412.00958}\ } (\bibinfo {year} {2024})}\BibitemShut {NoStop}%
\bibitem [{\citenamefont {Steinmetz}\ \emph {et~al.}(2024)\citenamefont {Steinmetz}, \citenamefont {Ostmann}, \citenamefont {Neville}, \citenamefont {Pankovich},\ and\ \citenamefont {Sohbi}}]{steinmetz2024simulating}%
  \BibitemOpen
  \bibfield  {author} {\bibinfo {author} {\bibfnamefont {J.}~\bibnamefont {Steinmetz}}, \bibinfo {author} {\bibfnamefont {M.}~\bibnamefont {Ostmann}}, \bibinfo {author} {\bibfnamefont {A.}~\bibnamefont {Neville}}, \bibinfo {author} {\bibfnamefont {B.}~\bibnamefont {Pankovich}},\ and\ \bibinfo {author} {\bibfnamefont {A.}~\bibnamefont {Sohbi}},\ }\bibfield  {title} {\bibinfo {title} {Simulating imperfect quantum optical circuits using unsymmetrized bases},\ }\href@noop {} {\bibfield  {journal} {\bibinfo  {journal} {arXiv preprint arXiv:2412.13330}\ } (\bibinfo {year} {2024})}\BibitemShut {NoStop}%
\bibitem [{\citenamefont {Apostol}(2013)}]{apostol2013introduction}%
  \BibitemOpen
  \bibfield  {author} {\bibinfo {author} {\bibfnamefont {T.~M.}\ \bibnamefont {Apostol}},\ }\href@noop {} {\emph {\bibinfo {title} {Introduction to analytic number theory}}}\ (\bibinfo  {publisher} {Springer Science \& Business Media},\ \bibinfo {year} {2013})\BibitemShut {NoStop}%
\bibitem [{\citenamefont {{OEIS Foundation Inc.}}(2024)}]{oeis}%
  \BibitemOpen
  \bibfield  {author} {\bibinfo {author} {\bibnamefont {{OEIS Foundation Inc.}}},\ }\href@noop {} {\bibinfo {title} {The {O}n-{L}ine {E}ncyclopedia of {I}nteger {S}equences}} (\bibinfo {year} {2024}),\ \bibinfo {note} {published electronically at \url{http://oeis.org}}\BibitemShut {NoStop}%
\bibitem [{\citenamefont {Piela}(2006)}]{piela2006ideas}%
  \BibitemOpen
  \bibfield  {author} {\bibinfo {author} {\bibfnamefont {L.}~\bibnamefont {Piela}},\ }\href@noop {} {\emph {\bibinfo {title} {Ideas of quantum chemistry}}}\ (\bibinfo  {publisher} {Elsevier},\ \bibinfo {year} {2006})\BibitemShut {NoStop}%
\bibitem [{\citenamefont {Quesada}\ \emph {et~al.}(2022)\citenamefont {Quesada}, \citenamefont {Helt}, \citenamefont {Menotti}, \citenamefont {Liscidini},\ and\ \citenamefont {Sipe}}]{quesada2022beyond}%
  \BibitemOpen
  \bibfield  {author} {\bibinfo {author} {\bibfnamefont {N.}~\bibnamefont {Quesada}}, \bibinfo {author} {\bibfnamefont {L.}~\bibnamefont {Helt}}, \bibinfo {author} {\bibfnamefont {M.}~\bibnamefont {Menotti}}, \bibinfo {author} {\bibfnamefont {M.}~\bibnamefont {Liscidini}},\ and\ \bibinfo {author} {\bibfnamefont {J.}~\bibnamefont {Sipe}},\ }\bibfield  {title} {\bibinfo {title} {Beyond photon pairs—nonlinear quantum photonics in the high-gain regime: a tutorial},\ }\href@noop {} {\bibfield  {journal} {\bibinfo  {journal} {Advances in Optics and Photonics}\ }\textbf {\bibinfo {volume} {14}},\ \bibinfo {pages} {291} (\bibinfo {year} {2022})}\BibitemShut {NoStop}%
\bibitem [{\citenamefont {Houde}\ \emph {et~al.}(2024)\citenamefont {Houde}, \citenamefont {McCutcheon},\ and\ \citenamefont {Quesada}}]{houde2024matrix}%
  \BibitemOpen
  \bibfield  {author} {\bibinfo {author} {\bibfnamefont {M.}~\bibnamefont {Houde}}, \bibinfo {author} {\bibfnamefont {W.}~\bibnamefont {McCutcheon}},\ and\ \bibinfo {author} {\bibfnamefont {N.}~\bibnamefont {Quesada}},\ }\bibfield  {title} {\bibinfo {title} {Matrix decompositions in quantum optics: Takagi/autonne, bloch--messiah/euler, iwasawa, and williamson},\ }\href@noop {} {\bibfield  {journal} {\bibinfo  {journal} {Canadian Journal of Physics}\ }\textbf {\bibinfo {volume} {102}},\ \bibinfo {pages} {497} (\bibinfo {year} {2024})}\BibitemShut {NoStop}%
\bibitem [{\citenamefont {Nijenhuis}\ and\ \citenamefont {Wilf}(2014)}]{nijenhuis2014combinatorial}%
  \BibitemOpen
  \bibfield  {author} {\bibinfo {author} {\bibfnamefont {A.}~\bibnamefont {Nijenhuis}}\ and\ \bibinfo {author} {\bibfnamefont {H.~S.}\ \bibnamefont {Wilf}},\ }\href@noop {} {\emph {\bibinfo {title} {Combinatorial algorithms: for computers and calculators}}}\ (\bibinfo  {publisher} {Elsevier},\ \bibinfo {year} {2014})\BibitemShut {NoStop}%
\bibitem [{\citenamefont {Hage}(2003)}]{hage2003enumerating}%
  \BibitemOpen
  \bibfield  {author} {\bibinfo {author} {\bibfnamefont {J.}~\bibnamefont {Hage}},\ }\bibfield  {title} {\bibinfo {title} {Enumerating submultisets of multisets},\ }\href@noop {} {\bibfield  {journal} {\bibinfo  {journal} {Information processing letters}\ }\textbf {\bibinfo {volume} {85}},\ \bibinfo {pages} {221} (\bibinfo {year} {2003})}\BibitemShut {NoStop}%
\bibitem [{\citenamefont {Guan}(1998)}]{guan1998generalized}%
  \BibitemOpen
  \bibfield  {author} {\bibinfo {author} {\bibfnamefont {D.-J.}\ \bibnamefont {Guan}},\ }\bibfield  {title} {\bibinfo {title} {Generalized gray codes with applications},\ }in\ \href@noop {} {\emph {\bibinfo {booktitle} {Proc. Natl. Sci. Counc. Repub. China Part A Phys. Sci. Eng.}}},\ Vol.~\bibinfo {volume} {22}\ (\bibinfo {organization} {Citeseer},\ \bibinfo {year} {1998})\ pp.\ \bibinfo {pages} {841--848}\BibitemShut {NoStop}%
\end{thebibliography}%

    \onecolumngrid
    
    \appendix

    \section{\label{app:algorithm} A proof of the algorithm for computing the \textit{f}-function}

        From Eq.~\eqref{eq:truncated_series_q} let us define 
        \begin{equation}
            g_k = \frac{\tr\left([\bm{X} \bm{A}]^k\right)}{2k} + \frac{\gl^\mathrm{T} [\bm{X}\bm{A}]^{k-1} \bm{X} \gl }{2}.
            \label{eq:f_term}
        \end{equation}
        We are interested in finding an expression for the derivatives 
        \begin{equation}
            \frac{1}{N!}\frac{\partial^N}{\partial \eta^N} \left. \sum_{j=0}^N \frac{1}{j!}  \left( \sum_{k=1}^N g_k\eta^k \right)^j \right|_{\eta=0}.
            \label{eq:expression_of_interest}
        \end{equation}
        To do this, let us consider first the multinomial theorem: let $\{x_1,\dots,x_m\}\in\mathbb{C}$, for every non-negative integer $n$, we have
        \begin{equation}
            \left(\sum_{i=1}^m x_i\right)^n = n!\sum_{\bm{a}}\prod_{j=1}^m\frac{x_j^{a_j}}{a_j!},
            \label{eq:multinomial_theorem}
        \end{equation}
        where $\bm{a}=(a_1,\dots,a_m)$ is a vector of non-negative integers satisfying the relation $a_1+\dots+a_m = n$.
        
        Using this theorem on the sum in Eq.~\eqref{eq:expression_of_interest} we obtain 
        \begin{equation}
            \sum_{j=0}^N\frac{1}{j!}\left(\sum_{k=1}^Ng_k\eta^k\right)^j=\sum_{j=0}^{N}\sum_{\bm{a}_j}\prod_{k=1}^N\frac{g_k^{a_k}\eta^{ka_k}}{a_k!},
            \label{eq:multinomial_sum_1}
        \end{equation}
        where we write $\bm{a}_j=(a_1,\dots,a_N)$, whose components satisfy $a_1+\dots+a_N=j$. Notice that this expression is a polynomial of degree $N^2$ in $\eta$, and for a fixed $\bm{a}_j$, $\eta$ will have the exponent $a_1+2a_2+\dots+Na_N$.
        
        Notice also that, for a given $\bm{a}_j$, $j\leq a_1 +2a_2+\dots+Na_N \leq jN$. To see this, first note that $j=a_1+\dots+a_N\leq a_1+\dots+Na_N$, because all the $\{a_k\}$ are non-negative integers, and in the sum $a_1+\dots +Na_N$ each of them is multiplied by a positive number. The upper bound can be proved by considering that for every $k\in\{1,\dots,N\}$, $ka_k\leq Na_k$. This implies that $a_1 + \dots +Na_N \leq N(a_1+\dots+a_N)=jN$.
        
        Keeping the above considerations in mind, we can reorganize Eq.~\eqref{eq:multinomial_sum_1} as follows. Let the non-negative integer components of $\bm{a}_{jl}=(a_1,\dots,a_N)$ satisfy the constraints $a_1+\dots+a_N = j$ and $a_1+\dots+Na_N = l$. Then, we can write 
        \begin{align}
            \sum_{j=0}^N\frac{1}{j!}\left(\sum_{k=1}^Ng_k\eta^k\right)^j&=\sum_{j=0}^{N}\sum_{\bm{a}_j}\prod_{k=1}^N\frac{g_k^{a_k}\eta^{ka_k}}{a_k!}=\sum_{j=0}^N\sum_{l=j}^{jN}\left(\sum_{\bm{a}_{jl}}\prod_{k=1}^N\frac{g_k^{a_k}}{a_k!}\eta^l\right)=\sum_{j=0}^N\left(\sum_{l=j}^{jN}h_{jl}\eta^l\right),
            \label{eq:multinomial_sum_2}
        \end{align}
        where we have introduced the term
        \begin{equation}
            h_{jl}=\sum_{\bm{a}_{jl}}\prod_{k=1}^N\frac{g_k^{a_k}}{a_k!}.
            \label{eq:partition_sum_1}
        \end{equation}
        
        We want now to interchange the summations over $j$ and $l$. We can do this considering the following argument. For every $l$, we can find a $k\in\{1,\dots,N\}$ such that $(k-1)N+1\leq l\leq kN$. This value of $k$ can be interpreted as the minimum value of $j$ that allows the terms with $l\geq (k-1)N+1$ to appear in the last sum in Eq.~\eqref{eq:multinomial_sum_2}. Indeed, if we only had $j<k$, the sum would only consist of terms with $l$ up to $(k-1)N$. Now, we have
        \begin{align*}
            (k-1)N+1\leq l\leq kN\Longrightarrow(k-1)N\leq l-1<kN\Longrightarrow k-1\leq (l-1)/N<k,
        \end{align*}
        which implies that $k=\lfloor(l-1)/N\rfloor + 1$.
        
        When grouping all the terms in the sum according to the values of $\eta^l$, that is, when gathering the different $\{h_{jl}\}$ according to their corresponding $l$, we can see that the minimum value of $j$ we can find in the sum, for a given $l$, is precisely $\lfloor(l-1)/N\rfloor + 1$. The maximum value of $j$ depends on $l$ in the following way: since $j\leq l$, if $l<N$ the maximum value $j$ can take is $l$. If $l\geq N$ the maximum value of $j$ is $N$ (recall that $0\leq j\leq N$). Therefore, we can say that for a fixed $l$, the allowed values of $j$ run from $\lfloor(l-1)/N\rfloor + 1$ to $\mathrm{min}(l,N)$. Finally, considering that $l$ can take values between $1$ and $N^2$, we can write
        \begin{align}
            \sum_{j=0}^N\frac{1}{j!}\left(\sum_{k=1}^Ng_k\eta^k\right)^j=\sum_{j=0}^N\left(\sum_{l=j}^{jN}h_{jl}\eta^l\right) = h_{00} + \sum_{l=1}^{N^2}\left(\sum_{j=\lfloor(l-1)/N\rfloor + 1}^{\mathrm{min}(l,N)}h_{jl}\right)\eta^l.
            \label{eq:multinomial_sum_3}
        \end{align}
        
        From Eq.~\eqref{eq:multinomial_sum_3}, it can be readily seen that
        \begin{equation}
            \frac{1}{N!}\frac{\partial^N}{\partial \eta^N} \left. \sum_{j=0}^N \frac{1}{j!}  \left( \sum_{k=1}^N g_k\eta^k \right)^j \right|_{\eta=0} = \sum_{j=\lfloor(N-1)/N\rfloor + 1}^{\mathrm{min}(N,N)}\!\!\!\!\!\!h_{jl}=\sum_{j=1}^{N}h_{jN} = \sum_{j=1}^{N}\sum_{\bm{a}_{jN}}\prod_{k=1}^N\frac{g_k^{a_k}}{a_k!},  \label{eq:expression_of_interest_2}
        \end{equation}
        where we took into account that $\lfloor(N-1)/N\rfloor + 1=1$ and $\mathrm{\min}(N,N)=N$.
        
        Recall that an \textit{integer partition} $\lambda$ of a positive integer $n$ (denoted $\lambda\vdash n$) is a sequence of positive integers $(\lambda_1,\lambda_2,\dots)$ that satisfy $\lambda_1\geq\lambda_2\geq\dots\geq 1$ and $\lambda_1+\lambda_2+\dots=n$. To every $\lambda$ we can assign a unique vector of non-negative integers $\vec{\mu}(\lambda)=(\mu_1,\mu_2,\dots,\mu_n)$ satisfying $\mu_1+2\mu_2+\dots+n\mu_n=n$. Indeed, we need only define each $\mu_k$ as the number of times that $k$ appears in $\lambda$ or, in other words, $\mu_k$ is defined as the multiplicity of $k$ in $\lambda$ ($\mu_k=0$ if $k$ does not appear in $\lambda$). Notice that since $\mu_k$ counts the number of times a given integer appears in $\lambda$, the sum $\mu_1+\dots+\mu_n$ counts the total number of elements in $\lambda$, i.e., the length of the partition $|\lambda|$. For any $\lambda$, it holds that $1\leq|\lambda|\leq n$. Finally, let us define $\delta(\lambda)!=\mu_1!\cdots\mu_n!$. 
        
        We can see that the vectors $\bm{a}_{jN}$ define a unique partition of $N$ with length $j$. In this case, each component $a_k$ corresponds to the multiplicity of $k$ in the partition defined by $\bm{a}_{jN}$. Therefore, we can make the correspondence $\vec{a}_{jN}\leftrightarrow \lambda\vdash N\text{ s.t. }|\lambda|=j$ and $a_1!\cdots a_n!\leftrightarrow \delta(\lambda)!$. This allows us to write Eq.~\eqref{eq:expression_of_interest_2} as
        \begin{equation}
            \frac{1}{N!}\frac{\partial^N}{\partial \eta^N} \left. \sum_{j=0}^N \frac{1}{j!}  \left( \sum_{k=1}^N g_k\eta^k \right)^j \right|_{\eta=0} = \sum_{j=1}^{N}\sum_{\substack{\lambda\vdash N \\ \mathrm{s.t.}\,|\lambda|=j}}\frac{1}{\delta(\lambda)!}\prod_{a\in\lambda}g_a = \sum_{\lambda \vdash N}\frac{1}{\delta(\lambda)!}\prod_{a\in\lambda}g_a.
            \label{eq:expression_of_interest_4}
        \end{equation}
        The last equality follows from noticing that every $\lambda$ has a length between $1$ and $N$, therefore 
        \[\sum_{j=1}^{N}\sum_{\substack{\lambda\vdash N \\ \mathrm{s.t.}\,|\lambda|=j}}=\sum_{\lambda \vdash N}.\]
        
        We are now in the position to give an explanation to the algorithm shown in~\cite{lhaf2024walrus} (it is worth mentioning that this is not the only possible explanation, Ref.~\cite{bjorklund2019faster} interprets the algorithm in terms of combinations of walks in a graph). In principle, we would need to compute all the integer partitions of $N$ to compute the sum in Eq.~\eqref{eq:expression_of_interest_4}. However, this algorithm takes into account that we are interested only on the \textit{sum of over all partitions}, not in generating each individual partition, and that the sums corresponding to $k<N$ contain information that can be used to obtain the sum over partitions of $N$.  
        
        Let $\bm{c}_{i}$ be vectors of $N+1$ components. We will use these vectors to store the sums over partitions for $1\leq l\leq N$. The idea is that for a given value of $i$, the component $(\bm{c}_{i})_l$ receives all the contributions to the sum $\sum_{\lambda\vdash (l-1)}\frac{1}{\delta(\lambda)!}\prod_{a\in\lambda}g_a$ from partitions $\lambda\vdash (l-1)$ whose summands take values of \textit{at most} $i$ (in the definition of $\lambda$ given above, this means that $\lambda_1$ is at most $i$). We set the initial condition $\bm{c}_{0}=(1,0,\dots,0)$.
        
        For an index $i\in\{1,\dots,N\}$, we compute $g_i$ and determine the maximum amount of times that $i$ can appear in a partition of $N$. This number is equal to the quotient between $N$ and $i$, $q(N,i)$. We also set an accumulation factor $p$, that will collect the different terms involved in the computation of $\frac{1}{\delta(\lambda)!}\prod_{a\in\lambda}g_a$. For each $i$, $p$ is initialized at the value $p(0)=1$. 
        
        A second index, $j\in\{1,\dots, q(N,i)\}$ will be used to count repetitions of $g_i$ within a given partition of $k\leq N$. Let $p(j)$ be the value of the accumulation factor at the step $j$. $p(j)$ will be related to $p(j-1)$ according to the relation $p(j)=(g_i/j)p(j-1)$. The rationale behind this definition of $p(j)$ is that the accumulation factor will have the information about the contribution to the product $\frac{1}{\delta(\lambda)!}\prod_{a\in\lambda}g_a$ of the factor $g_i$. If $i$ appears $\nu_i$ times in $\lambda$, then the contribution of $g_i$ will have the form $g_i^{\nu_i}/\nu_i!$. This also means that $p(j)$ will take values between $\{g_i, \dots, \frac{1}{q(N,i)!}g_i^{q(N,i)}\}$.
        
        At this point, it is convenient to introduce the vectors of $N+1$ components $\bm{c}_{i,j}$. These vectors satisfy the ``initial'' and ``final'' conditions $\bm{c}_{i,0}\equiv\bm{c}_{i-1}$ and $\bm{c}_{i,q(N,i)}\equiv\bm{c}_{i}$. We can interpret them as auxiliary vectors that lead to the construction of $\bm{c}_i$ from $\bm{c}_{i-1}$. As such, the components $(\bm{c}_{i,j})_l$, for $j\in\{1,\dots, q(N,i)\}$, receive contributions from \textit{some} (but not all) partitions of $l-1$ with summands having values of at most $i$. Namely, $(\bm{c}_{i,j})_l$ receives all contributions from partitions of $l-1$, where summands are at most $i$, and where $i$ appears at most $j$ times.

        We will use a third index, $k$, to determine which component $(\bm{c}_{i,j})_k$ receives a contribution from $p(j)$. For a given $j$, $p(j)$ will represent a contribution of $ij$ to any integer partition. This means that $p(j)$ will only contribute to components of $\bm{c}_{i,j}$ with $k\in \{ij+1, N+1\}$. For $k\not\in \{ij+1, N+1\}$, we set $(\bm{c}_{i,j})_k=(\bm{c}_{i,j-1})_k$. Since $(\bm{c}_{i,j})_k$ collects contributions corresponding \textit{only} to partitions of $k-1$, in general, $p(j)$ alone cannot contribute to $(\bm{c}_{i,j})_k$; there is a remainder of $k-ij-1$ that needs to be taken into account. Seeing that $(\bm{c}_{i-1})_{k-ij}$ has all contributions corresponding to partitions of $k-ij-1$ (where summands are at most $i-1$), the term $(\bm{c}_{i-1})_{k-ij}p(j)$ gives a new contribution to the sum over partitions of $k-1$ with summands taking values of at most $i$, where $i$ appears at most $j$ times. This suggests that each $(\bm{c}_{i,j})_{k}$ can be computed according to the relation $(\bm{c}_{i,j})_{k}=(\bm{c}_{i,j-1})_k + (\bm{c}_{i-1})_{k-ij}p(j)$. We included the term $(\bm{c}_{i,j-1})_k$ because $(\bm{c}_{i,j})_{k}$ must have all contributions from partitions of $k-1$ where $i$ appears at most $j$ times and, in order to do so, it must also include the contributions from all partitions where $i$ appears at most $j-1$ times.
        
        When the indices $j$ and $k$ complete their run, the components $(\bm{c}_i)_l$, for $l\leq i+1$, will be equal to the sum $\sum_{\lambda\vdash (l-1)}\frac{1}{\delta(\lambda)!}\prod_{a\in\lambda}g_a$. After the index $i$ completes its run, the last entry of $\bm{c}_N$ will be equal to the sum in Eq.~\eqref{eq:expression_of_interest_4}.
        
        Before showing that this algorithm works as intended, let us prove a preliminary result of particular importance.
        
        \begin{proposition}
        \label{prop:central_part_algo}
            We have that
            \begin{equation}
                \sum_{\substack{\lambda \vdash n \\ \mathrm{s.t.}\,\lambda_1\leq m+1}}\frac{1}{\delta(\lambda)!}\prod_{a\in \lambda}g_a= \sum_{\substack{\lambda \vdash n \\ \mathrm{s.t.}\,\lambda_1\leq m}}\frac{1}{\delta(\lambda)!}\prod_{a\in \lambda}g_a + \sum_{l=1}^{q(n, m+1)}\frac{1}{l!}g_{m+1}^l\left(\sum_{\substack{\bar{\lambda} \vdash [n-l(m+1)] \\\mathrm{s.t.}\, \bar{\lambda}_1 \leq m}}\frac{1}{\delta(\bar{\lambda})!}\prod_{a\in \bar{\lambda}}g_a\right),
                \label{eq:algorithm_central_equation}
            \end{equation}
            where, let us recall, $\lambda\vdash n$ is a sequence of positive integers $\lambda = (\lambda_1,\lambda_2,\dots)$ satisfying $\lambda_1\geq\lambda_2\geq\dots\geq 1$ and $\lambda_1+\lambda_2\cdots=n$. Defining $\mu_k$ as the multiplicity of $k$ in the partition $\lambda$ (i.e., the number of times $k$ appears in $\lambda$), we have $\delta(\lambda)!=\prod_{k=1}^n\mu_k!$. $q(n, m)$ stands for the quotient between $n$ and $m$.
        \end{proposition}
        
        \begin{proof}
            Let us define the following sets: $\Lambda_{\leq m}^{(n)}$ is the set of all partitions $\lambda \vdash n$ whose summands take values of at most $m$. This means that the largest part of $\lambda$, $\lambda_1$, satisfies $\lambda_1\leq m$. $\Lambda_{m}^{(n)}$ is the set of all partitions $\lambda \vdash n$ whose largest summand is exactly $m$. Finally, $\Lambda_{m, l}^{(n)}$ is the set of all partitions $\lambda\vdash n$, where the largest summand is exactly $m$ and its multiplicity is exactly $l$. In equations, these definitions read
            \begin{align}
                \Lambda_{\leq m}^{(n)}&=\{\lambda \vdash n\,|\, \lambda_1\leq m\},\\
                \Lambda_{m}^{(n)}&=\{\lambda \vdash n\,|\, \lambda_1=m\},\\
                \Lambda_{m,l}^{(n)}&=\{\lambda \vdash n\,|\,\lambda_1=\dots=\lambda_l=m\}.
            \end{align}
            
            Taking into account that for each $\lambda$ we can assign a unique vector $\vec{\mu}(\lambda)=(\mu_1(\lambda),\dots,\mu_n(\lambda))$, where each $\mu_k(\lambda)$ is the multiplicity of $k$ in $\lambda$, and $\mu_1(\lambda)+2\mu_2(\lambda)+\dots+n\mu_n(\lambda)=n$, we can give the following equivalent definitions of the sets $\Lambda_{\leq m}^{(n)}$, $\Lambda_{m}^{(n)}$ and $\Lambda_{m, l}^{(n)}$:
            \begin{align}
                \Lambda_{\leq m}^{(n)}&=\{\lambda \vdash n\,|\, \mu_k(\lambda) = 0\text{ for }k>m\},\\
                \Lambda_{m}^{(n)}&=\{\lambda \vdash n\,|\, \mu_m(\lambda)\neq 0\text{ and }\, \mu_k(\lambda)= 0\text{ for }k>m\},\\
                \Lambda_{m,l}^{(n)}&=\{\lambda \vdash n\,|\,\mu_m(\lambda)= l\text{ and }\, \mu_k(\lambda)= 0\text{ for }k>m\}.
            \end{align}
        
            It can be readily seen that $\Lambda_{\leq m}^{(n)}\cap\Lambda_{m+1}^{(n)}=\emptyset$ and $\Lambda_{\leq m}^{(n)}\cup\Lambda_{m+1}^{(n)}=\Lambda_{\leq m+1}^{(n)}$. These relations allow us to write
            \begin{equation}
                \sum_{\lambda\in \Lambda_{\leq m+1}^{(n)}}\frac{1}{\delta(\lambda)!}\prod_{a\in \lambda}g_a = \sum_{\lambda\in \Lambda_{\leq m}^{(n)}}\frac{1}{\delta(\lambda)!}\prod_{a\in \lambda}g_a + \sum_{\lambda\in \Lambda_{m+1}^{(n)}}\frac{1}{\delta(\lambda)!}\prod_{a\in \lambda}g_a.
                \label{eq:sums_over_sets_1}
            \end{equation}
        
            We can also easily verify that $\Lambda_{m,l}^{(n)}\cap\Lambda_{m,l'}^{(n)}=\emptyset$ for $l\neq l'$. Moreover, taking into account that $m$ can appear in any partition of $n$ at most $q(n,m)\equiv q$ times, we have 
            \begin{equation}
                \Lambda_{m}^{(n)}=\bigcup_{l=1}^{q}\Lambda_{m,l}^{(n)}.
                \label{eq:union_relation}
            \end{equation}
            On account of this relation, we can rewrite Eq.~\eqref{eq:sums_over_sets_1} as
            \begin{equation}
                \sum_{\lambda\in \Lambda_{\leq m+1}^{(n)}}\frac{1}{\delta(\lambda)!}\prod_{a\in \lambda}g_a = \sum_{\lambda\in \Lambda_{\leq m}^{(n)}}\frac{1}{\delta(\lambda)!}\prod_{a\in \lambda}g_a + \sum_{l=1}^{q}\left(\sum_{\lambda\in \Lambda_{m+1,l}^{(n)}}\frac{1}{\delta(\lambda)!}\prod_{a\in \lambda}g_a\right).
                \label{eq:sums_over_sets_2}
            \end{equation}
        
            Now, consider an arbitrary partition $\lambda \in\Lambda_{m+1,l}^{(n)}$. We can write this partition as 
            \[\lambda=(m+1,\dots,m+1,\lambda_{l+1},\lambda_{l+2},\dots) = (m+1,\dots,m+1)\oplus(\lambda_{l+1},\lambda_{l+2},\dots).\]
            Notice that $\bar{\lambda}=(\bar{\lambda}_{1},\bar{\lambda}_{2},\dots)=(\lambda_{l+1},\lambda_{l+2},\dots)$ is a partition of $n-l(m+1)$, but it is not arbitrary. Indeed, we have the restriction $\bar{\lambda}_1=\lambda_{l+1}\leq m$, which is inherited from the fact that $\lambda \in\Lambda_{m+1,l}^{(n)}$. Taking this into account, we can write 
            \begin{equation}
                \Lambda_{m+1,l}^{(n)}=(m+1,\dots,m+1)\oplus \Lambda_{\leq m}^{(n-l(m+1))},
                \label{eq:set_relation}
            \end{equation}
            where it is understood that $(m+1,\dots,m+1)$ has $l$ elements, and $a\oplus B\equiv\{a\oplus b\,|\,b\in B\}$. Eq.~\eqref{eq:set_relation} allows us to write 
            \begin{equation}
                \sum_{\lambda\in \Lambda_{m+1,l}^{(n)}}\frac{1}{\delta(\lambda)!}\prod_{a\in \lambda}g_a=\frac{1}{l!}g_{m+1}^l\left(\sum_{\bar{\lambda} \in \Lambda^{(n-l(m+1))}_{\leq m}}\frac{1}{\delta(\bar{\lambda})!}\prod_{a\in\bar{\lambda}}g_a\right),
                \label{eq:sum_re_structure}
            \end{equation}
            where we noted that $\delta(\lambda)!=l!\delta(\bar{\lambda})!$ for $\lambda = (m+1,\dots,m+1)\oplus\bar{\lambda}$.
        
            Combining Eqs.~\eqref{eq:sums_over_sets_2} and~\eqref{eq:sum_re_structure} we obtain the result
            \begin{equation}
                \sum_{\lambda\in \Lambda_{\leq m+1}^{(n)}}\frac{1}{\delta(\lambda)!}\prod_{a\in \lambda}g_a = \sum_{\lambda\in \Lambda_{\leq m}^{(n)}}\frac{1}{\delta(\lambda)!}\prod_{a\in \lambda}g_a + \sum_{l=1}^{q}\frac{1}{l!}g_{m+1}^l\left(\sum_{\bar{\lambda} \in \Lambda^{(n-l(m+1))}_{\leq m}}\frac{1}{\delta(\bar{\lambda})!}\prod_{a\in\bar{\lambda}}g_a\right).
                \label{eq:sums_over_sets_3}
            \end{equation}
            Eq.~\eqref{eq:algorithm_central_equation} follows from noticing that
            \[\sum_{\lambda\in \Lambda_{\leq m}^{(n)}}\equiv\sum_{\substack{\lambda \vdash n \\ \mathrm{s.t.}\,\lambda_1\leq m}}.\]
        \end{proof}
        
        We will now run the algorithm for $i=1,2$. For $i=1$, we have $g_i=g_1$ and $q(N,i)=N$. This means that $j\in\{1, \dots, N\}$, $p(j)\in\{g_1,\dots, \frac{1}{N!}g_1^N\}$, and $k\in\{j+1,\dots,N+1\}$. For a given $j$, $(\bm{c}_{1,j})_k=(\bm{c}_{1,j-1})_k + (\bm{c}_{0})_{k-j}p(j)$. Since $\bm{c}_0=(1,0,\dots,0)$, we see that the term $(\bm{c}_{0})_{k-j}p(j)$ will be different from zero only for $k=j+1$, therefore $(\bm{c}_{1,j})_{j+1}=(\bm{c}_{1,j-1})_{j+1} + p(j)$ and $(\bm{c}_{1,j})_k=(\bm{c}_{1,j-1})_k$ for $k\neq j+1$. Notice that in each step of $j$ we change only one entry of $\bm{c}_{i,j}$, namely the $(j+1)$-th entry. This means that $(\bm{c}_{1,j-1})_{j+1}$ will keep the value it had at the step $j=1$ (the step $j-1$ only modifies the component $(\bm{c}_{1,j-1})_{j}$), which is equivalent to say that $(\bm{c}_{1,j-1})_{j+1}=0$. Then, we can see that $(\bm{c}_{1,j})_{j+1}= p(j)$ and $(\bm{c}_{1,j})_k=(\bm{c}_{1,j-1})_k$ for $k\neq j+1$. At the end of the run of index $j$ we obtain the result
        \[\bm{c}_{1,N}\equiv\bm{c}_1 = \left(1,g_1,\frac{1}{2}g_1^2,\dots,\frac{1}{N!}g_1^N\right).\]
        We can readily verify that the $l$-th component of $\bm{c}_1$ contains all the contributions from partitions of $l-1$ with summands taking the value of at most 1.
        
        For $i=2$, we have $g_i=g_2$ and $q(N,i)=\lfloor N / 2\rfloor$, which implies that $j\in\{1, \dots, \lfloor N / 2\rfloor\}$, $p(j)\in\{g_2,\dots, \frac{1}{\lfloor N / 2\rfloor!}g_2^{\lfloor N / 2\rfloor}\}$, and $k\in\{2j+1,\dots,N+1\}$. For a given $j$,
        \[(\bm{c}_{2,j})_k=(\bm{c}_{2,j-1})_k + (\bm{c}_{1})_{k-2j}p(j)=(\bm{c}_{2,j-1})_k+\frac{1}{j!(k-2j-1)!}g_{1}^{k-2j-1}g_2^j.\]
        Noticing that
        \[(\bm{c}_{2,0})_k=(\bm{c}_{1})_k=\frac{1}{(k-1)!}g_1^{k-1},\]
        and
        \[(\bm{c}_{2,1})_k=\frac{1}{(k-1)!}g_1^{k-1} + \frac{1}{(k-3)!}g_{1}^{k-3}g_2=\sum_{l=0}^1\frac{1}{l!(k-2l-1)!}g_1^{k-2l-1}g_2^{l},\]
        we can work out the recurrence relation that defines $(\bm{c}_{2,j})_k$ and obtain
        \[(\bm{c}_{2,j})_k=\sum_{l=0}^j\frac{1}{l!(k-2l-1)!}g_1^{k-2l-1}g_2^{l}\]
        for $k\in\{2j+1,\dots,N+1\}$. At the end of the run of index $j$ we obtain the result
        \[\bm{c}_{2,\lfloor N/2\rfloor}\equiv\bm{c}_2 = \left(1,g_1,\frac{1}{2}g_1^2+g_2,\dots,\sum_{l=0}^{\lfloor N/2\rfloor}\frac{1}{l!(N-2l)!}g_1^{N-2l}g_2^{l}\right).\]
        We verify again that the $l$-th component of $\bm{c}_2$ contains all the contributions from partitions of $l-1$ with summands taking the value of at most 2.
        
        Suppose that we run the algorithm up to step $i=m$. This means that
        \[(\bm{c}_m)_l = \sum_{\lambda\vdash (l-1)}\frac{1}{\delta(\lambda)!}\prod_{a\in\lambda}g_a\]
        for $2\leq l \leq m+1$, while
        \[(\bm{c}_m)_l = \sum_{\substack{\lambda\vdash (l-1) \\ \text{s.t. }\lambda_1\leq m}}\frac{1}{\delta(\lambda)!}\prod_{a\in\lambda}g_a\]
        for $m+2\leq l \leq N+1$. 
        
        Let us now compute $\bm{c}_{m+1}$. For an arbitrary value of $j$, and for $k\geq m+2$, we have
        \begin{align*}
            (\bm{c}_{m+1,j})_k&=(\bm{c}_{m+1,j-1})_k + (\bm{c}_{m})_{k-(m+1)j}p(j)\\
            &=(\bm{c}_{m+1,j-1})_k + \frac{1}{j!}g_{m+1}^{j}\left(\sum_{\substack{\bar{\lambda}\vdash (k-(m+1)j-1) \\ \text{s.t. }\bar{\lambda}_1\leq m}}\frac{1}{\delta(\bar{\lambda})!}\prod_{a\in\bar{\lambda}}g_a\right).
        \end{align*}
        Noticing that 
        \[(\bm{c}_{m+1,0})_k=(\bm{c}_{m})_k=\sum_{\substack{\lambda\vdash (k-1) \\ \text{s.t. }\lambda_1\leq m}}\frac{1}{\delta(\lambda)!}\prod_{a\in\lambda}g_a,\]
        and 
        \begin{align*}
            (\bm{c}_{m+1,1})_k &= (\bm{c}_{m})_k + g_{m+1}\left(\sum_{\substack{\bar{\lambda}\vdash (k-m-2) \\ \text{s.t. }\bar{\lambda}_1\leq m}}\frac{1}{\delta(\bar{\lambda})!}\prod_{a\in\bar{\lambda}}g_a\right)\\
            &=\sum_{\substack{\lambda\vdash (k-1) \\ \text{s.t. }\lambda_1\leq m}}\frac{1}{\delta(\lambda)!}\prod_{a\in\lambda}g_a + g_{m+1}\left(\sum_{\substack{\bar{\lambda}\vdash (k-m-2) \\ \text{s.t. }\bar{\lambda}_1\leq m}}\frac{1}{\delta(\bar{\lambda})!}\prod_{a\in\bar{\lambda}}g_a\right),
        \end{align*}
        we can, once more, work out the recurrence relation that defines $(\bm{c}_{m+1,j})_k$ and obtain
        \begin{align*}
            (\bm{c}_{m+1,j})_k =\sum_{\substack{\lambda\vdash (k-1) \\ \text{s.t. }\lambda_1\leq m}}\frac{1}{\delta(\lambda)!}\prod_{a\in\lambda}g_a + \sum_{l=1}^j\frac{1}{l!}g_{m+1}^{l}\left(\sum_{\substack{\bar{\lambda}\vdash (k-(m+1)l-1) \\ \text{s.t. }\bar{\lambda}_1\leq m}}\frac{1}{\delta(\bar{\lambda})!}\prod_{a\in\bar{\lambda}}g_a\right).
        \end{align*}
        
        For each value of $k$, there is a maximum value of $j$, say $j'$, for which $(\bm{c}_{i,j})_k$ no longer receives contributions when $j>j'$. Indeed, we know that $ij + 1\leq k$, which implies that $(\bm{c}_{i,j})_k$ will receive a contribution as long as $j\leq (k-1)/i$. We may therefore say that $j'=q(k-1, i)$. On this account, when $j$ completes its run, we conclude that
        \begin{align*}
            (\bm{c}_{m+1,q(N,m+1)})_k\equiv (\bm{c}_{m+1})_k =\sum_{\substack{\lambda\vdash (k-1) \\ \text{s.t. }\lambda_1\leq m}}\frac{1}{\delta(\lambda)!}\prod_{a\in\lambda}g_a +\!\!\!\!\!\!\sum_{l=1}^{q(k-1, m+1)}\frac{1}{l!}g_{m+1}^{l}\left(\sum_{\substack{\bar{\lambda}\vdash (k-(m+1)l-1) \\ \text{s.t. }\bar{\lambda}_1\leq m}}\frac{1}{\delta(\bar{\lambda})!}\prod_{a\in\bar{\lambda}}g_a\right),
        \end{align*}
        which, according to Proposition~\ref{prop:central_part_algo}, allows us to see that
        \begin{align*}
            (\bm{c}_{m+1})_k =\sum_{\substack{\lambda\vdash (k-1) \\ \text{s.t. }\lambda_1\leq m+1}}\frac{1}{\delta(\lambda)!}\prod_{a\in\lambda}g_a.
        \end{align*}
        
        Recalling that $(\bm{c}_{m+1})_k=(\bm{c}_{m})_k$ for $2\leq k \leq m+1$, we verify that $(\bm{c}_{m+1})_k = \sum_{\lambda\vdash (k-1)}\frac{1}{\delta(\lambda)!}\prod_{a\in\lambda}g_a$. For $k=m+2$, we obtain the same sum. Indeed,
        \begin{align*}
            (\bm{c}_{m+1})_{m+2} =\sum_{\substack{\lambda\vdash (m+1) \\ \text{s.t. }\lambda_1\leq m+1}}\frac{1}{\delta(\lambda)!}\prod_{a\in\lambda}g_a =\sum_{\lambda\vdash (m+1)}\frac{1}{\delta(\lambda)!}\prod_{a\in\lambda}g_a
        \end{align*}
        because any partition $\lambda\vdash(m+1)$ satisfies $\lambda_1\leq m+1$. For $k>m+2$, we verify again that $(\bm{c}_{m+1})_k$ contains all contributions from partitions of $k-1$ whose summands take values of at most $m+1$. We have thus proven that the algorithm works as intended.
        
        We know give a quick estimation of the computational cost, in terms of time, of running this algorithm. For each step $i$, the computation of $g_i$ can be done in $\mathcal{O}(M^3)$ time (recall that in the definition of $g_i$ given in Eq.~\eqref{eq:f_term}, matrix $\bm{A}$ has size $2M\times 2M$ and vector $\gl$ has length $2M$)~\cite{bjorklund2019faster}. The computation of $q(N,i)$ and the initialization of $p$ can both be done in $\mathcal{O}(1)$ time. For each index $j$, the computation of $p(j)$ can also be made in $\mathcal{O}(1)$ time. Finally, for each $k$, the computation of $(\bm{c}_{ij})_k$ is done in $\mathcal{O}(1)$ time. This suggest that the total computational cost of running the algorithm can be written as
        \[t(N)=\sum_{i=1}^N\mathcal{O}(M^3)+\sum_{i=1}^N\sum_{j=1}^{q(N,i)}\sum_{k=ij+1}^{N+1}1.\]
        Developing the sum, we obtain
        \begin{align*}
            t(N)&=\mathcal{O}(NM^3)+\sum_{i=1}^N\sum_{j=1}^{q(N,i)}(N-ij +1) = \mathcal{O}(NM^3)+(N+1)\sum_{i=1}^N\sum_{j=1}^{q(N,i)}1-\sum_{i=1}^Ni\sum_{j=1}^{q(N,i)}j\\
            &=\mathcal{O}(NM^3)+(N+1)\sum_{i=1}^Nq(N,i)-\sum_{i=1}^N\frac{1}{2}\,i\,q(N,i)\left[q(N,i)+1\right].
        \end{align*}
        A result by Dirichlet (see for instance Sec. 3.5 of Ref.~\cite{apostol2013introduction}) states that \[\sum_{i=1}^{N}q(N,i)=N\log N +(2\gamma-1)N + \mathcal{O}(\sqrt{N}),\]
        where $\gamma$ is the Euler-Mascheroni constant. On the other hand, the sum $\sum_{i=1}^N\frac{1}{2}\,i\,q(N,i)\left[q(N,i)+1\right]$, seen as a function of $N$, corresponds to the OEIS sequence \href{https://oeis.org/A143127}{A143127}~\cite{oeis}, for which the relation \[\sum_{i=1}^N\frac{1}{2}\,i\,q(N,i)\left[q(N,i)+1\right] = \frac{1}{2}N^2\log N + \left(\gamma - \frac{1}{4}\right)N^2+\mathcal{O}(N)\] 
        has been established. Consequently, we may see that 
        \begin{align*}
            t(N)=\mathcal{O}(NM^3)+ N(N+1)\log N-\frac{1}{2}N^2\log N + \mathcal{O}(N^2) = \mathcal{O}(NM^3)+\mathcal{O}(N^2\log N).
        \end{align*}
        This is much better than computing all partitions of $N$ in order to compute Eq.~\eqref{eq:expression_of_interest_4}.

    \section{\label{app:lowdin}Löwdin orthogonalization}

        In Sec.~\ref{sec:numerical}, we considered multimode photonic systems described by Gaussian states where the external modes corresponded to the input/output ports of a photonic circuit, while the internal modes where associated to the spectral degree of freedom of the light propagating through the circuit. In particular, we assumed that we had input spectrally impure squeezed states. We modeled this spectral impurity by doubling the number of modes in the system, and by defining a second set of squeezing parameters associated to a second spectral mode in each input port. This allowed us to write the state of the system before propagation through the circuit as $|\bm{\xi}\rangle\otimes|\bm{\xi}'\rangle$, where $|\bm{\xi}\rangle=\bigotimes_{k=1}^M|\xi_k\rangle$, $|\bm{\xi}'\rangle=\bigotimes_{k=1}^M|\xi_k'\rangle$ are $M$-mode squeezed states (with $M$ the number of ports --or external modes-- of the system). The squeezing parameters of the second spectral mode, $\{\xi_k'\}$, were functions of the squeezing parameters of the first spectral mode, $\{\xi_k\}$, and of a ``spectral impurity parameter'', $p$, i.e. $\bm{\xi}'= \bm{\xi}'(\bm{\xi}, p)$. 
        
        The simple model described in Sec.~\ref{sec:numerical} was useful for the purpose of our numerical simulations. However, it is not readily generalizable to systems with more than two internal modes. Moreover, in many scenarios we do not know before hand how many mutually orthogonal states are associated to the internal modes of the system, nor we have a clear way of defining the corresponding squeezing parameters. Instead, we often have access to some information about the states $\{|\phi_k\rangle\}_{k=1}^M$ of the internal degrees of freedom of the input light. For instance, we can think of having input light in the form of pulses that have some relative temporal delay between them (see Fig.~\ref{fig:inter_delay}). Measuring this time delay gives us information about the different internal (temporal) states of the pulses. In this section, we will explain how to use this available information to define the mutually orthogonal internal modes of the system and the associated squeezing parameters. 

        \begin{figure}[!h]
            {
                \includegraphics[width=0.6\columnwidth]{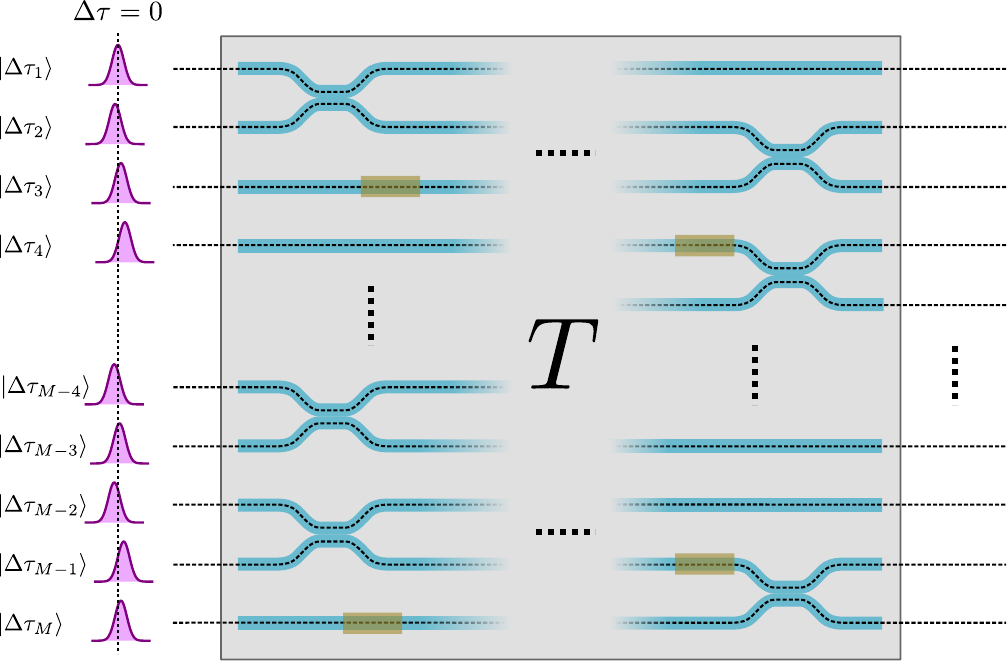}%
            }
            \caption{Illustration of an interferometric setup where the input light has a different state in the temporal degree of freedom for each port of the interferometer. These temporal states are parametrized by the delays, $\Delta\tau_k$, between the different input pulses and some reference pulse. Having information about the state of the internal degrees of freedom of the input light allows to construct an orthonormal set of vectors which we can use to define the mutually orthogonal internal modes of the system, and the corresponding squeezing parameters.} 
            \label{fig:inter_delay}
        \end{figure}

        Suppose for a moment that we have complete knowledge of the internal mode states $\{|\phi_k\rangle\}_{k=1}^M$. We will assume that these vectors are linearly independent, but not necessarily pairwise orthogonal. The idea is to use these states to construct an orthonormal set of vectors that can be associated to the mutually orthogonal internal modes of the system. We will construct this orthonormal set using \textit{Löwdin symmetric orthogonalization}~\cite{piela2006ideas}, which unlike the usual Gram-Schmidt process, generates an orthonormal set that does not depend on the order in which one chooses the different $\{|\phi_k\rangle\}_{k=1}^M$. 
        
        Let $\bm{O}$ be the $M\times M$\textit{overlap matrix}, whose entries are defined as $O_{jk}=\langle \phi_j|\phi_k\rangle$. This matrix is Hermitian, $\bm{O}^\dagger=\bm{O}$, and the linear independence of the $\{|\phi_k\rangle\}_{k=1}^M$ ensures that it is also positive definite, $\bm{O}>0$~\cite{piela2006ideas}. These properties allow us to uniquely define the matrices $\bm{O}^{1/2}$ and $\bm{O}^{-1/2}$. According to Löwdin orthogonalization, the vectors $\{|\psi_k\rangle\}_{k=1}^M$ defined as 
        \begin{equation}
            |\psi_k\rangle = \sum_{l=1}^M\left(\bm{O}^{-1/2}\right)_{k,l}|\phi_l\rangle,
            \label{eq:lowdin_orthoganlization}
        \end{equation}
        form an orthonormal set. We identify the different $\{|\psi_k\rangle\}_{k=1}^M$ with the $M$ mutually orthogonal modes of the internal degrees of freedom, and expand each $|\phi_k\rangle$ as a linear combination of the elements of this orthonormal basis:
        \begin{equation}
            |\phi_k\rangle = \sum_{l=1}^M\left(\bm{O}^{1/2}\right)_{k,l}|\psi_l\rangle.
            \label{eq:lowdin_orthoganlization_inverse}
        \end{equation}

        Suppose now that we prepare $M$ single-mode squeezed states at the modes defined by the external label $k$ and the internal state $|\phi_k\rangle$. We associate to each of these modes a pair of creation and annihilation operators $\hat{a}_{k,\phi_k}^\dagger$, $\hat{a}_{k,\phi_k}$ that satisfy the canonical commutation relations $[\hat{a}_{k,\phi_k},\hat{a}_{l,\phi_l}]=[\hat{a}_{k,\phi_k}^\dagger,\hat{a}_{l,\phi_l}^\dagger]=0$, $[\hat{a}_{k,\phi_k},\hat{a}_{l,\phi_l}^\dagger]=\delta_{k,l}$. The preparation of each squeezed state is mathematically represented by the application of the single-mode squeezing operator $\hat{S}_{k,\phi_k}(\xi_k)=\exp\left[\frac{1}{2}(\xi_k\hat{a}_{k,\phi_k}^\dagger\hat{a}_{k,\phi_k}^\dagger-\xi_k^*\hat{a}_{k,\phi_k}\hat{a}_{k,\phi_k})\right]$ over the corresponding vacuum state. Here, $\xi_k$ is the generally complex squeezing parameter associated to the squeezed state at mode $(k,\phi_k)$. We would like to rewrite this squeezing operator in terms of the mutually orthogonal internal modes obtained in Eq.~\eqref{eq:lowdin_orthoganlization}. Furthermore, we would like to investigate if it can be factorized as a tensor product of squeezing operators corresponding to a \textit{single internal mode}. In this way we would obtain a similar model as the one we used in Sec.~\ref{sec:numerical}. 

        Following the expansion in Eq.~\eqref{eq:lowdin_orthoganlization_inverse}, we can define new creation and annihilation operators $\hat{a}_{k,\psi_l}^\dagger,\hat{a}_{k,\psi_l}$ (where label $k$ indicates the external mode, and label $l$ the internal mode), with $k,l\in\{1,\dots,M\}$, that allow us to write 
        \begin{equation}
            \hat{a}_{k,\phi_k}^{\dagger}=\sum_{l=1}^M\left(\bm{O}^{1/2}\right)_{k,l}\hat{a}_{k,\psi_l}^\dagger\;\;,\;\;\:\hat{a}_{k,\phi_k}=\sum_{l=1}^M\left(\bm{O}^{1/2}\right)_{l,k}\hat{a}_{k,\psi_l}.
            \label{eq:lowdin_creation_annihilation}
        \end{equation}
        The operators $\hat{a}_{k,\psi_l}^\dagger,\hat{a}_{k,\psi_l}$ follow similar commutation relations as operators $\hat{a}_{k,\phi_k}^\dagger$, $\hat{a}_{k,\phi_k}$ do. Namely, $[\hat{a}_{k,\psi_l},\hat{a}_{m,\psi_n}]=[\hat{a}_{k,\psi_l}^\dagger,\hat{a}_{m,\psi_n}^\dagger]=0$, $[\hat{a}_{k,\psi_l},\hat{a}_{m,\psi_n}^\dagger]=\delta_{k,m}\delta_{l,n}$. Replacing Eq.~\eqref{eq:lowdin_creation_annihilation} into the definition of $\hat{S}_{k,\phi_k}(\xi_k)$, we obtain
        \begin{align}
            \hat{S}_{k,\phi_k}(\xi_k)&=\exp\left[\frac{1}{2}\sum_{l,m=1}^M\xi_k\left(\bm{O}^{1/2}\right)_{k,l}\left(\bm{O}^{1/2}\right)_{k,m}\hat{a}_{k,\psi_l}^\dagger\hat{a}_{k,\psi_m}^\dagger-\,\mathrm{H.c.}\right]\nonumber\\
            &=\exp\left[\frac{1}{2}\sum_{l,m=1}^MJ_{l,m}^{(k)}\hat{a}_{k,\psi_l}^\dagger\hat{a}_{k,\psi_m}^\dagger-\,\mathrm{H.c.}\right],
            \label{eq:squeezing_internal_modes}
        \end{align}
        where $\mathrm{H.c.}$ stands for the Hermitian conjugate of the expression between brackets. Notice that matrix $\bm{J}^{(k)}$, whose components are $J_{l,m}^{(k)}$, is symmetric and generally non-diagonal. This later property of $\bm{J}^{(k)}$ keeps us from readily factorizing operator $\hat{S}_{k,\phi_k}(\xi_k)$ as a tensor product of squeezing operators corresponding to individual internal modes. Nevertheless, we can avoid this issue by properly redefining the operators $\hat{a}_{k,\psi_l}^\dagger,\hat{a}_{k,\psi_l}$.

        Since $\bm{J}^{(k)}$ is symmetric, we can use its Takagi-Autonne decomposition~\cite{quesada2022beyond, houde2024matrix} to write 
        \begin{equation}
            \bm{J}^{(k)}=\bm{F}^{(k)}\left[\bigoplus_{n=1}^M\tilde{\xi}_{k,n}\right]\left(\bm{F}^{(k)}\right)^{\mathrm{T}},    
        \end{equation}
        where $\tilde{\xi}_{k,n}\geq0$ for all $n\in\{1,\dots,M\}$, and $\bm{F}^{(k)}$ is unitary. This decomposition, in turn, allows us to see that
        \begin{equation}
            \sum_{l,m=1}^M J_{l,m}^{(k)}\hat{a}_{k,\psi_l}^\dagger\hat{a}_{k,\psi_m}^\dagger=\sum_{n=1}^M\tilde{\xi}_{k,n}\left(\sum_{l=1}^MF_{l,n}^{(k)}\hat{a}_{k,\psi_l}^\dagger\right)\left(\sum_{m=1}^MF_{m,n}^{(k)}\hat{a}_{k,\psi_m}^\dagger\right)=\sum_{n=1}^M\tilde{\xi}_{k,n}\hat{b}_{k,n}^\dagger\hat{b}_{k,n}^\dagger,
            \label{eq:mode_redefinition}
        \end{equation}
        where $F_{m,n}^{(k)}$ are the entries of matrix $\bm{F}^{(k)}$, and we have defined the creation and annihilation operators $\hat{b}_{k,n}^\dagger, \hat{b}_{k,n}$ according to the relations $\hat{b}_{k,n}^\dagger=\sum_{l=1}^MF_{l,n}^{(k)}\hat{a}_{k,\psi_l}^\dagger$, $\hat{b}_{k,n}=\sum_{l=1}^M[F_{l,n}^{(k)}]^*\hat{a}_{k,\psi_l}$. Since, $\bm{F}^{(k)}$ is unitary, the operators $\hat{b}_{k,n}^\dagger, \hat{b}_{k,n}$ satisfy canonical commutation relations analogous to those fulfilled by $\hat{a}_{k,\psi_l}^\dagger,\hat{a}_{k,\psi_l}$: $[\hat{b}_{k,l},\hat{b}_{m,n}]=[\hat{b}_{k,l}^\dagger,\hat{b}_{m,n}^\dagger]=0$, $[\hat{b}_{k,l},\hat{b}_{m,n}^\dagger]=\delta_{k,m}\delta_{l,n}$.

        With relation~\eqref{eq:mode_redefinition} in place, we can recast operator $\hat{S}_{k,\phi_k}(\xi_k)$ into the form
        \begin{align}
            \hat{S}_{k,\phi_k}(\xi_k)=\exp\left[\frac{1}{2}\sum_{n=1}^M\tilde{\xi}_{k,n}\hat{b}_{k,n}^\dagger\hat{b}_{k,n}^\dagger-\,\mathrm{H.c.}\right]=\bigotimes_{n=1}^M\hat{S}_{k,n}(\tilde{\xi}_{k,n}),
            \label{eq:squeezing_internal_modes_recast}
        \end{align}
        with $\hat{S}_{k,n}(\tilde{\xi}_{k,n})=\exp\left[\frac{1}{2}\tilde{\xi}_{k,n}(\hat{b}_{k,n}^\dagger\hat{b}_{k,n}^\dagger-\hat{b}_{k,n}\hat{b}_{k,n})\right]$. We have succeeded in defining an orthonormal basis for the internal degrees of freedom of the system, in which the initial squeezed state can be written as a tensor product of single mode squeezed states associated to individual internal modes, just as in Sec.~\ref{sec:numerical}. Notice, that although we assumed complete knowledge of the internal mode states $\{|\phi_k\rangle\}_{k=1}^M$, in order to define the squeezing parameters $\{\tilde{\xi}_{k,n}\}_{k,n=1}^M$ we need only know the original set of squeezing parameters, $\{\xi_{k}\}_{k=1}^M$, and the overlap matrix $\bm{O}$. 

    \section{\label{app:coarse_grain_fock}Coarse-graining for input Fock states}

        In this section, we will apply the formalism we developed in the main text to the computation of coarse-grained probability distributions for photon number (or Fock) states propagating through lossy optical circuits. In Sec.~\ref{sec:numerical}, we showed that these systems can be studied by assuming that the input Fock states are obtained through the partial measurement of two-mode squeezed states. In what follows, we will consider a more general case, in which the partial measurement of a Gaussian state is not necessary.

        Suppose that we have a photonic circuit with $M$ input/output ports. The action of this circuit is mathematically described by a $M\times M$ transmission matrix $\bm{T}$. Given the presence of losses and other sources of noise, $\bm{T}$ is generally subunitary, i.e., $\bm{T}^\dagger \bm{T}\leq\mathbb{I}$. Let $|\bm{p}\rangle\langle \bm{p}|$, with $\bm{p}=(p_1, \dots, p_M)$, represent the multimode photon-number state of the light entering the circuit. We call $N=\sum_{k=1}^Mp_k$ the total number of input photons. The state of the system after propagation through the circuit can be written as $\hat{\varrho}=\Phi\left[|\bm{p}\rangle\langle \bm{p}|\right]$, where $\Phi[\cdot]$ is the Gaussian loss channel defined by $\bm{T}$. It can be shown that the density matrix elements of $\hat{\varrho}$ in the Fock basis take the form~\cite{yao2024riemannian}
        \begin{equation}
            \langle\bm{m}|\hat{\varrho}|\bm{n}\rangle = \left(\prod_{k=1}^M\frac{1}{p_k!\sqrt{n_k!m_k!}}\right)\mathrm{haf}\left[\left(\bm{A}_{\Phi}\right)_{(\bm{p}\oplus\bm{m})\oplus(\bm{p}\oplus\bm{n})}\right],
            \label{eq:density_matrix_fock_circuit}
        \end{equation}
        where $\mathrm{haf}(\bm{A})\equiv\mathrm{lhaf}(\bm{A},\bm{0})$ is the \textit{Hafnian} of matrix $\bm{A}$. The $4M\times 4M$ symmetric matrix $\bm{A}_\Phi$ is defined as
        \begin{equation}
            \bm{A}_\Phi =
            \begin{pmatrix}
                0 & \bm{T}^{\dagger} & \mathbb{I}-\bm{T}^{\dagger}\bm{T} & 0\\
                \bm{T}^{*} & 0 & 0 & 0\\
                \mathbb{I}-\bm{T}^{\mathrm{T}}\bm{T}^* & 0 & 0 & \bm{T}^{\mathrm{T}}\\
                0 & 0 & \bm{T} & 0
            \end{pmatrix}.
            \label{eq:loss_channel_matrix}
        \end{equation}

        Suppose now that we define a set partition $\Lambda=\{\Lambda_1,\dots,\Lambda_L\}$ of $\{1,\dots, M\}$, and a coarse-grained detection pattern $\bm{b}=(b_1,\dots, b_L)$ such that $|\bm{b}|=\sum_{k=1}^Lb_k \leq N$. The probability of the detection event defined by $\bm{b}$ and $\Lambda$ is computed as
        \begin{equation}
            \Pr(\bm{b}|\bm{p},\bm{T}) = \sum_{\bm{n}\in\mathcal{K}_{\Lambda}^{\bm{b}}} \langle\bm{n}|\hat{\varrho}|\bm{n}\rangle=\sum_{\bm{n}\in\mathcal{K}_{\Lambda}^{\bm{b}}} \left(\prod_{k=1}^M\frac{1}{p_k!n_k!}\right)\mathrm{haf}\left[\left(\bm{A}_{\Phi}\right)_{(\bm{p}\oplus\bm{n})\oplus(\bm{p}\oplus\bm{n})}\right]
            \label{eq:coarse_grained_Fock}
        \end{equation}
        If we further define the set partition $\Lambda'$ of $\{1,\dots,2M\}$ as $\Lambda'=\{\{1\},\dots,\{M\}\}\cup \Lambda$, where $\Lambda$ is now defined over the set $\{M+1,\dots,2M\}$ (i.e. all the elements in each block of $\Lambda$ are shifted by $M$), and the detection pattern $\bm{b}'=\bm{p}\oplus\bm{b}$, we can recast $\Pr(\bm{b}|\bm{p},\bm{T})$ into the form
        \begin{equation}
            \Pr(\bm{b}|\bm{p},\bm{T}) =\sum_{\bm{n}'\in\mathcal{K}_{\Lambda'}^{\bm{b}'}} \left(\prod_{k=1}^M\frac{1}{n_k'!}\right)\mathrm{haf}\left[\left(\bm{A}_{\Phi}\right)_{\bm{n}'\oplus\bm{n}'}\right],
            \label{eq:coarse_grained_Fock_2}
        \end{equation}
        where we wrote $\bm{n}'=\bm{p}\oplus\bm{n}$. We can immediately notice that this expression is related to a blocked loop Hafnian. Consequently, we may write
        \begin{equation}
            \Pr(\bm{b}|\bm{p},\bm{T}) = \left(
            \prod_{j=1}^M\frac{1}{p_j!}\prod_{k=1}^L\frac{1}{b_k!}\right)\lhaf_{\Lambda'}\left(\bm{A}_{\Phi},\bm{0},\bm{b}'\right) = \left(
            \prod_{j=1}^M\frac{1}{p_j!}\prod_{k=1}^L\frac{1}{b_k!}\right) \left(\prod_{l=1}^{M+L} D_{w_l}^{(b_l')} \right) f_{N'}(\bm{A}_{\Phi}, \bm{0}, \bm{w}),
            \label{eq:coarse_grained_Fock_3}
        \end{equation}
        with $N'=N+|\bm{b}|$.

        To make clear the advantage of using Eq.~\eqref{eq:coarse_grained_Fock_3} to compute coarse-grained probabilities, notice that when $\bm{m}=\bm{n}$, Eq.~\eqref{eq:density_matrix_fock_circuit} is equivalent to~\cite{yao2024riemannian}
        \begin{equation}
            \langle\bm{n}|\hat{\varrho}|\bm{n}\rangle = \left(\prod_{k=1}^M\frac{1}{p_k!n_k!}\right)\mathrm{perm}\left[
            \begin{pmatrix}
                \mathbb{I}-\bm{T}^\dagger\bm{T} & \bm{T}^\dagger\\
                \bm{T} & 0
            \end{pmatrix}_{\bm{n}\oplus\bm{p}}
            \right],
            \label{eq:probs_fock_circuit}
        \end{equation}
        where $\mathrm{perm}(\bm{A})$ is the permanent of matrix $\bm{A}$. Thus, $\Pr(\bm{b}|\bm{p},\bm{T})$ can also be expressed as 
        \begin{equation}
            \Pr(\bm{b}|\bm{p},\bm{T}) =\sum_{\bm{n}\in\mathcal{K}_{\Lambda}^{\bm{b}}} \left(\prod_{k=1}^M\frac{1}{p_k!n_k!}\right)\mathrm{perm}\left[
            \begin{pmatrix}
                \mathbb{I}-\bm{T}^\dagger\bm{T} & \bm{T}^\dagger\\
                \bm{T} & 0
            \end{pmatrix}_{\bm{n}\oplus\bm{p}}
            \right].
            \label{eq:coarse_grained_Fock_alt}
        \end{equation}
        Taking into account that calculating each of the permanents in this sum can be done in $\mathcal{O}(N'2^{N'})$ time (since all the matrices inside the permanent have size $N'\times N'$), we can see that obtaining coarse-grained probabilities using Eq.~\eqref{eq:coarse_grained_Fock_alt} has time complexity 
        \[\mathcal{O}\left(N'2^{N'}\prod_{k=1}^L\binom{|\Lambda_k| + b_k -1}{b_k}\right),\]
        where $|\Lambda_k|$ is the number of elements in block $\Lambda_k$. In contrast, Eq.~\eqref{eq:coarse_grained_Fock_3} has time complexity
        \[\mathcal{O}\left(\left(N'M^3 + N'^2\log N'\right)\prod_{j=1}^M(p_j+1)\prod_{k=1}^L(b_k+1)\right),\]
        which has better scaling for an increasing number of input photons and modes.

        In addition to computing coarse-grained probability distributions, we may also study the generation of photonic states conditioned to a coarse-grained measurement of some of the output modes of the optical circuit (just as we did in Sec.~\ref{sec:off_diag_internal}). Suppose that the modes in $\mathcal{H}=\{1,\dots,H\}$ are measured, while the modes in $\mathcal{G}=\{H+1,\dots,M\}$ support the heralded state. We define a partition $\tilde{\Lambda}=\{\tilde{\Lambda}_1,\dots,\tilde{\Lambda}_L\}$ of $\mathcal{H}$, and we condition the heralding of the desired state to the measurement of the coarse-grained detection pattern $\tilde{\bm{b}}=(\tilde{b}_1,\dots,\tilde{b}_L)$. As shown in Sec.~\ref{sec:off_diag_internal}, the density matrix elements of the unnormalized heralded state on the modes in $\mathcal{G}$, $\hat{\varrho}_{\mathcal{G}(\tilde{\Lambda})}$, read
        \begin{align}
        	\langle\bm{v}&|\hat{\varrho}_{\mathcal{G}(\tilde{\Lambda})}|\bm{u} \rangle = \sum_{\tilde{\bm{n}}\in\mathcal{K}_{\tilde{\Lambda}}^{\tilde{\bm{b}}}}\langle \tilde{\bm{n}},{\bm{v}}|\hat{\varrho}|\tilde{\bm{n}},{\bm{u}}\rangle =\left(\prod_{k=1}^M\frac{1}{p_k!}\prod_{j=1}^{M-H}\frac{1}{\sqrt{u_k!v_k!}}\right)\sum_{\tilde{\bm{n}}\in\mathcal{K}_{\tilde{\Lambda}}^{\tilde{\bm{b}}}}\left(\prod_{k=1}^H\frac{1}{\tilde{n}_k!}\right)\mathrm{haf}\left[\left(\bm{A}_{\Phi}\right)_{(\bm{p}\oplus\tilde{\bm{n}}\oplus\tilde{\bm{v}})\oplus(\bm{p}\oplus\tilde{\bm{n}}\oplus\tilde{\bm{u}})}\right],
        	\label{eq:lhaf_heralded_offd_internal_fock}
        \end{align}
        where $\tilde{\bm{n}}=(\tilde{n}_1,\dots,\tilde{n}_H)$, $\bm{v}=(v_1,\dots,v_{M-H})$ and $\bm{u}=(u_1,\dots,u_{M-H})$.

        As explained in Sec.~\ref{sec:off_diag_no_internal}, there exist a matrix $\bm{A}_{\Phi}'$, and a detection pattern $\bm{r}$ for which 
        \[\mathrm{haf}\left(\bm{A}_{\Phi}\right)_{(\bm{p}\oplus\tilde{\bm{n}}\oplus\tilde{\bm{v}})\oplus(\bm{p}\oplus\tilde{\bm{n}}\oplus\tilde{\bm{u}})}=\mathrm{haf}\left(\bm{A}_{\Phi}'\right)_{(\bm{p}\oplus\tilde{\bm{n}}\oplus\tilde{\bm{r}})\oplus(\bm{p}\oplus\tilde{\bm{n}}\oplus\tilde{\bm{r}})}.\]
        Moreover, we can define the set of modes $\{1,\dots,M+H+R\}$, with $R$ the length of vector $\bm{r}$; the set partition $\tilde{\Lambda}'=\{\{1\},\dots,\{M\}\}\cup\tilde{\Lambda}\cup\{\{M+H+1\},\dots,\{M+H+R\}\}$, where $\tilde{\Lambda}$ is now defined over $\{M+1,\dots,M+H\}$; and the coarse-grained detection pattern $\tilde{\bm{b}}'=\bm{p}\oplus\tilde{\bm{b}}\oplus\bm{r}$. Then, we can recast Eq.~\eqref{eq:lhaf_heralded_offd_internal_fock} as 
        \begin{align}
        	\langle\bm{v}&|\hat{\varrho}_{\mathcal{G}(\tilde{\Lambda})}|\bm{u} \rangle  =\left(\prod_{k=1}^Rr_k!\prod_{j=1}^{M-H}\frac{1}{\sqrt{u_k!v_k!}}\right)\sum_{\tilde{\bm{n}}'\in\mathcal{K}_{\tilde{\Lambda}'}^{\tilde{\bm{b}}'}}\left(\prod_{k=1}^{M+H+R}\frac{1}{\tilde{n}'_k!}\right)\mathrm{haf}\left[\left(\bm{A}'_{\Phi}\right)_{\tilde{\bm{n}}'\oplus\tilde{\bm{n}}'}\right],
        	\label{eq:lhaf_heralded_offd_internal_fock_2}
        \end{align}
        where $\tilde{\bm{n}}'=\bm{p}\oplus\tilde{\bm{n}}\oplus\bm{r}$. As can be seen, this expression can be readily related to a blocked loop Hafnian, and so we may write
        \begin{align}
        	\langle\bm{v}&|\hat{\varrho}_{\mathcal{G}(\tilde{\Lambda})}|\bm{u} \rangle  =\left(\prod_{j=1}^M\frac{1}{p_j!}\prod_{k=1}^L\frac{1}{\tilde{b}_k!}\prod_{l=1}^{M-H}\frac{1}{\sqrt{u_l!v_l!}}\right)\left(\prod_{j=1}^{M} D_{w'_j}^{(p_j)} \prod_{k=M+1}^{M+L} D_{w'_k}^{(\tilde{b}_{k-M})}\prod_{l=M+L+1}^{M+L+R} D_{w'_l}^{(r_{l-M-L})} \right)f_{\tilde{N}+|\bm{r}|}(\bm{A}_{\Phi}', \bm{0}, \bm{w}'),
        	\label{eq:lhaf_heralded_offd_internal_fock_3}
        \end{align}
        with $\tilde{N}=N+|\tilde{\bm{b}}|$.

    \section{\label{fully_dist}Photon-number probabilities for fully distinguishable internal modes}

        Let us consider a similar situation as the one we studied in Sec.\ref{sec:photon_dist_internal}. Suppose that we have a Gaussian state with $M$ external modes and $K$ internal modes per external mode. For simplicity, let us assume that this state is non-displaced, so $\bar{\gl}=\bm{0}$, and it will be completely described by its $2MK\times 2MK$ adjacency matrix $\bar{\bm{A}}$. Furthermore, suppose that there is no entanglement between the different internal modes of the system. Then, the matrix $\bm{X}\bar{\bm{A}}$ can be expressed as $\bm{X}\bar{\bm{A}}=\bigoplus_{l=1}^K\bm{B}_l$, where the $\{\bm{B}_l\}$ are $2M\times 2M$, Hermitian matrices. One important example of this situation, on which we will focus from now on, is found when interfering spectrally distinguishable (i.e., in mutually orthogonal spectral states) single-mode Gaussian states into a linear optical circuit. In this case, we have that all the $\{\bm{B}_l\}$ are rank 2. This can be verified, for instance, by noticing that each $\bm{B}_l$ can be interpreted as describing a single-mode Gaussian state propagating through a $M$-mode optical circuit. In this section, we will show that these properties of matrix $\bm{X}\bar{\bm{A}}$ allow us to construct an algorithm for computing photon-number probabilities that has a polynomial speed-up with respect to the results shown in Sec.\ref{sec:photon_dist_internal}.

        Let us first recall that the probability of obtaining the photon-number detection pattern $\bm{n}=(n_1,\dots, n_M)$ when measuring the Gaussian state represented by $\bar{\bm{A}}$ (Eq.~\eqref{eq:prob_internal_modes}) reads
        \begin{align}
            \Pr(\bm{n}|\bar{\bm{A}},\bm{0})=\frac{\Pr\left(\bm{0}| \bar{\bm{A}}, \bm{0} \right)}{\prod_{k=1}^Mn_k!}\left(\prod_{k=1}^M D_{w_k}^{(n_k)} \right) f_N(\bar{\bm{A}}, \bm{0}, \bar{\bm{w}}).
            \label{eq:prob_internal_modes_app}
        \end{align}
        
        Let us also recall the definition of the finite-difference operators, $D_{w_k}^{(n_k)}$ (Eq.~\eqref{eq:non_multilinear_finite_diff}),
        \begin{align}
    		D_{w_k}^{(n_k)} P(w_1,\dots,w_M) = \frac{1}{(u_k - v_k)^{n_k}}\sum_{m=0}^{n_k} \binom{n_k}{m} (-1)^{n_k-m} P(w_1,\dots, w_k = v_k + m(u_k - v_k), \dots,w_M),
            \label{eq:non_multilinear_finite_diff_app}
		\end{align}
        where $P(w_1,\dots,w_M)$ is a polynomial of degree less than or equal to $n_1+\cdots+n_M=N$. Let $u_k=1$ and $v_k=0$ for all $k\in\{1,\dots, M\}$, then we may write
        \begin{align}
    		\left(\prod_{k=1}^M D_{w_k}^{(n_k)} \right) P(w_1,\dots,w_M) = \sum_{m_1=0}^{n_1}\cdots\sum_{m_k=0}^{n_M} \binom{n_1}{m_1} \cdots \binom{n_M}{m_M}(-1)^{N-\sum_{k=1}^M m_k} P(m_1,\dots, m_M).
            \label{eq:finite_diff_inclusion_exclusion}
		\end{align}
        Define the \textit{multiset} $\mathcal{M}=\{1,\dots,1,\cdots,M,\dots, M\}$, where each value of $k\in\{1,\dots, M\}$ appears a total of $n_k$ times (so $|\mathcal{M}|=N$). Denote $\Omega[\mathcal{M}]$ the power set (i.e. the set of all \textit{submultisets}) of $\mathcal{M}$. Then, we can see that 
        \begin{equation}
            \sum_{m_1=0}^{n_1}\cdots\sum_{m_k=0}^{n_M} \binom{n_1}{m_1} \cdots \binom{n_M}{m_M}(-1)^{N-\sum_{k=1}^M m_k} P(m_1,\dots, m_M)\equiv \sum_{\mathcal{V}\in\Omega[\mathcal{M}]}(-1)^{N-|\mathcal{V}|} \binom{\bm{n}}{\bm{m}_\mathcal{V}} P(\bm{m}_{\mathcal{V}}),
            \label{eq:multiset_sum}
        \end{equation}
        where $\bm{m}_\mathcal{V}$ is a vector whose $k$th component is the number of times $k$ appears in $\mathcal{V}$ (i.e., the multiplicity of $k$ in $\mathcal{V}$), and we denote $\binom{\bm{n}}{\bm{m}_\mathcal{V}}=\binom{n_1}{m_1} \cdots \binom{n_M}{m_M}$. With this notation at hand, we can recast Eq.~\eqref{eq:prob_internal_modes_app} as
        \begin{align}
            \Pr(\bm{n}|\bar{\bm{A}},\bm{0})=\frac{\Pr\left(\bm{0}| \bar{\bm{A}}, \bm{0} \right)}{\prod_{k=1}^Mn_k!}\sum_{\mathcal{V}\in\Omega[\mathcal{M}]}(-1)^{N-|\mathcal{V}|} \binom{\bm{n}}{\bm{m}_\mathcal{V}}f_N(\bar{\bm{A}}_\mathcal{V},\bm{0}),
            \label{eq:prob_internal_modes_app_2}
        \end{align}
        where $\bm{X}\bar{\bm{A}}_\mathcal{V} = \bigoplus_{l=1}^K \bm{D}(\bm{m}_{\mathcal{V}})\bm{B}_l$, and $\bm{D}(\bm{m}_{\mathcal{V}})=\mathrm{diag}(\bm{m}_{\mathcal{V}})\oplus\mathrm{diag}(\bm{m}_{\mathcal{V}})$.
        
        The calculation of $f_N(\bar{\bm{A}}_\mathcal{V}, \bm{0})$ requires us to compute power traces of the form (see Eq.~\eqref{eq:f_term})
        \begin{equation}
            \tr\left(\left[\bm{X}\bar{\bm{A}}_\mathcal{V}\right]^j\right) = \sum_{l=1}^K\tr\left(\left[\bm{D}(\bm{m}_{\mathcal{V}})\bm{B}_l\right]^j\right) = \sum_{l=1}^K\tr\left(\left[\sqrt{\bm{D}(\bm{m}_{\mathcal{V}})}\bm{B}_l\sqrt{\bm{D}(\bm{m}_{\mathcal{V}})}\right]^j\right).
            \label{eq:power_trace_dist}
        \end{equation}
        These quantities can readily be obtained if we have access to the eigenvalues of $\bm{X}\bar{\bm{A}}_\mathcal{V}$, or equivalently, the eigenvalues of $\sqrt{\bm{D}(\bm{m}_{\mathcal{V}})}\bm{B}_l\sqrt{\bm{D}(\bm{m}_{\mathcal{V}})}$ for all $l$.
        
        Since $\bm{D}(\bm{m}_\mathcal{V})$ and all the $\{\bm{B}_l\}$ are Hermitian, $\sqrt{\bm{D}(\bm{m}_{\mathcal{V}})}\bm{B}_l\sqrt{\bm{D}(\bm{m}_{\mathcal{V}})}$ will also be Hermitian for all $l$. Moreover, since $2\leq \mathrm{rank}[\bm{D}(\bm{m}_{\mathcal{V}})]\leq 2M$ (ignoring the case in which $\bm{n}=\bm{0}$),  $\mathrm{rank}[\sqrt{\bm{D}(\bm{m}_{\mathcal{V}})}\bm{B}_l\sqrt{\bm{D}(\bm{m}_{\mathcal{V}})}] \leq 2$ for $l\in\{1,\dots, K\}$. From these properties, we can conclude that each $\sqrt{\bm{D}(\bm{m}_{\mathcal{V}})}\bm{B}_l\sqrt{\bm{D}(\bm{m}_{\mathcal{V}})}$ has at most two singular values $s_1(l,\mathcal{V})$, $s_2(l,\mathcal{V})$, and at most two real eigenvalues satisfying $s_1(l,\mathcal{V})=|\lambda_1(l,\mathcal{V})|$, $s_2(l,\mathcal{V})=|\lambda_2(l,\mathcal{V})|$. Notice that $\tr[\sqrt{\bm{D}(\bm{m}_{\mathcal{V}})}\bm{B}_l\sqrt{\bm{D}(\bm{m}_{\mathcal{V}})}]=\lambda_1(l,\mathcal{V})+\lambda_2(l,\mathcal{V})$. Moreover, we have that $\|\sqrt{\bm{D}(\bm{m}_{\mathcal{V}})}\bm{B}_l\sqrt{\bm{D}(\bm{m}_{\mathcal{V}})}\|_\mathrm{F}=\sqrt{s_1^2(l,\mathcal{V})+s_2^2(l,\mathcal{V})}= \sqrt{\lambda_1^2(l,\mathcal{V})+\lambda_2^2(l,\mathcal{V})}$, where $\|\bm{A}\|_\mathrm{F}$ is the Frobenius norm of $\bm{A}$. Let us call $T(l,\mathcal{V})=\tr[\sqrt{\bm{D}(\bm{m}_{\mathcal{V}})}\bm{B}_l\sqrt{\bm{D}(\bm{m}_{\mathcal{V}})}]$ and $F(l,\mathcal{V})=\|\sqrt{\bm{D}(\bm{m}_{\mathcal{V}})}\bm{B}_l\sqrt{\bm{D}(\bm{m}_{\mathcal{V}})}\|_\mathrm{F}$. The previous relations allow us to see that
        \begin{equation}
            \lambda_{1,2}(l,\mathcal{V})=\frac{1}{2}T(l,\mathcal{V})\pm \frac{1}{2}\sqrt{2F^2(l,\mathcal{V})-T^2(L,\mathcal{V})}.
            \label{eq:trace_frob_norm_eigen}
        \end{equation}
        
        Taking into account that $\tr(\bm{A})=\sum_{i}A_{i,i}$ and $\|\bm{A}\|_\mathrm{F}^2=\sum_{i,j}|A_{i,j}|^2$, Eq.~\eqref{eq:trace_frob_norm_eigen} allows us to obtain the eigenvalues of $\bm{X}\bar{\bm{A}}_\mathcal{V}$ in $\mathcal{O}(KM^2)$ time. Indeed, we need only sum over two indices, and obtain $K$ pairs of eigenvalues (one for each spectral mode). This reduces the time computation of $f_N(\bar{\bm{A}}_\mathcal{V}, \bm{0})$ from $\mathcal{O}(NK^3M^3 + N^2\log N)$ to $\mathcal{O}(NKM^2 + N^2\log N)$. Thus, the time of computation of $\Pr(\bm{n}|\bar{\bm{A}},\bm{0})$ reduces to $\mathcal{O}\left((NKM^2 + N^2\log N)\prod_{k=1}^M(n_k+1)\right)$.

        A further speed-up can be obtained if we assume that we generate all of the submultisets we need to sum over in Eq.~\eqref{eq:prob_internal_modes_app_2} in a generalized Gray code order~\cite{nijenhuis2014combinatorial, hage2003enumerating, guan1998generalized}, i.e., an order such that two subsequent submultisets $\mathcal{V}, \mathcal{V'}\subseteq \mathcal{M}$ differ by only one element, say $k^\star$. Notice that every $\mathcal{V}\subseteq\mathcal{M}$ is uniquely defined by its vector of multiplicities $\bm{m}_\mathcal{V}$. Thus, if $\mathcal{V}$ and $\mathcal{V}'$ differ only by the element $k^\star$, $\bm{m}_\mathcal{V}$ and $\bm{m}_\mathcal{V'}$ will differ only in the $k^\star$th component, which satisfies $(\bm{m}_\mathcal{V})_{k^\star}=(\bm{m}_\mathcal{V'})_{k^\star}\pm 1$. On this account, we can see that
        \begin{align}
            T(l,\mathcal{V})&=\sum_{i=1}^M(\bm{m}_\mathcal{V})_{i}\left[(\bm{B}_l)_{i,i} + (\bm{B}_l)_{i+M,i+M}\right],\nonumber\\
            &=T(l,\mathcal{V}') \pm \left[(\bm{B}_l)_{k^\star,k^\star} + (\bm{B}_l)_{k^\star+M,k^\star+M}\right],
        \end{align}
        and
        \begin{align}
            F^2(l,\mathcal{V})&=\sum_{i,j =1}^M(\bm{m}_\mathcal{V})_{i}(\bm{m}_\mathcal{V})_{j}\left[|(\bm{B}_l)_{i,j}|^2 + |(\bm{B}_l)_{i+M,j+M}|^2 + |(\bm{B}_l)_{i+M,j}|^2 + |(\bm{B}_l)_{i,j+M}|^2\right]\nonumber\\
            &=F^2(l,\mathcal{V}') + \left[|(\bm{B}_l)_{k^\star,k^\star}|^2 + |(\bm{B}_l)_{k^\star+M, k^\star+M}|^2 + |(\bm{B}_l)_{k^\star+M,k^\star}|^2 + |(\bm{B}_l)_{k^\star,k^\star+M}|^2\right]\nonumber\\
            &\pm2\sum_{i=1}^M(\bm{m}_\mathcal{V'})_{i}\left[|(\bm{B}_l)_{i,k^\star}|^2 + |(\bm{B}_l)_{i+M, k^\star+M}|^2 + |(\bm{B}_l)_{i+M,k^\star}|^2 + |(\bm{B}_l)_{i,k^\star+M}|^2\right].
        \end{align}
        Thus, by storing the values of $T(l,\mathcal{V}')$ and $F^2(l,\mathcal{V}')$, we can update the traces in constant time, and the Frobenius norm in $\mathcal{O}(KM)$ time. This reduces the time of computation of $f_N(\bar{\bm{A}}_\mathcal{V}, \bm{0})$ from $\mathcal{O}(NKM^2 + N^2\log N)$ to $\mathcal{O}(NKM + N^2\log N)$. Thus, the time of computation of $\Pr(\bm{n}|\bar{\bm{A}},\bm{0})$ reduces to $\mathcal{O}\left((NKM + N^2\log N)\prod_{k=1}^M(n_k+1)\right)$.

\end{document}